%% file: 20260123.tex
\documentclass[11pt,reqno]{amsart}
\usepackage{amsmath,amssymb}
\usepackage{amsthm} 
\usepackage{amscd}

\usepackage[pdftex]{graphicx} 
\usepackage[all]{xy}

\makeatletter
\@addtoreset{equation}{section}
\makeatother
\usepackage{enumerate}

\input{define.tex}

\newtheorem{theorem}{Theorem}[section]
\newtheorem{definition}[theorem]{Definition}

\newtheorem{proposition}[theorem]{Proposition}
\newtheorem{corollary}[theorem]{Corollary}
\newtheorem{remark}[theorem]{Remark}
\newtheorem{lemma}[theorem]{Lemma}
\newtheorem{assumption}[theorem]{Assumption}

\newtheorem{condition}[theorem]{Condition}

\def\book#1{\rm{#1}, }
\def\paper#1{\textit{#1}, }
\def\jour#1{\rm{#1}, }
\def\yr#1{({\rm{#1}) }}
\def\vol#1{\textbf{#1}}
\def\pages#1{\rm{#1}}

\def\publaddr#1{\rm{#1}, }
\def\publ#1{\rm{#1}, }
\def\by#1{{\rm{#1}, }}

\begin{document}


\title[Closed real hyperelliptic plane curves of FGMKdV equation of genus $g$]{Closed real plane curves of hyperelliptic solutions of focusing gauged modified KdV equation of genus $g$}

\author{Shigeki Matsutani}
%

\date{\today}

\begin{abstract}
The real part of the focusing modified Korteweg-de Vries (MKdV) equation defined over the complex field $\CC$ is reduced to the focusing gauged MKdV (FGMKdV) equation.
In this paper, we construct the real hyperelliptic solutions of FGMKdV equation in terms of data of the hyperelliptic curves of genus $g$ and demonstrate the closed hyperelliptic plane curves of genus $g=5$ whose curvature obeys the FGMKdV equation by extending the previous results of genus three (Matsutani, {\it{J. Geom. Phys}} \vol{215} (2025) 105540).
These are a generalization of Euler's elasticae. 
\end{abstract}




\maketitle
{\bf{Keywords:}}
{modified KdV equation, gauged modified KdV equation, real hyperelliptic solutions, hyperelliptic curves, focusing MKdV equation}



\section{Introduction}\label{sec:1}

In 1744, Euler drew typical real plane curves $Z: [0,1] \to \CC$ whose curvature $k$ obeys the focusing static modified Korteweg-de Vries (FSMKdV) equation \cite{Euler44,PG, M25E},
\begin{equation}
a\partial_{s}k
           +\frac{3}{2}k^2 \partial_s k
+\partial_{s}^3 k=0, \quad
a\partial_{s}\phi
           +\frac{1}{2}(\partial_s \phi)^3
+\partial_{s}^3 \phi=0,
\label{4eq:SMKdV_k}
\end{equation} 
where $\partial_s := \partial/\partial s$, $s$ is the arclength and $\phi := \log \partial_s Z/\ii$ is the tangential angle of the curve i.e., the curvature $k = \partial_s \phi$.
The symbol $\ii$ represents the imaginary unit.
Recently, it was discovered that Euler's derivation (\ref{4eq:SMKdV_k}) is profound and timeless \cite{Euler44,M25E}.
We note that $|\partial_s Z|=1$.
Euler's elasticae, as solutions of (\ref{4eq:SMKdV_k}), are well-written in terms of the theory of the elliptic function and $Z$ is expressed by the Weierstrass zeta function as in \cite{Mat10}.

Euler classified the shapes of $Z$ as the solutions of the minimal problem of the energy $E[Z]:=$ $\displaystyle{\int k^2 ds}$, so-called the Euler-Bernoulli energy functional under the isometric (non-stretching) condition.
They are the ground states of $E[Z]$.
Euler also concluded that there are only two cases, the circle and the figure-eight as closed elasticae, as in Figure \ref{fg:shape1744_24} (a) for the solution of (\ref{4eq:SMKdV_k}) \cite{Euler44}.

This paper provides a generalization of Euler's figure-eight of elliptic curve of genus one to that of hyperelliptic curve of genus $g$ by extending the previous paper on genus three \cite{Ma24d} as illustrated in Figure \ref{fg:shape1744_24} (b)--(d).
We demonstrate the case of the genus five as the focusing gauged modified Korteweg-de Vries (FGMKdV) flow as in Figures \ref{fg:shapeA} --  \ref{fg:shapeD} in Sections 5 and 6.

\bigskip

The modified Korteweg-de Vries (MKdV) equation is given by
\begin{equation}
\partial_t q \pm 6q^2 \partial_s q +\partial_s^3 q =0,
\label{eq:MKdV1}
\end{equation}
where $\partial_t := \partial/\partial t$ and $\partial_s := \partial/\partial s$ for the real axes $t$ and $s$.
The \lq\lq$+$\rq\rq case in (\ref{eq:MKdV1}) is referred to as the focusing MKdV (FMKdV) equation and the \lq\lq$-$\rq\rq case is referred to as the defocusing MKdV equation due to \cite{ZakharovShabat}. 
The FMKdV equation appeared as an integrable system in geometry.
By investigating an integrable system, Konno, Ichikawa and Wadati \cite{KIW, KIW2}, and Ishimori \cite{Ishimori, Ishimori2} found plane curves that a half of their curvature $k/2$, ($k=\partial_s \phi$) obeys the FMKdV equation (\ref{eq:MKdV1}), i.e., 
\begin{equation}
\partial_t \phi
           +\frac{1}{2}(\partial_s \phi)^3
+\partial_{s}^3 \phi=0,
\label{4eq:MKdV1phi}
\end{equation}
where $\phi$ is the tangential angle of the plane curve, which is known as the loop soliton.
In this paper, we also call (\ref{4eq:MKdV1phi}) the focusing modified KdV (FMKdV) equation, although we referred to (\ref{4eq:MKdV1phi}) as the focusing modified pre-KdV equation in \cite{Ma24b}.
They showed that (\ref{4eq:MKdV1phi}) can be regarded as a generalization of Euler's elastica \cite{KIW, KIW2, Ishimori, Ishimori2}.

Independently, Goldstein and Petrich showed that the isometric deformation of a real curve on a plane is connected with the recursion relations of the FMKdV hierarchy \cite{GoldsteinPetrich1}.
After \cite{GoldsteinPetrich1}, the MKdV flow has been studied by several authors \cite{CKPP, Langer, LangerPerline, MP16, P}, which provides quite interesting geometrical properties.
Following them, Previato and the author of this paper found that the Goldstein-Petrich scheme and the FMKdV equation \cite{GoldsteinPetrich1} play an essential role in the isometric flows of the plane curves and in the statistical mechanics of the elasticae \cite{P, Mat97, MP16}.
The orbits in the equi-energy sets (of excited states) $\{Z \ |\ E[Z]=E_0\}$ of the elasticae rather than the ground state are described by the solutions of (\ref{4eq:MKdV1phi}) \cite{Mat97, MP16}.

As in the study of Vologodskii and Cozzarelli \cite{VC}, the excited states of $E[Z]$ of elasticae play the crucial roles in the shapes of the supercoiled DNA due to the thermal effect.
The paper \cite{M24a} showed that finding the hyperelliptic solutions of the FMKdV equation of genus three based on the previous results \cite{Ma24b} is critical to reproduce the shapes of the supercoiled DNA in observed in laboratories.
It provides a fascinating relationship between modern mathematics and life sciences.
Thus, it is crucial to find the real hyperelliptic solutions of the FMKdV equation of the higher genus (see Figure \ref{fg:shape_vsDNA}).

\bigskip

For a hyperelliptic curve $X$ of genus $g$ given by $y^2 -(-1)^g (x-b_0)(x-b_1)\cdots(x-b_{2g})=0$ for $b_i \in \CC$, due to Baker \cite{Baker97, Baker03, BuMi2, BEL97b, M25}, we find the hyperelliptic solutions of the KdV equation as $\wp_{gg}(u)=x_1+ \cdots +x_g$ for $((x_1,y_1), \ldots, (x_g, y_g)) \in S^g X$ ($g$-th symmetric product of $X$) as a function of $u \in \CC^g$ through the Abel-Jacobi map $v: S^g X \to J_X$ for the Jacobi variety $J_X$, $u=v((x_1,y_1), \ldots, (x_g, y_g))$.
With the help of the Miura map, it is not difficult to find the hyperelliptic solutions of the FMKdV equation over $\CC$ \cite{Mat02c}, i.e.,
\begin{equation}
\partial_{u_{g-1}}\psi
-\frac{1}{2}(\lambda_{2g}+3b_0) \partial_{u_g}\psi
+\frac{1}{8}
\left(
\partial_{u_g} \psi\right)^3
+\frac{1}{4}\partial_{u_g}^3 \psi = 0,
\label{1eq:fprMKdV}
\end{equation}
where  $\psi:=\log( (b_0-x_1)\cdots(b_0-x_g))/\ii$.

Let $a$-th component $u_a$ of $u \in \CC^g$ be decomposed to its real and imaginary parts, $u_a = u_{a\,\rr} + \ii u_{a\, \ri}$, $\partial_{u_a}=\frac{1}{2}(\partial_{u_a\,\rr}-\ii \partial_{u_a\,\ri})$, $(a=1, \ldots, g)$, and let $\psi=\psi_\rr + \ii \psi_\ri$ of the real valued functions  $\psi_\rr$ and $\psi_\ri$.
The Cauchy-Riemann relations mean $\partial_{u_a\,\rr} \psi_{\rr}=\partial_{u_a\,\ri} \psi_{\ri}$ and $\partial_{u_a\,\rr} \psi_{\ri}=-\partial_{u_a\,\ri} \psi_{\rr}$, and thus $\partial_{u_a}\psi = \partial_{u_a\,\rr} \psi_{\rr}-\ii \partial_{u_a\,\ri} \psi_{\rr}$ or $\partial_{u_a}\psi = \partial_{u_a\,\rr} (\psi_{\rr}+\ii \psi_{\ri})$ $(a=1,\ldots,g)$. 
Since (\ref{1eq:fprMKdV}) contains the cubic term $(\partial_{u_g} \psi)^3=(\partial_{u_g\,\rr} \psi)^3$, it generates the term $-3(\partial_{u_g\, \rr} \psi_\ri)^2  \partial_{u_g\, \rr} \psi_\rr$, which behaves like a gauge potential.
Thus we encounter coupled differential relations from the FMKdV equation over $\CC$ (\ref{1eq:fprMKdV}) as
\begin{equation}
\begin{split}
&(\partial_{u_{g-1}\, \rr}-
A^+(u)\partial_{u_g\, \rr})\psi_\rr
           +\frac{1}{8}
\left(\partial_{u_g\, \rr} \psi_\rr\right)^3
+\frac{1}{4}\partial_{u_g\, \rr}^3 \psi_\rr=0,\\
-&(\partial_{u_{g-1}\, \ri}-
A^-(u)\partial_{u_g\, \ri})\psi_\rr
           +\frac{1}{8}
\left(\partial_{u_g\, \ri} \psi_\rr\right)^3
+\frac{1}{4}\partial_{u_g\, \ri}^3 \psi_\rr=0,
\label{1eq:gaugedMKdV2}
\end{split}
\end{equation}
where $A^+(u)=(\lambda_{2g}+3b_0+\frac{3}{4}(\partial_{u_{g}\, \rr}\psi_\ri)^2)/2$ and $A^-(u)=(\lambda_{2g}+3b_0-\frac{3}{4}(\partial_{u_{g}\, \rr}\psi_\rr)^2)/2$.
We refer to (\ref{1eq:gaugedMKdV2}) as the focusing gauged MKdV (FGMKdV) equations. 
Here we take both cases $(u_{g\,\rr},u_{g-1\,\rr}) \in \RR^2$ and $\ii(u_{g\,\ri},u_{g-1\,\ri})$ in $\ii\RR^2$ in (\ref{1eq:gaugedMKdV2}).

\cite{MP22} gave that to obtain the real solution of the FMKdV equation (\ref{4eq:MKdV1phi}) is to find the situation that the following conditions are satisfied for the solutions of (\ref{1eq:gaugedMKdV2}):
\begin{enumerate}

\item[CI] $\prod_{i=1}^g |x_i - b_0|=$ a constant $(> 0)$,

\item[CII] $d u_{g\,\ri}=d u_{g-1\, \ri}=0$ or $d u_{g\,\rr}=d u_{g-1\, \rr}=0$, and

\item[CIII] $A^\pm(u)$ is a real constant:
if $A^\pm(u)=$ constant, (\ref{1eq:gaugedMKdV2}) is reduced to (\ref{4eq:MKdV1phi}), i.e., $\psi_\rr=\phi$. 
\end{enumerate}

\bigskip
\begin{figure}
\begin{center}

\includegraphics[height=0.3\hsize, angle=90]{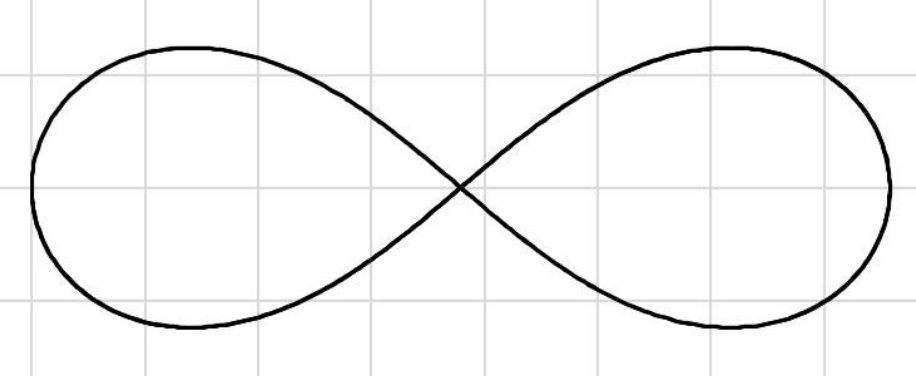}
\hskip 0.05\hsize
\includegraphics[width=0.5\hsize]{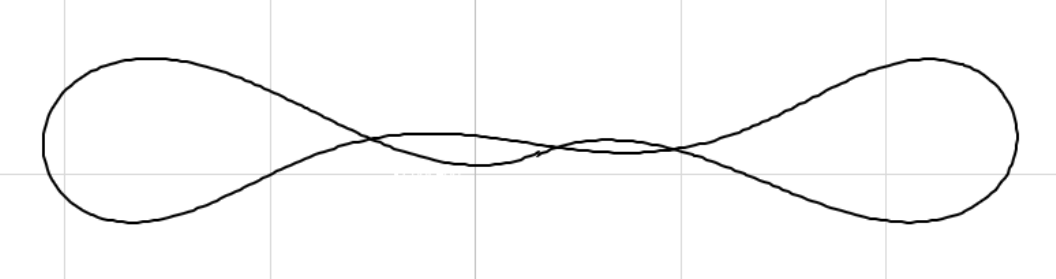}

(a)
\hskip 0.45\hsize
(b)
\hskip 0.1\hsize\ 

\smallskip
\smallskip
\smallskip

\includegraphics[width=0.35\hsize,]{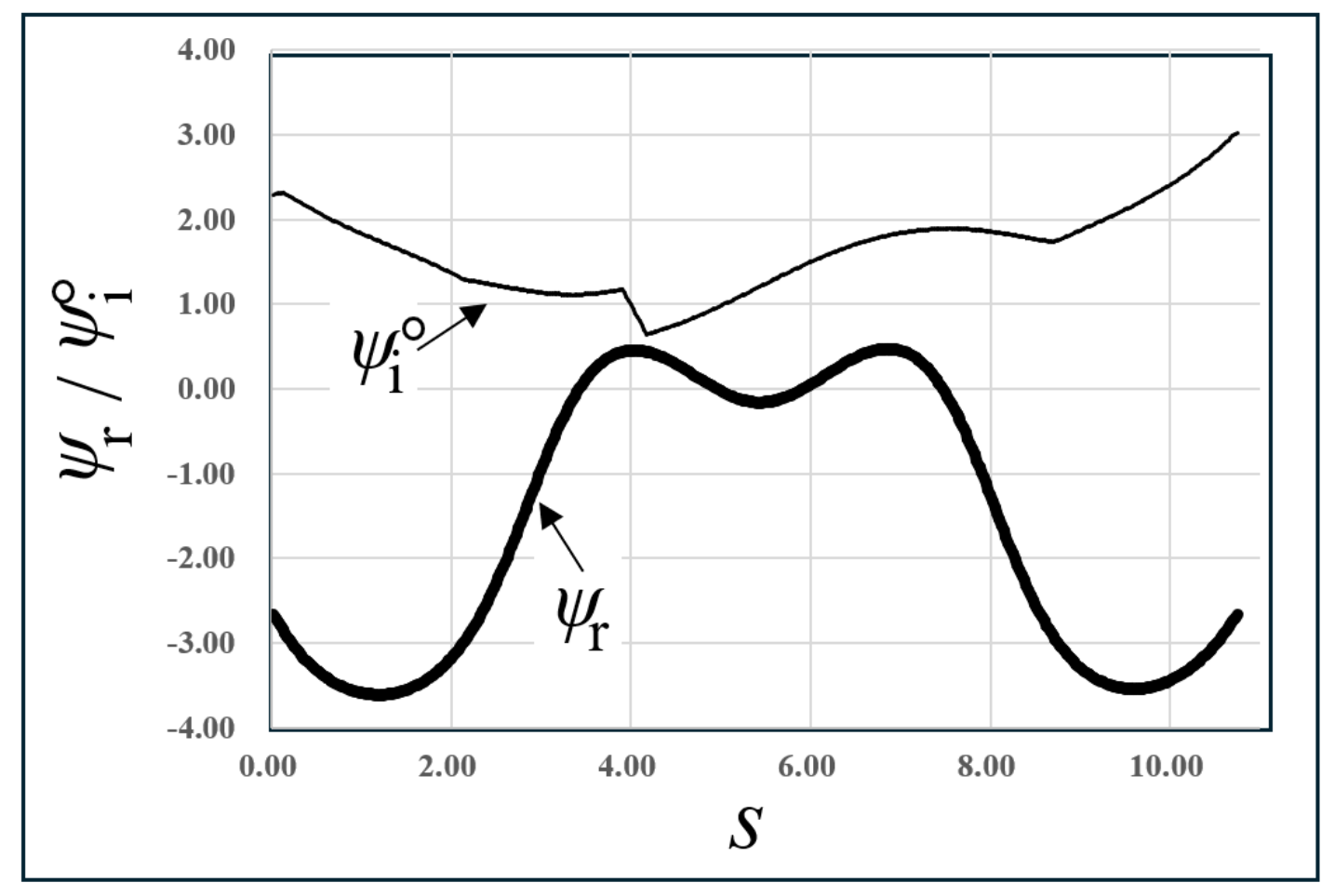}
\hskip 0.05\hsize
\includegraphics[width=0.35\hsize,]{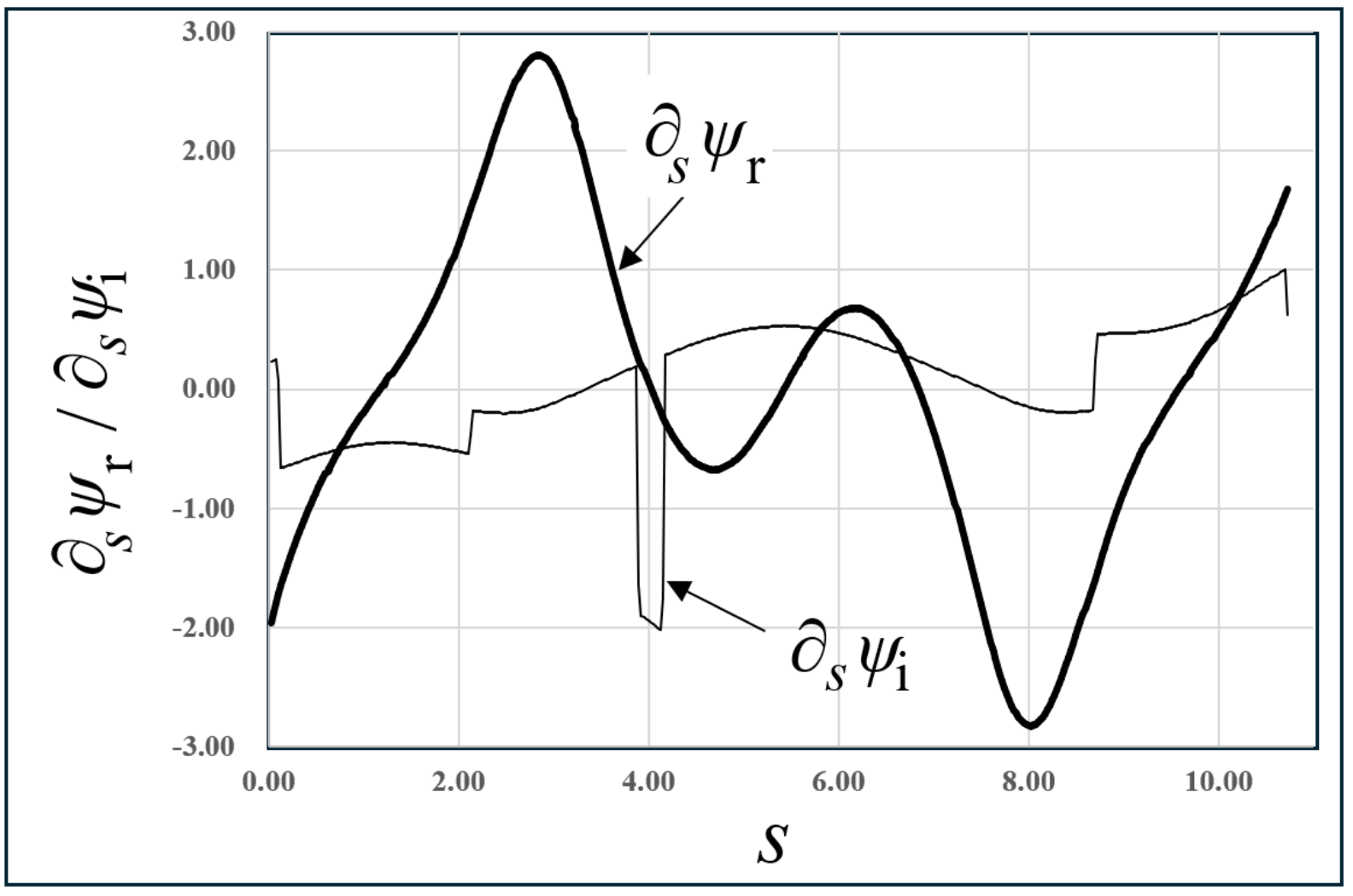}

(c)
\hskip 0.35\hsize
(d)

\smallskip
\smallskip
\smallskip

\end{center}

\caption{
Euler's figure-eight and previous result in \cite{Ma24d}: (a): Euler's figure-eight \cite{Euler44}, (b) $Z(s)$, (c): $\psi_\rr$ and $\psi_\ri^\circ$. (d): $\partial_s\psi_\rr$ and $\partial_s\psi_\ri$.
}\label{fg:shape1744_24}
\end{figure}

\bigskip

However, it is quite difficult to find the real plane $\{(u_g, u_{g-1})\}$ in the Jacobi variety $J_X$ which corresponds to the preimage $((x_1, y_1), \ldots, (x_g, y_g)) \in S^g X$ of $v$ with the unit circle valued $\ee^{\ii \psi}=(b_0-x_1)\cdots(b_0-x_g)\in \mathrm{U}(1):=\{\zeta\in \CC \ |\ |\zeta|=1\}$.
\bigskip

In the previous papers \cite{Ma24d}, we showed the real hyperelliptic solutions of the FGMKdV equation for the case of the genus three only by considering the conditions CI and CII.

The aim of this paper is to extend the results for genus three in \cite{Ma24d} to general genus g in a straightforward way: we will also consider only the conditions CI and CII following \cite{Ma24d}.
We will show that the extension is naturally achieved by investigating the angle expressions of the hyperelliptic integrals in \cite{Ma24b} as in Section 3, and a modification of the elementary symmetric polynomials defined in Definition \ref{def:varepsilon}: we refer to the modification of the elementary symmetric polynomials as shifted elementary symmetric polynomials in this paper.
Since the computations of the polynomials are slightly complicated, we first demonstrates them of the genus five case in Subsection \ref{ssc:g=5} and prepare Appendix, which provides the detail explanations, to show the general $g$ case in Subsection \ref{ssc:g=g}.

Further, we demonstrate the closed plane curves of genus five numerically in Section 5.
There we display a generalization of Euler's result and the results from the previous paper \cite{Ma24d}.

As in \cite{Mat02c}, the symmetric polynomials determine a fundamental property of the Jacobi matrices between the cotangent spaces $T^* S^g X$ and $T^* J_X$ of $S^g X$ and $J_X$, respectively.
Weierstrass and Baker essentially studied the correspondence between $T^* S^g X$ and $T^* J_X$ to obtain the differential identities on hyperelliptic curves of genus $g$, which are related to the sine-Gordon equation and the KdV hierarchy.
They implicitly and explicitly used the elementary symmetric polynomials.
(Recently, such a picture is sophisticated from a modern point of view and extended by Buchstaber and Mikhailov \cite{BuMi1, BuMi2}.)

In this paper, we apply their method to the configurations satisfying the condition CI or  $(b_0-x_1)\cdots(b_0-x_g)\in \mathrm{U}(1)$ to obtain the real FGMKdV equation as a real differential identities of genus $g$, although we implicitly employed this approach in \cite{Ma24b,Ma24d} for the genus three case.
The process of extending from genus three to genus $g$ is essentially straightforward, although the treatment of the shifted elementary symmetric polynomials is slightly more complicated in the extension (thus, we present the polynomials of genus five in Subsection \ref{ssc:g=5} and the general cases in the Appendix).

\bigskip

The content is following:
Section 2 reviews the hyperelliptic solutions of the FMKdV equation over $\CC$ of genus $g$ in Theorem \ref{4th:MKdVloop} following \cite{Mat02c,Mat02b,MP15,MP22} for the hyperelliptic curve $X$ of genus $g$.

Since the real solutions are related to the covering $\hX$ of $X$ and the angle expression, Section 3 is devoted to the double covering $\hX$ of $X$ and the angle expression of the hyperelliptic curves of genus $g$.
To obtain the real hyperelliptic solutions of the FGMKdV equation, we make use of Assumption \ref{Asmp} and the assumptions in Definition \ref{Asmp2}, after which the angle expression becomes applicable.
Under the assumptions, in order to express the map from $v_* : T_* S^g X \to T_* J_X$ and its inverse, we explicitly evaluate the maps of genus $g=5$ in Subsection \ref{ssc:g=5}, and we extend them to general $g$ in Subsection \ref{ssc:g=g}

Section 4 provides local and global properties of the solutions of the FGMKdV equation (\ref{1eq:gaugedMKdV2}).
In Theorem \ref{th:4.2}, we present the local properties.
Based on these, we also show the global behavior of the hyperelliptic solutions of genus $g$ of the FGMKdV equation in Theorems \ref{4th:reality_gga} and \ref{pr:solgMKdV}.
Theorems \ref{th:4.2}, \ref{4th:reality_gga} and \ref{pr:solgMKdV} are our main theorems, showing that we have the real solutions of the FGMKdV equation of genus $g>2$.

Section 5 uses Theorem \ref{pr:solgMKdV} to draw plane curves as snapshots, and  displays the numerical results of the closed plane curves whose tangential angle obeys the FGMKdV equation (\ref{1eq:gaugedMKdV2}).

Section 6 gives the conclusion of this paper.

In this paper, Lemmas  \ref{lm4.1} and \ref{4lm:dudphig} are pivotal. However, their proofs are rather intricate. To elucidate these proofs, we have prepared  Appendix that provides detailed explanations of shifted elementary symmetric polynomials $\varepsilon$, as defined in Definition \ref{def:varepsilon}.

\section{Hyperelliptic solutions of the FMKdV equation$/\CC$ of genus $g$}
\label{sec:HESGE}

To obtain the relation (\ref{1eq:fprMKdV}),  we handle a hyperelliptic curve $X$ of genus $g\ge 3$ over $\CC$,
\begin{equation}
X=\left\{(x,y) \in \CC^2 \ |
\ y^2 -(-1)^g f(x)=0\right\}
\cup \{\infty\},
\label{4eq:hypC}
\end{equation}
where $f(x):=(x-b_0)(x-b_1)(x-b_2)\cdots(x-b_{2g})$, and $b_i$'s are mutually distinct complex numbers.
Let $\lambda_{2g}=\displaystyle{-\sum_{i=0}^{2g} b_i}$ and $S^k X$ be the $k$-th symmetric product of the curve $X$. 
Further, we introduce the Abelian covering $\tX$ of $X$ by abelianization of the path-space of $X$ divided by the homotopy equivalence, $\kappa_X: \tX \to X$, ($\gamma_{P, \infty} \mapsto P$) \cite{Baker98, Wei54, MP15, M25}.
Here $\gamma_{P, \infty}$ means a path from $\infty$ to $P$.
We also consider an embedding $\iota_X : X \to \tX$ and will fix it.
$S^k \tX$ also means the $k$-th symmetric product of the space $\tX$. 
The Abelian integral $\tv:=\displaystyle{\begin{pmatrix} v_1\\ \vdots \\ v_g\end{pmatrix}} : S^g  \tX \to \CC^g$ is defined by
\begin{equation}
\tv_i(\gamma_1, \ldots, \gamma_g)=\sum_{j=1}^g
 \tv_i(\gamma_j), \quad
\tv_i(\gamma_{(x,y), \infty}) = \int_{\gamma_{(x,y), \infty}} \nuI{i},\quad
\nuI{i} = \frac{x^{i-1}d x}{2y}.
\label{4eq:firstdiff}
\end{equation}
Then we have the Jacobian $J_X$, $\kappa_J: \CC^g \to J_X=\CC^g/\Gamma_X$, where $\Gamma_X$ is the lattice generated by the period matrix for the standard homology basis of $X$.
Due to the Abel-Jacobi theorem \cite{FarkasKra}, we also have the bi-rational map $v$ from $S^g X$ to $J_X$ by letting $v:=\tv$ modulo $\Gamma_X$.
We refer to $v$ as the Abel-Jacobi map.

\cite{Mat02c} shows the hyperelliptic solutions of the MKdV equation over $\CC$, derived by a natural extension of the investigations of Weierstrass \cite{Wei54} and Baker \cite{Baker03}.

\begin{definition}\label{4df:KdV_def2}
Let $\{(x_i, y_i)\}_{i=1, \ldots, g} \in S^g X$.
\begin{enumerate}
\item 
We define the polynomials associated with $F(x)=(x-x_1) \cdots (x-x_g)$ by
\begin{equation}
\pi_i(x) := \frac{F(x)}{x-x_i}=\chi_{i,g-1}x^{g-1} +\chi_{i,g-2} x^{g-2}
            +\cdots+\chi_{i,1}x+\chi_{i,0}.
\label{4eq:KdV_def2.1}
\end{equation}

\item We define $g\times g$ matrices as follows.
{\small{
$$
 \cX\!:=\! 
\left[
\begin{array}{cccc}
     \chi_{1,0} & \chi_{1,1} & \cdots & \chi_{1,g-1}  \\
      \chi_{2,0} & \chi_{2,1} & \cdots & \chi_{2,g-1}  \\
   \vdots & \vdots & \ddots & \vdots  \\
    \chi_{g,0} & \chi_{g,1} & \cdots & \chi_{g,g-1}
     \end{array}\right]
,\ 
\cY\!:=\! 
\left[\begin{array}{cccc}
     y_1 & \ & \ & \  \\
      \ & y_2& \ & \   \\
      \ & \ & \ddots   & \   \\
      \ & \ & \ & y_g  \end{array}
\right],
$$
$$
	\cF'\!:=\!
\left[
\begin{array}{cccc}F'(x_1)& &  &   \\
       & F'(x_2)&  &    \\
       &  &\ddots&    \\
       &  &  &F'(x_{g})\end{array}
\right],\ 
\cU\!:=\! 
\left[
\begin{array}{cccc} 1 & 1 & \cdots & 1 \\
                   x_1 & x_2 & \cdots & x_g \\
                   x_1^2 & x_2^2 & \cdots & x_g^2 \\
                    \vdots& \vdots & \ddots & \vdots \\
                   x_1^{g-1} & x_2^{g-1} & \cdots & x_g^{g-1}
                 \end{array}\right],
$$
}}
where $F'(x):=d F(x)/d x$.
\end{enumerate}

\end{definition}

Using these and the properties of the Vandermonde matrix, we obtain the following \cite{Baker03, Mat02c}:

\begin{lemma}\label{4lm:KdV1}
Let $u = \tv(\iota_X((x_1, y_1), \ldots, (x_g, y_g)))$.

\begin{enumerate}
\item 
$$
      \left[\begin{array}{c} d u_1\\
                 d u_2\\
                 \vdots\\
                 d {u_g}
         \end{array}\right]
   =\frac{1}{2} \cU \cY^{-1}
   \left[\begin{array}{c} d {x_1}\\
                  d {x_2}\\
                 \vdots\\
                 d {x_g}
         \end{array}\right].
$$

\item The inverse matrix of $\cU$ is given as $\cX$, i.e, 
$\cX \cU=\cF^{\prime}$.

\item For $\partial_{u_i}:=\partial/\partial{u_i}$ and
$\partial_{x_i}:=\partial/\partial{x_i}$, we have
$$
      \left[\begin{array}{c} \partial_{u_1}\\
                 \partial_{u_2}\\
                 \vdots\\
                 \partial_{u_g}
         \end{array}\right]
   =2 \trp \cX \cF^{\prime -1}\cY
   \left[\begin{array}{c} \partial_{x_1}\\
                 \partial_{x_2}\\
                 \vdots\\
                 \partial_{x_g}
         \end{array}\right],
\quad
\frac{\partial x_i}{\partial u_r}=
\frac{2y_i}{F'(x_i)} \chi_{i, r-1}, 
$$
\begin{equation}
\frac{\partial}{\partial u_g }=
         \sum_{i=1}^g \frac{2y_i}{F'(x_i)} \frac{\partial}{\partial x_i},
           \quad
	\frac{\partial}{\partial u_{g-1} }=
         \sum_{i=1}^g \frac{2y_i\chi_{i,g-2}}{F'(x_i)}
              \frac{\partial}{\partial x_i}.
\label{4eq:hyp_dxdu}
\end{equation}
\end{enumerate}
\end{lemma}

By applying these differential operators $\displaystyle{
\frac{\partial}{\partial u_g }}$ and $\displaystyle{
\frac{\partial}{\partial u_{g-1} }}$ to $\log F(b_0)$, we obtain the following theorem.

\begin{theorem}\label{4th:MKdVloop} {\textrm{\cite{Mat02b}}}
For $((x_1,y_1),\cdots, (x_g,y_g)) \in S^g X$, the fixed branch point $(b_0, 0)$, and $u:= v( (x_1,y_1),$ $\cdots,(x_g,y_g))$,
$$
\displaystyle{
   \psi(u) :=-\ii \log (b_0-x_1)(b_0-x_2)\cdots(b_0-x_g)
}
$$
satisfies the MKdV equation over $\CC$,
\begin{equation}
(\partial_{u_{g-1}}-\frac{1}{2}
(\lambda_{2g}+3b_0)
          \partial_{u_{g}})\psi
           +\frac{1}{8}
\left(\partial_{u_g} \psi\right)^3
 +\frac{1}{4}\partial_{u_g}^3 \psi=0,
\label{4eq:loopMKdV2}
\end{equation}
where $\partial_{u_i}:= \partial/\partial u_i$ as an differential identity in $S^g X$ and $\CC^g$.
\end{theorem}

We, here, emphasize the difference between the FMKdV equations (\ref{4eq:MKdV1phi}) over $\RR$ and (\ref{4eq:loopMKdV2}) over $\CC$.
In (\ref{4eq:MKdV1phi}), $\phi$ is a real valued function over $\RR^2$ but $\psi$ in (\ref{4eq:loopMKdV2}) is a complex valued function over $\CC^2 \subset \CC^g$.
The difference is crucial since our ultimate goal is to obtain solutions of (\ref{4eq:MKdV1phi}), not (\ref{4eq:loopMKdV2}).
However, the latter is expressed well in terms of the hyperelliptic function theory.

As mentioned in \cite[(11)]{MP22}, we describe the difference.
By introducing real and imaginary parts, $ u_a = u_{a\,\rr} + \ii u_{a\,\ri}$, $(a=1,2,\ldots,g)$, and $ \psi = \psi_{\rr} + \ii \psi_{\ri}$, the real and imaginary part of (\ref{4eq:loopMKdV2}) are reduced to the FGMKdV equations with the gauge fields $A^+(u)=(\lambda_{2g}+3b_0+\frac{3}{4}(\partial_{u_{g}\, \ri}\psi_\rr)^2)/2$,
$A^-(u)=(\lambda_{2g}+3b_0-\frac{3}{4}(\partial_{u_{g}\, \rr}\psi_\rr)^2)/2$,
\begin{eqnarray}
(\partial_{u_{g-1}\, \rr}-
A^+(u)\partial_{u_{g}\, \rr})\psi_\rr
           +\frac{1}{8}
\left(\partial_{u_g\, \rr} \psi_\rr\right)^3
+\frac{1}{4}\partial_{u_g\, \rr}^3 \psi_\rr&=0, \nonumber\\
-(\partial_{u_{g-1}\, \ri}-
A^-(u)\partial_{u_{g}\, \ri})\psi_\rr
          +\frac{1}{8}
\left(\partial_{u_g\, \ri} \psi_\rr\right)^3
+\frac{1}{4}\partial_{u_g\, \ri}^3 \psi_\rr&=0
\label{4eq:gaugedMKdV2}
\end{eqnarray}
by the Cauchy-Riemann relations $\partial_{u_a\,\rr} \psi_{\rr}=\partial_{u_a\,\ri} \psi_{\ri}$ and $\partial_{u_a\,\rr} \psi_{\ri}=-\partial_{u_a\,\ri} \psi_{\rr}$ as mentioned in \cite[(11)]{MP22} and (\ref{1eq:gaugedMKdV2}).
We note that $\partial_{u_a} \psi = \partial_{u_a\, \rr} (\psi_{\rr} + \ii \psi_{\ri})$ because $\partial_{u_a} =(\partial_{u_a\, \rr}-\ii\partial_{u_a\, \ri})/2$, and thus $(\partial_{u_g} \psi)^3$ contains the term $-3(\partial_{u_g\,\rr}\psi_{\ri})^2 \partial_{u_g\,\rr}\psi_{\rr}$.
In the derivation of the second relation in (\ref{4eq:gaugedMKdV2}), we used the relations $\partial_{u_a} \psi =\partial_{u_a, \ri} \psi_\ri-\ii\partial_{u_a, \ri} \psi_\rr$ and
$\partial_{u_a}^3 \psi = -\partial_{u_a, \ri}^3 \psi_\ri+\ii\partial_{u_a, \ri}^3 \psi_\rr$.
We also note that the latter one has an alternative expression as defocusing gauged MKdV equation,
\begin{equation}
(\partial_{u_{g-1}\, \rr}-
A^-(u)\partial_{u_{g}\, \rr})\psi_\ri
          -\frac{1}{8}
\left(\partial_{u_g\, \rr} \psi_\ri\right)^3
+\frac{1}{4}\partial_{u_g\, \rr}^3 \psi_\ri=0,
\label{4eq:gaugedMKdV2II}
\end{equation}
even though we will not touch this expression.

\bigskip
A solution of (\ref{4eq:MKdV1phi}) in terms of the data in Theorem \ref{4th:MKdVloop} must satisfy  the following conditions \cite{MP22}:

\begin{condition}\label{cnd}
{\rm{
\begin{enumerate}

\item[CI] $\prod_{i=1}^g |x_i - b_0|=$ a constant $(> 0)$ in Theorem \ref{4th:MKdVloop},

\item[CII] $d u_{g\,\ri}=d u_{g-1\, \ri}=0$ or $d u_{g\, \rr}=d u_{g-1\, \rr}=0$ in Theorem \ref{4th:MKdVloop}, and

\item[CIII] $A(u)$ is a real constant:
if $A(u)=$ constant (or $\partial_{u_{g}\, \rr}\psi_\ri=$ constant), (\ref{4eq:gaugedMKdV2}) is reduced to (\ref{4eq:MKdV1phi}), i.e., $\psi_\rr=\phi$. 
\end{enumerate}
}}
\end{condition}
It is obvious that if we have the solutions $\psi_\rr$ of (\ref{4eq:gaugedMKdV2}) satisfying the conditions CI--CIII, $\partial_{u_g, \rr}\psi_{\rr}/2$ obeys the FMKdV equation (\ref{eq:MKdV1}).

However, in this paper we focus on the conditions CI and CII and the real hyperelliptic solutions of the FGMKdV equation (\ref{4eq:gaugedMKdV2}) of genus $g$ instead of (\ref{4eq:MKdV1phi}).

\begin{figure}
\begin{center}

\includegraphics[width=0.6\hsize]{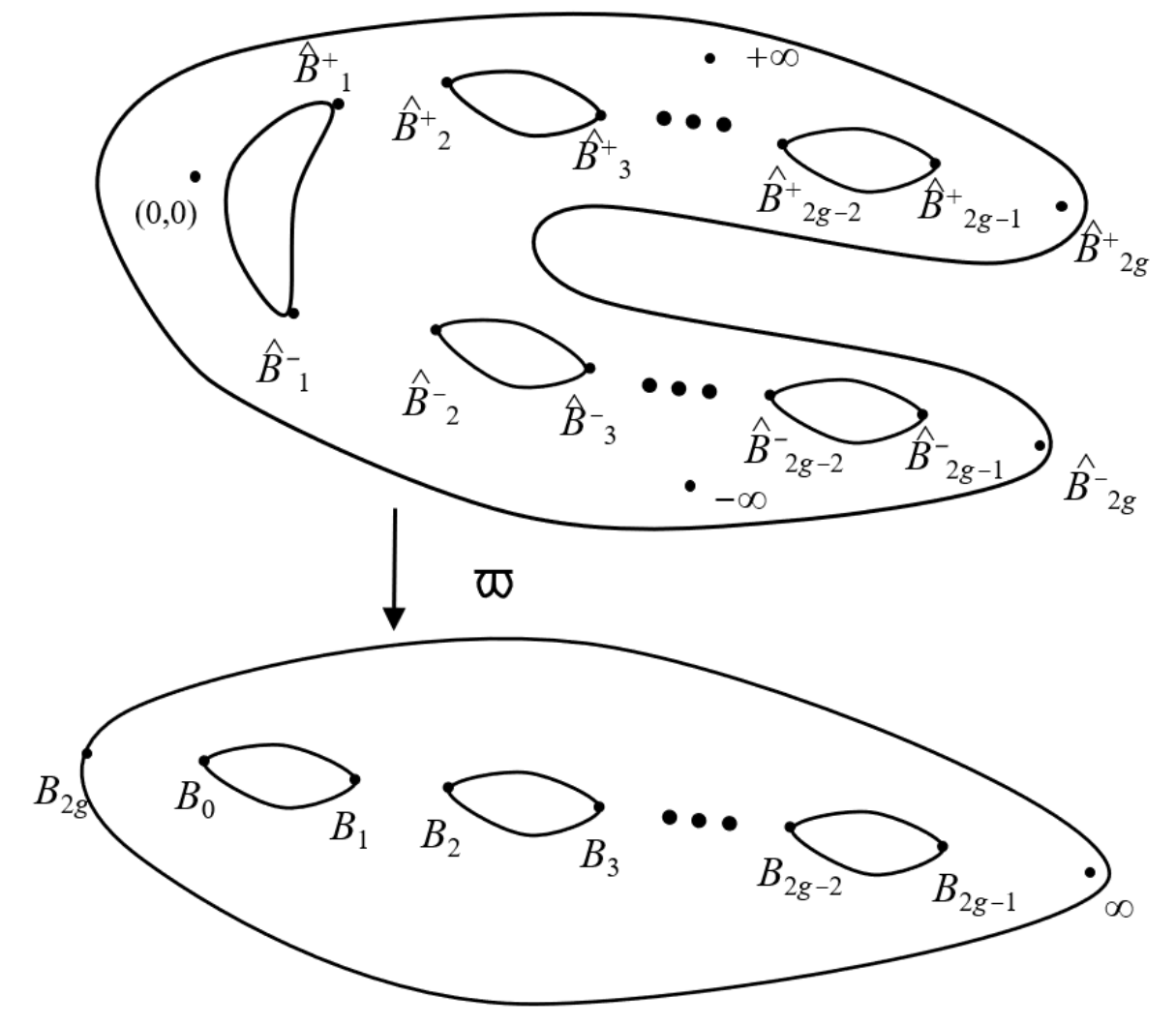}

\end{center}

\caption{The double covering $\varpi_\hX: \hX \to X$,
$\varpi_X: (w, z) \mapsto (w^2+b_0, zw)=(x,y)$.
}\label{fg:Fig00}
\end{figure}

\bigskip

\section{Hyperelliptic curves of genus $g$ in angle expression}\label{sec:g=3}
To find real solutions of the FGMKdV equation (\ref{4eq:gaugedMKdV2}) under Condition \ref{cnd}  CI and CII, we introduce the angle expression \cite{MP22, Ma24b, Ma24d} for $X\setminus\{(b_0,0)\}$ as mentioned in Introduction.

\bigskip
Since the angle expression is connected with a double covering $\hX$ of $X$, we introduce the double covering $\hX$ as in Figure \ref{fg:Fig00}.
We consider the $\al_a(u)$ function $\sqrt{\prod_{i=1}^g(b_a-x_i)}$ for a branch point $B_a:=(b_a, 0) \in X$ ($\varpi_x: X\to \PP^1$) and $((x_i, y_i))_i \in S^g X$.
It means that we consider a line bundle on $X$ and its local section on an open set $U \subset X$.
Here $u$ is given by $u:=\tv_i(\gamma_1,\ldots,\gamma_g)\in \CC^g$ such that $\kappa_X(\gamma_i)=(x_i, y_i)$.
Fix $a=0$.
The square root leads the transformation of $w^2 = (x-b_0)$, i.e., the double covering $\hX$ of the curve $X$, $\varpi_\hX: \hX \to X$, although the precise arguments are left to the Appendix in \cite{MP15}.
Since $\tX$ is also an Abelian covering of $\hX$, we have a natural commutative diagram,
\begin{equation}
\xymatrix{ 
 \tX \ar[dr]^{\kappa_X}\ar[r]^-{\kappa_\hX}& \hX \ar[d]^-{\varpi_\hX} \\
  & X.
}\label{2eq:al_hyp_kappas}
\end{equation}
The $\al_a(u)$ function is a generalization of the Jacobi $\sn, \cn, \dn$ functions because the Jacobi function consists of $\sqrt{x-e_i}$, ($i=1,2,3$) of genus one for a curve $y^2 = \prod_{i=1}^3(x-e_i)$.

The curve $\hX$ is given by $f_\hX(w,z) = z^2 -(-)^g (w^2-e_1)\cdots(w^2 - e_{2g})=0$, where $z:=y/w$ (due to normalization), and $e_j := b_j- b_0$, $j=1, \ldots, 2g$.
Its affine ring is $R_\hX:=\CC[w,z]/(f_\hX(w, z))$, and the ring of its $g$-th symmetric polynomials is denoted by $S^g R_\hX$.

Since the genus of $\hX$ is $2g-1$, we have $2g-1$ holomorphic one-forms,
$$
\hnu_j:= \frac{w^{j} d w}{z}, \quad (j=0, 1, 2, 3, \ldots, 2g-2),
$$
and the Jacobi variety, $J_{\hX}$ of $\hX$ is given by the complex torus $J_{\hX}=\CC^{2g-1} /\Gamma_{\hX}$ for the lattice $\Gamma_{\hX}$ generated by the period matrix.
As in \cite[Appendix, Proposition 11.9]{MP15}, we have the correspondence $\varpi_X^*\nuI{i}=\hnu_{2i-2}$, $(i=1, \ldots, g)$ and thus the Jacobian $J_{\hX}$  contains a subvariety $\hJ_X\subset J_{\hX}$ which is a double covering of the Jacobian $J_X$ of $X$, $\hvarpi_J: \hJ_X \to J_X$, and $\hkappa_J : \CC^g \to \hJ_X:= \CC^g/(\Gamma_\hX\cap \CC^g)$.

Since for each branch point $B_j:=(b_j, 0)\in X$ $(j=1, \ldots, 2g)$, we have double branch points $\hB^\pm_j:=(\pm\sqrt{e_j},0) \in \hX$ as illustrated in Figure~\ref{fg:Fig00}.

Similar to the Jacobi elliptic functions, $\hJ_X=\CC^g/\hGamma_X$ is determined by the same Abelian integral $\tv$, and thus we use the same symbol $\tv$ as $\tv : S^g\tX \to \CC^g$ for $\hX$ \cite{MP15}.

\bigskip

We restrict the moduli (rather, parameter) space of the curve $X$ by the following.
We choose coordinates $u = {}^t(u_1, \ldots, u_g)$
 in  $\CC^g$;
$u_i = u_i^{(1)}+u_i^{(2)}+\cdots+u_i^{(g)}$, where $u_i^{(j)}
=v_i((x_j, y_j))$ for $(x_j, y_j) \in X$.
There are the projection $\varpi_x : X \to \PP^1$, $((x,y) \mapsto x)$, and similarly $\hvarpi_x : \hX \to \PP^1$, $((w,z) \mapsto w)$;
$\varpi_{\PP^1}\circ \hvarpi_x=\varpi_x\circ \varpi_\hX$,
where $\varpi_{\PP^1}(w) = w^2 + b_0$.

\begin{assumption}\label{Asmp}
{\rm{
As in Figure~\ref{fg:Fig01}, we let $e_j := b_j - b_0$, $(j=1, 2, \ldots, 2g)$ be on a unit circle $\mathrm{U}(1)=S^1$, $(|e_j|=1)$, whose center is the origin $b_0$ such that $e_{2i-1} = \overline{e_{2i}}$, $(i = 1, 2, \ldots, g)$; there are $\varphi_{\fb,i}^{++}$, $(i = 1, 2, \ldots, g)$ such that 
$$
e_{2i -1} = e^{2 \ii \varphi_{\fb,i}^{++}}, \quad
e_{2i} = e^{-2 \ii \varphi_{\fb,i}^{++}}, \quad (i = 1, 2, \ldots, g).
$$
We let $\varphi_{\fb,i}^{+-}:=-\varphi_{\fb,i}^{++}$, and 
$\varphi_{\fb,i}^{-\pm}:=\pi\mp \varphi_{\fb,i}^{++}$.
}}
\end{assumption}

Further, we rename them and define several sets.
\begin{definition}
{\rm{
\begin{enumerate}
\item $g$ odd case:
$\varphi_{\fb}^{[1\pm]}:=\varphi_{\fb,1}^{+\pm}$,
$$ 
\varphi_{\fb}^{[\ell+]}:=\varphi_{\fb, 2\ell-1}^{++}, \ 
\varphi_{\fb}^{[\ell-]}:=\varphi_{\fb, 2\ell-2}^{++}, \ (\ell =  2, 3, 
\cdots, (g+1)/2),
$$
$$ 
\varphi_{\fb}^{[(\ell+(g-1)/2+)]}:=\varphi_{\fb, 2\ell-2}^{+-}, \ 
\varphi_{\fb}^{[(\ell+(g-1)/2-)]}:=\varphi_{\fb, 2\ell-1}^{+-}, \ 
(\ell =  2, 3, \cdots, (g+1)/2),
$$
\begin{equation}
\cA_X^\circ:= \{ \ee^{\ii\varphi} \ |\ \varphi \in \cA_X^{\varphi}\}, \qquad
\cA_X^{\varphi}:=\bigcup_{\ell=1}^{2g}
[\varphi_{\fb}^{[\ell-]}, \varphi_{\fb}^{[\ell+]}].
\label{eq:cAXodd}
\end{equation}
\item $g$ even case:
$$ 
\varphi_{\fb}^{[\ell+]}:=\varphi_{\fb, 2\ell}^{++}, \ 
\varphi_{\fb}^{[\ell-]}:=\varphi_{\fb, 2\ell-1}^{++}, \ (\ell =  1, 2, 
\cdots, g/2),
$$
$$ 
\varphi_{\fb}^{[(\ell+g/2)+]}:=\varphi_{\fb, 2\ell-1}^{+-}, \ 
\varphi_{\fb}^{[(\ell+g/2)-]}:=\varphi_{\fb, 2\ell}^{+-}, \ (\ell =  1, 2, 
\cdots, g/2),
$$
\begin{equation}
\cA_X^\circ:= \{ \ee^{\ii\varphi} \ |\ \varphi \in \cA_X^{\varphi}\}, \qquad
\cA_X^{\varphi}:=\bigcup_{\ell=1}^{2g}[\varphi_{\fb}^{\ell-}, \varphi_{\fb}^{\ell+}].
\label{eq:cAXeven}
\end{equation}

\item For both cases, we define $\cA_X:=\cA_X^\circ \cup 
[\overline{\ii \cA_X^\circ}]\ii^{-1}$. 
\end{enumerate}
}}
\end{definition}

\begin{figure}
\begin{center}

\includegraphics[width=0.42\hsize, bb= 0 0 641 668]{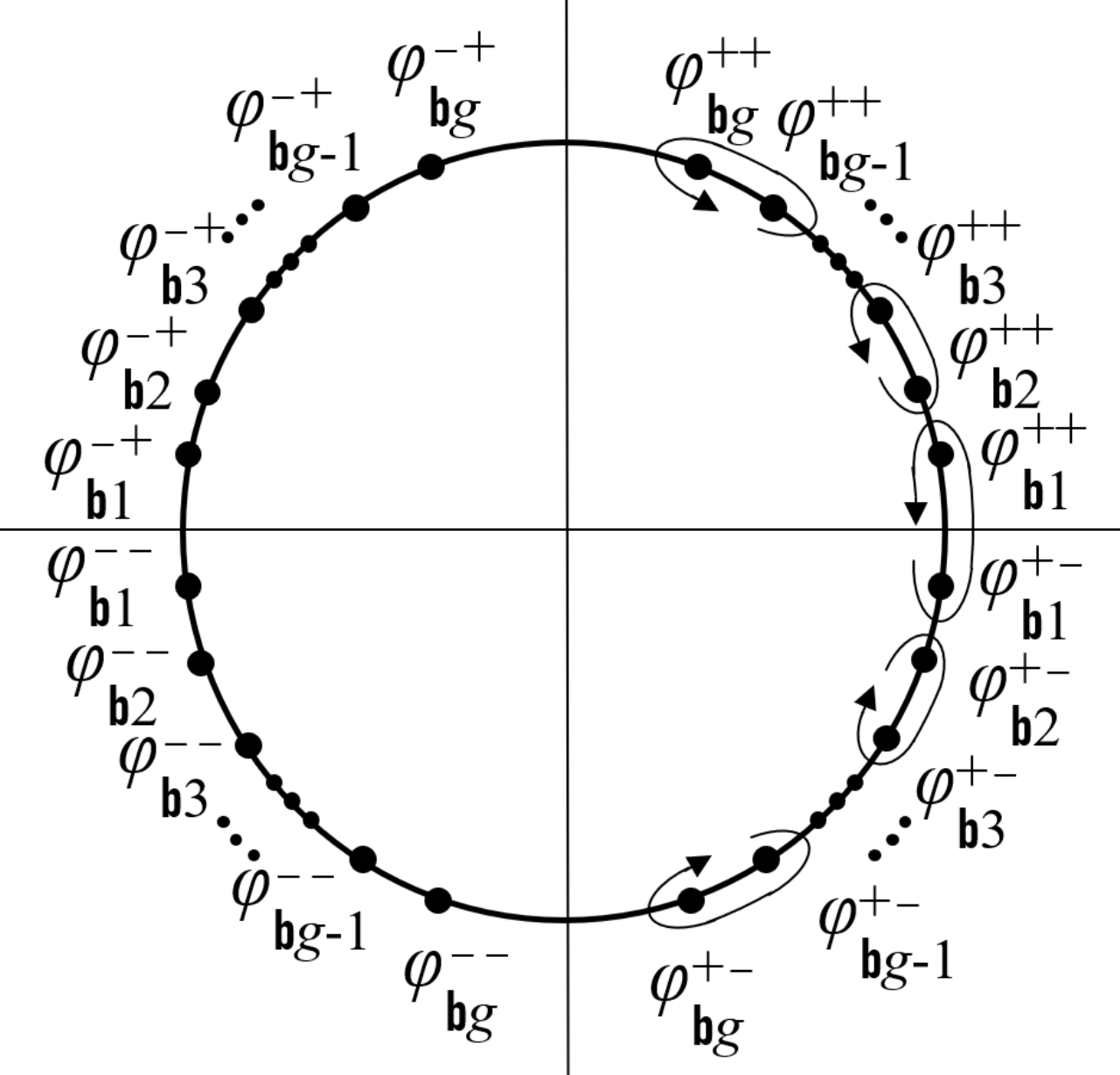}
\includegraphics[width=0.42\hsize, bb= 0 0 641 668]{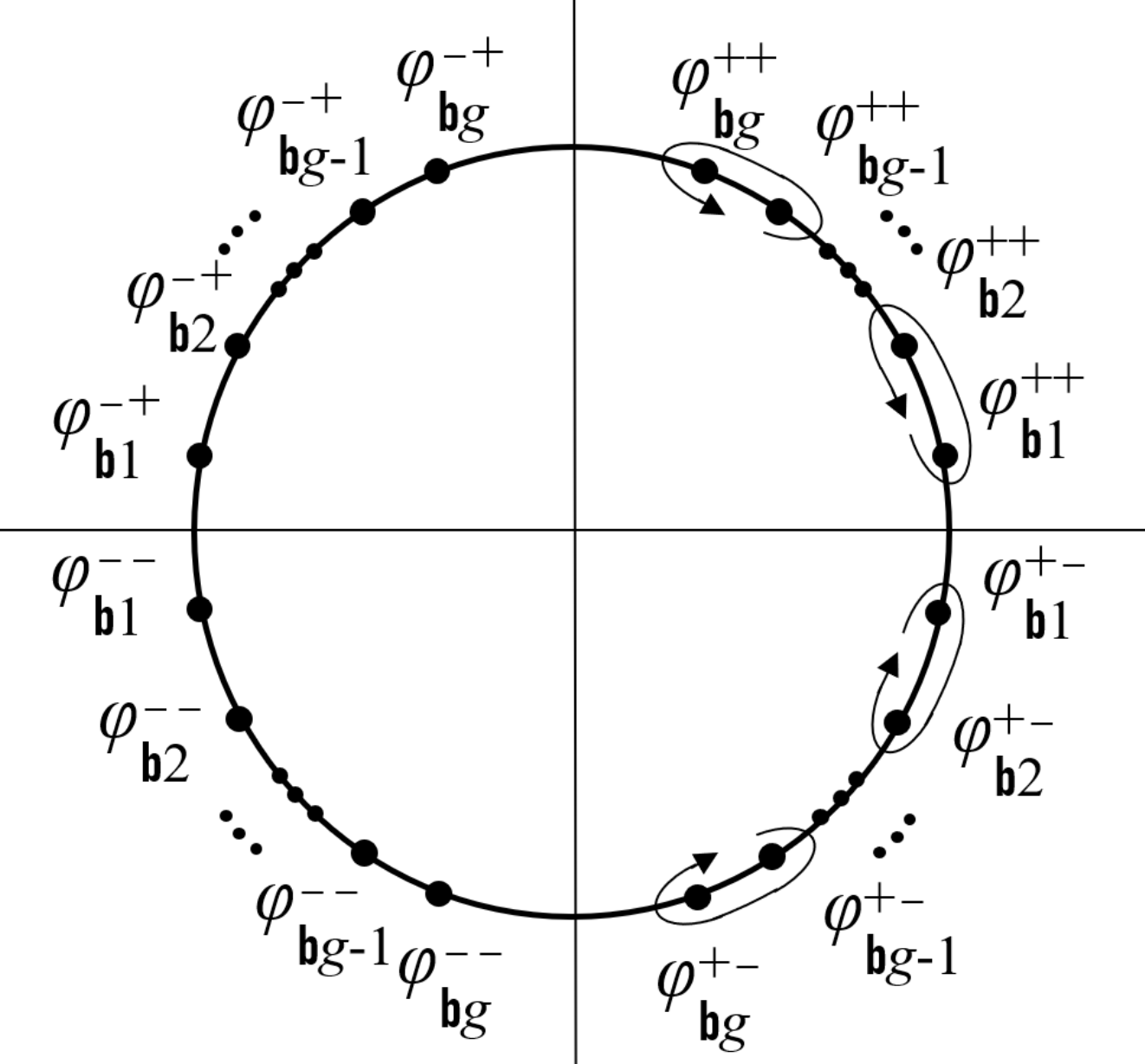}

(a) \hskip  0.35\hsize 
(b)
\end{center}

\caption{
$\cA_X \subset S^1$:
$k_1> k_2 > \cdots > k_g>1.0$ 
 for the odd $g$ (a) and the even $g$ (b).
}\label{fg:Fig01}
\end{figure}

We recall $w^2 = (x-b_0)$ and $w=\ee^{\ii \varphi}\in \hX\setminus\{(0,0)\}$.
For a {\lq\lq}real{\rq\rq} expression of (\ref{4eq:hypC}), we use the following transformation, which is a generalization of {\lq\lq}the angle expression{\rq\rq} of the elliptic integral as mentioned in \cite{MP22}.
\begin{lemma}
$(w^2-e_1)(w^2-e_2)=4\frac{1}{k_1^2}\ee^{2\ii \varphi}
(1-k^2\sin^2\varphi)$, where 
$$
w= \ee^{\ii \varphi}, \quad 
k_1 = \frac{2\ii\sqrt[4]{e_{1}e_{2}}}{\sqrt{e_{1}}- \sqrt{e_{2}}}
=\frac{1}{\sin\varphi_{\fb1}^{++}}, \quad
 e_1 e_2 = 1.
$$
\end{lemma}

\begin{proof}
Let $ e_1 e_2=1$.
We recall the double angle formula $\cos 2\varphi = 1-2\sin^2 \varphi$.
\begin{gather*}
\begin{split}
(w^2 - e_1)(w^2 -e_2)& =w^2(w^2-(e_1+e_2) +e_1 e_2 w^{-2})\\
&=2w^2 \left(\cos(2\varphi) - \frac{e_1+e_2}{2}\right)\\
&=-w^2 
\left(
e_1+e_2-2\sqrt{e_1e_2}+4\sin^2 \varphi\right) \\
&=4w^2 \frac{1}{k_1^2}
\left(1-k_1^2\sin^2 \varphi\right), \\
\end{split}
\end{gather*}
where $(e_1+e_2-2\sqrt{e_1e_2})=(\sqrt{e_1}-\sqrt{e_2})^2 = e_1^{-1}(e_1+1)^2=-4/k_1^2$.
\qed
\end{proof}

\begin{figure}
\begin{center}
\includegraphics[width=0.85\hsize]{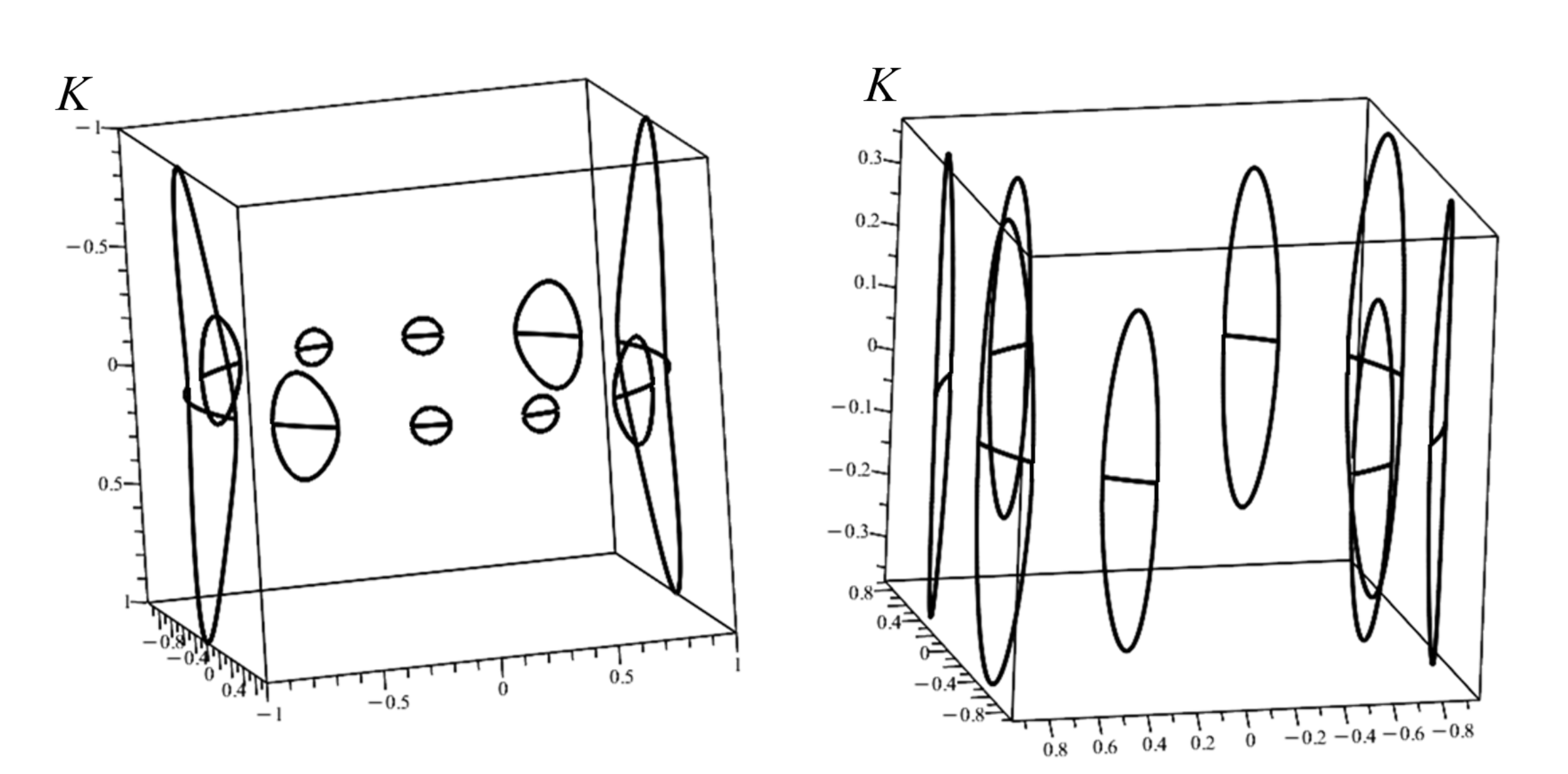}

(a) \hskip 0.35\hsize 
(b)
\end{center}

\caption{
$\cA_X \cup H_\RR \subset \CC \times \RR$
 for the odd $g$ ($g=5$) (a), and the even $g$ ($g=4$) (b).
}\label{fg:Fig02}
\end{figure}

Under these assumptions, we have the {\lq\lq}real{\rq\rq} extension of the hyperelliptic curve $X$ by $(\ee^{\ii\varphi},$ $ y/\ee^{\ii\varphi})\in \hX$. 
The direct computation shows the following:

\begin{lemma} \label{4lm:g3gene_y2}
Let $ \ee^{2\ii\varphi} :=(x-b_0)\in X \setminus\{(0,1)\}$.
(\ref{4eq:hypC}) is written by
\begin{gather}
\begin{split}
y^2&=(-4)^g \frac{\ee^{(2g+2)\ii\varphi}}{\prod_i (k_i^2) } 
(1-k_1^2 \sin^2 \varphi)(1-k_2^2 \sin^2 \varphi)\cdots
(1-k_g^2 \sin^2 \varphi)\\
&= ((2\ii)^{g} \ee^{(g+1)\ii\varphi}K)^2,
\label{4eq:HEcurve_phi}
\end{split}
\end{gather}
where 
$\displaystyle{
k_a = \frac{2\ii\sqrt[4]{e_{2a-1}e_{2a}}}{\sqrt{e_{2a-1}}- \sqrt{e_{2a}}}
=\frac{1}{\sin(\varphi_{\fb,a}^{++})}}$, $(a=1,2,\ldots, g)$, $K=\tgamma\tK(\varphi)$, $\tgamma = \pm 1$ and
$$
\displaystyle{
\tK(\varphi):=\frac{\sqrt{
 (1-k_1^2 \sin^2 \varphi)(1-k_2^2 \sin^2 \varphi)\cdots
(1-k_g^2 \sin^2 \varphi)}}
{k_1 k_2 \cdots k_g}}.
$$
\end{lemma}

\begin{definition}\label{Asmp2}
{\rm{
Hereafter, we assume that $\varphi\in [-\pi, \pi)=\RR/(2\pi \ZZ)-\pi$ as a local parameter of the covering $\RR$ of $S^1:=\{\ee^{\ii \varphi}\}$ and consider $H:=\hvarpi_x^{-1}S^1$.
$H$ is parameterized by $(\ee^{\ii \varphi}, K)$ and $(\varphi, K)$.
We introduce the sets
$$
H_\RR :=\{ (\ee^{2\ii \varphi}, K)\in H \ |\ K \in \RR\}, 
\quad
H_\RR^\varphi :=\{ (\varphi, K) \ |\ \varphi \in [-\pi, \pi), (\ee^{\ii \varphi}, K)\in H_\RR \},
$$
so that we have natural immersion $\iota_H$ and projection $\varpi_H$:
\begin{equation}
\xymatrix{ 
 H_\RR \ar[dr]^{\varpi_H}\ar[r]^-{\iota_H}& \hX \ar[d]^-{\varpi_\hX|_H} \\
  & {\ \ \ \  \cA_X\subset S^1}.
}\label{2eq:H_hyp_kappas}
\end{equation}

Further, we define the  their covering spaces
$$
\tH :=\kappa_{\hX}^{-1} H, \quad
\tH_\RR :=\kappa_{\hX}^{-1} H_\RR, \quad
\tH_\RR^\varphi,
$$
and their $g$-th symmetric products
$$
S^g H, \quad S^g H_\RR, \quad
S^g H_\RR^\varphi, \quad
S^g \tH, \quad S^g \tH_\RR, \quad
S^g \tH_\RR^\varphi.
$$
Here we also define the map,
\begin{equation}
\Phi: \tX \supset \tH_\RR \to [-\pi, \pi), \quad\left( \gamma \mapsto 
\varphi_\gamma=\frac{1}{\ii}\log (\hvarpi_x\hkappa_\hX\gamma -b_0)\right).
\label{eq:trans_Phi}
\end{equation}
Moreover, we define $[S^g \tH_\RR]^0:=$
 $(\varpi_H\circ\kappa_\hX)^{-1}\prod_{a=1}^g [\varphi_\fb^{[a-]}, \varphi_\fb^{[a+]}]$, where \break
$\prod_{a=1}^g [\varphi_\fb^{[a-]}, \varphi_\fb^{[a+]}]\subset S^g \cA_X^\varphi$.
In other words, for $(\gamma_1, \ldots, \gamma_g) \in [S^g \tH_\RR]^0$, we set $\varpi_H\circ\kappa_\hX\gamma_i$ into $[\varphi_\fb^{[i-]}, \varphi_\fb^{[i+]}]$ respectively.

}}
\end{definition}

We note that Figure \ref{fg:Fig02} displays an example of $H_\RR$ with $\cA_X$ in (\ref{eq:cAXodd}) and (\ref{eq:cAXeven}).
From here on, we will loosely identify $H_\RR$ with $H_\RR^\varphi$.
For simplicity,  we restrict $S^g \tH_\RR \subset S^g \tX$ to $[S^g \tH_\RR]^0$.
 
\begin{lemma}
For $(\gamma_1, \ldots, \gamma_g) \in [S^g \tH_\RR]^0$, there is no intersection between $\gamma_i$ and $\gamma_j$ for $i\neq j$.
Further, $H_\RR$ consists of $2g$ loops.
$\dim_\RR S^g H_\RR = \dim_\RR S^g H_\RR^\varphi =$ 
$\dim_\RR S^g \tH_\RR = \dim_\RR S^g \tH_\RR^\varphi=g$.
\end{lemma}

\begin{proof}
It is obvious.
 $H_\RR$ is homeomorphic to $(S^1)^{2g}$ in Figure \ref{fg:Fig02}.
\end{proof}

\begin{lemma}
For $x-b_0 = \ee^{2\ii\varphi}\in \varpi_x(H)$ and $b_0=1$, 
$\nuI{\ell}$  in (\ref{4eq:firstdiff}) is equal to
$$
\nuI{\ell}=\frac{(2\ii) \ee^{-(g-\ell)\ii\varphi}
(2\ii\sin(\varphi))^{\ell-1} d\varphi}{(2\ii)^{g} K}, \quad
(\ell = 1, 2, \ldots, g).
$$
\end{lemma}

\begin{proof}
Since we have $x=  \ee^{\ii \varphi}(\ee^{\ii \varphi}- \ee^{-\ii \varphi})$ $=2 \ii  \ee^{\ii \varphi} \sin \varphi$ and  $d x = 2 \ii  \ee^{2\ii \varphi}d\varphi$, we have
$
x^{\ell-1} d x =$ $ (2\ii)^{\ell}  \ee^{(\ell+1)\ii \varphi}\sin^{\ell-1} \varphi\ d \varphi$.
\end{proof}

Here we note that $\displaystyle{
\nuI{g}=\frac{(\sin(\varphi))^{g-1} d\varphi}{K}
}$ is real if $(\varphi, K)$ belongs to $H_\RR^\varphi$ and $d \varphi$ is real.

Then we obviously have the following lemmas:
\begin{lemma} \label{4lm:dudphi}
For $( \ee^{\ii \varphi_j}, K_j)_{j=1, 2, \ldots, g}\in S^gH$, the following holds:
$$
\left[\begin{matrix} d u_1 \\ d u_2\\ \vdots \\d u_\ell\\ \vdots \\ du_g
\end{matrix}\right]
=-
\left[
\frac{ \ee^{-(g-\ell)\ii\varphi_k}(2\ii \sin \varphi_k)^{\ell-1}}
{(2\ii)^{g-1}  K_k}
\right]
\left[\begin{matrix} d \varphi_1 \\ d \varphi_2 \\ \vdots \\
d \varphi_k \\ \vdots \\ d \varphi_g\end{matrix}\right].
$$
Let the matrix be denoted by $\cL=(\cL_{ij})$,
$$
\cL_{ij}=\displaystyle{\left[
\frac{(\cos\varphi_j-\ii\sin\varphi_j)^{(g-i)}(2\ii \sin\varphi_j)^{i-1}}
{(2\ii)^{g-1}  K_j}\right]_{ij}}.
$$
Then the determinant of $\cL$, 
$$
\det(\cL)=\displaystyle{
\frac{\prod_{i<j}\sin(\varphi_i-\varphi_j)}{(2\ii)^{g(g-1)/2} K_1 K_2 \cdots K_g}}.
$$
\end{lemma}
We denote $\fs_i:= \sin \varphi_i$, $\fc_i:= \cos \varphi_i$, and $\hfc_i = \ii
\fc_i$, and then we have
$$
\cL_{ij}=\displaystyle{\left[\frac{(\fc_j-\ii \fs_j)^{g-i}(2 \ii \fs_j)^{i-1}}{
(2\ii)^{g-1} K_j}\right]}.
$$
\begin{lemma}\label{lm:3.9}
{\small{
$$
\cK\!:=\!2^{g-1}
\left[\begin{matrix}
 K_1 &  & & \\
& K_2 &  &  \\
 &  &\ddots & \\
& & &  K_g \\
\end{matrix}\right], 
\cM\!:=\!
\left[\begin{matrix}
(\fs_1+\ii\fc_1)^{g-1} &  \cdots &  (\fs_g+\ii\fc_g)^{g-1}\\
2(\fs_1+\ii\fc_1)^{g-2}\fs_1  & \cdots 
& 2 (\fs_g+\ii\fc_g)^{g-2}\fs_g\\
\vdots & \ddots &\vdots \\
2^{g-1}\fs_1^{g-1} &  \cdots &  2^{g-1}\fs_g^{g-1}\\
\end{matrix}\right], 
$$
}}
then $\cL =  \cM\cK^{-1}$.
\end{lemma}

We extend (\ref{4eq:KdV_def2.1}) and Lemma \ref{4lm:KdV1}.

\begin{definition} \label{def:varepsilon}
We define $\varepsilon_{j,k}$ by the polynomial of $\eta$, $\displaystyle{\prod_{i=1,\neq j}^{g} ((\fs_i + \ii \fc_i) \eta - 2 \fs_i)}$
$=\varepsilon_{j,g-1} \eta^{g-1} + \varepsilon_{j,g-2} \eta^{g-2} + \cdots + \varepsilon_{j,1} \eta + \varepsilon_{j,0}$.
\end{definition}

Then we obviously have the relations: $\varepsilon_{j,0}=\displaystyle{\prod_{i=1,\neq j}^{g} 2\ii \fs_i}$, and $\varepsilon_{j,g-1}=\displaystyle{\prod_{i=1,\neq j}^{g} (\fs_i + \ii \fc_i)}$. 

\begin{lemma}\label{lm:varepsilon}
{\small{
\begin{align*}
\cE\!:=\!
\left[\begin{matrix}
\varepsilon_{1,0} & \varepsilon_{1,1} & \cdots &  \varepsilon_{1,g-1}\\
\varepsilon_{2,0} & \varepsilon_{2,1} & \cdots &  \varepsilon_{2,g-1}\\
\vdots & \vdots &\ddots &\vdots \\
\varepsilon_{g,0} & \varepsilon_{g,1} & \cdots &  \varepsilon_{g,g-1}\\
\end{matrix}\right], \quad
\cD\!:=\!
2^{g-1}\left[\begin{matrix}
\displaystyle{\prod_{j\neq 1}}(\fc_j\fs_1-\fc_1\fs_j) &  &  \\
  &\ddots & \\
& & \displaystyle{\prod_{j\neq g}}(\fc_j\fs_g-\fc_g\fs_j) 
\end{matrix}\right], 
\end{align*}
}}
$\cE \cM = \cD$, and thus $\cM\cD^{-1}  \cE=I$, 
where $I$ is the identity matrix.
\end{lemma}

\begin{proof}
See Appendix A but its essential is the same as Lemma \ref{4lm:KdV1}.
\end{proof}

We extend Lemma \ref{4lm:KdV1} to the case of angle expressions.

\begin{lemma} \label{4lm:dudphig1}
For $(\ee^{\ii\varphi_a}, K_a)_{a} \in [S^g \tH_\RR]^0$ such that $\varphi_a \neq \varphi_b$ $(a\neq b)$, we have
\begin{equation}
\displaystyle{
\begin{pmatrix} d \varphi_1 \\ d \varphi_2 \\ \vdots \\ d \varphi_g\end{pmatrix}
=\cK^{-1}\cD^{-1} \cE
\begin{pmatrix} d u_1 \\ d u_2\\ \vdots \\ du_g\end{pmatrix}
},
\qquad
\cL^{-1}=\cK \cE.
\label{3eq:Elas3.7}
\end{equation}
\end{lemma}

Since the elements of $\cE$ are not real even for $\varphi_i \in \RR$, we decompose these elements in $\cE$ into imaginary parts and real parts.
However the decomposition is, a little bit, complicated and thus, we first consider the case of genus 5 in Subsection \ref{ssc:g=5} and then, we treat general $g$ in Subsection \ref{ssc:g=g}.

\subsection{The case of $g=5$:}\label{ssc:g=5}
This subsection is on the matrix $\cE$ only for $g=5$ before we deal with general $g$ case since its structure is slightly complicated.

Before we show the decomposition of $\cE$ into the real parts and imaginary parts, we concretely write down the entities in Lemmas \ref{lm:3.9}, Definition \ref{def:varepsilon} and Lemma \ref{lm:varepsilon}.
{\small{ 
$$
\cK\!:=\!16
\left[\begin{matrix}
K_1 &  & & \\
&  K_2 &  &  \\
 &  &\ddots & \\
& & & K_5 \\
\end{matrix}\right],\  
\cM\!:=\!
\left[\begin{matrix}
(\fs_1+\ii\fc_1)^{4} &  \cdots &  (\fs_g+\ii\fc_g)^{4}\\
2(\fs_1+\ii\fc_1)^{3}\fs_1& \cdots 
& 2 (\fs_4+\ii\fc_4)^{3}\fs_g\\
\vdots  &\ddots &\vdots \\
16\fs_1^4  & \cdots &  16\fs_5^{4}\\
\end{matrix}\right], 
$$
}}

\begin{align}
\fe_{5,0} &:= 16 \fs_{1} \fs_{2} \fs_{3}\fs_{4},\nonumber\\
\fe_{5,1} &:= 8 \fs_{1} \fs_{2} \fs_{3}\fc_{4}
             +8 \fs_{1} \fs_{2} \fc_{3}\fs_{4}
             +8 \fs_{1} \fc_{2} \fs_{3}\fs_{4}
             +8 \fc_{1} \fs_{2} \fs_{3}\fs_{4},\nonumber\\
\fe_{5,2} &= 4 \fs_{1} \fs_{2} \fc_{3} \fc_{4} 
            + 4 \fs_{1} \fs_{3} \fc_{2} \fc_{4}
            + 4 \fs_{1} \fs_{4} \fc_{2} \fc_{3}
            + 4 \fs_{2} \fs_{3} \fc_{1} \fc_{4}
            + 4 \fs_{2} \fs_{4} \fc_{1} \fc_{3}
            + 4 \fs_{3} \fs_{4} \fc_{1} \fc_{2},\nonumber\\
\fe_{5,3} &=  2 \fs_{1} \fc_{2} \fc_{3} \fc_{4}
            + 2 \fs_{2} \fc_{1} \fc_{3} \fc_{4}
            + 2 \fs_{3} \fc_{1} \fc_{2} \fc_{4}
            + 2 \fs_{4} \fc_{1} \fc_{2} \fc_{3},\nonumber\\
\fe_{5,4} &= \fc_{1} \fc_{2} \fc_{3} \fc_{4}.
\end{align}

\begin{lemma}
\begin{eqnarray*}
\varepsilon_{j,0}&=& \fe_{j,0},\\ 
\varepsilon_{j,1}&=& -2\fe_{j,0}-\ii\fe_{j,1},\\ 
\varepsilon_{j,2}&=& \frac{3}{2}\fe_{j,0} +\ii \frac{3}{2} \fe_{j,1}-\fe_{j,2},\\ 
\varepsilon_{j,3}&=& -\frac{1}{2}\fe_{j,0}
-\ii\frac{3}{4}\fe_{j,1}+\fe_{j,2}+\ii\fe_{j,3},\\ 
\varepsilon_{j,4}&=& \frac{1}{16}\fe_{j,0}
+\ii\frac{1}{8}\fe_{j,1}-\frac{1}{4}\fe_{j,2}-\ii\frac{1}{2}\fe_{j,3}+\fe_{j,4}.
\end{eqnarray*}
\end{lemma}

\begin{lemma}
{\small{
\begin{align*}
\cE\!:=\!
\left[\begin{matrix}
\varepsilon_{1,0} & \varepsilon_{1,1} & \cdots &  \varepsilon_{1,4}\\
\varepsilon_{2,0} & \varepsilon_{2,1} & \cdots &  \varepsilon_{2,4}\\
\vdots & \vdots &\ddots &\vdots \\
\varepsilon_{5,0} & \varepsilon_{5,1} & \cdots &  \varepsilon_{5,4}\\
\end{matrix}\right]\!, \ \ 
\cD\!:=\!
16\left[\begin{matrix}
\displaystyle{\prod_{j\neq 1}}(\fc_j\fs_1-\fc_1\fs_j)  & & \\
   &\ddots & \\
 & & \displaystyle{\prod_{j\neq 5}}(\fc_j\fs_5-\fc_5\fs_j)  \\
\end{matrix}\right]\!, 
\end{align*}
}}
$\cE \cM = \cD$, and thus $\cM \cD^{-1} \cE=I$.
\end{lemma}

As mentioned above, since the elements of $\cE$ are not real even for $\varphi_i \in \RR$, we will decompose these elements in $\cE$ into imaginary parts and real parts.

\begin{align}
&\varepsilon_{j,0,\even}:= \fe_{j,0},\
\varepsilon_{j,2,\even}:= \frac{3}{2}\fe_{j,0} -  \fe_{j,2},\
\varepsilon_{j,4,\even}:=  \frac{1}{16}\fe_{j,0}- \frac{1}{4}\fe_{j,2}
                           +\fe_{j,4},\nonumber\\
&\varepsilon_{j,2,\odd}:= \frac{3}{2}\fe_{j,1},\
\varepsilon_{j,4,\odd}:= \frac{1}{8}\fe_{j,1}-\frac{1}{2}\fe_{j,3}.
\label{eq:varepsilon_eo5}
\end{align}

\begin{lemma}
\begin{align*}
\tcV:=\left[\begin{matrix}
\varepsilon_{1,0,\even} & \varepsilon_{1,2,\odd} & 
\varepsilon_{1,2,\even} & 
\varepsilon_{1,4,\odd} & 
\varepsilon_{1,4,\even}\\
\varepsilon_{2,0,\even} & 
\varepsilon_{2,2,\odd} & 
\varepsilon_{2,2,\even} & 
\varepsilon_{2,4,\odd} & 
\varepsilon_{2,4,\even}\\
\vdots & \vdots &\ddots &\vdots \\
\varepsilon_{5,0,\even} & 
\varepsilon_{5,2,\odd} & 
\varepsilon_{5,2,\even} & 
\varepsilon_{5,4,\odd} & 
\varepsilon_{5,4,\even}
\end{matrix}\right], 
\end{align*}
\begin{align*}
\cI:=\left[\begin{matrix}
1 &  &  &  & \\
 & \ii&  &  & \\
 &   & 1    & &  \\
 &  &    &\ii &  \\
 &   &   & & 1 \\
\end{matrix}\right], \quad
\cB:=\left[\begin{matrix}
1 & -2 & 0 & 1 & 0\\
0 & -2/3& 1 & -1/3 & 0\\
0 & 0  & 1    &-1 & 0 \\
0 & 0  & 0    &-2 & 1 \\
0 & 0  &0   &0 & 1 \\
\end{matrix}\right],
\end{align*}
then $\cE=\tcV \cI \cB$ and 
$\cM \cD^{-1} \tcV \cI \cB=I$.

By letting $\cV:=(\cV_1, \ldots, \cV_5):=\cK\cD^{-1}\tcV$, we have $\cL\cV\cI\cB=I$.
Further, for $(\varphi_a, K_a)_{a=1,\ldots,5} \in S^5 H_{\RR}^\varphi$, $(\varphi_a \neq \varphi_b)$ $(a\neq b)$, 
$\cV_{i,j}$ is real valued, and $\dim_\RR\langle\cV_1, \ldots, \cV_5\rangle_\RR=5$.
\end{lemma}

\begin{proof}
We note $\cL=\cM\cK^{-1}$ and $\cM^{-1}=\cD^{-1}\cE$.
Straight computations show them.
$\cV_1, \ldots, \cV_5$ are linear independent as an $\RR$ vector space.
\end{proof}

\begin{lemma} \label{4lm:dudphi5}
For $(\varphi_a, K_a)_{a=1, 2, \ldots, 5} \in S^5 H$ such that $\varphi_a \neq \varphi_b$ $(a\neq b)$, we have
\begin{equation}
\displaystyle{
\left[
\begin{matrix} d \varphi_1 \\ d \varphi_2 \\ 
d \varphi_3\\ d \varphi_4 \\ d \varphi_5\end{matrix}
\right]
\!=\![\cK\cD]^{-1} \cE
\left[\begin{matrix} d u_1 \\ d u_2\\ du_3\\ du_4 \\ du_5\end{matrix}
\right]
\!=\!\cK\cD^{-1}\tcV\cI\cB
\left[\begin{matrix} d u_1 \\ d u_2\\ du_3\\ du_4 \\ du_5\end{matrix}
\right]
\!=\!\cK\cD^{-1}\tcV
\left[
\begin{matrix} d t_1 \\ d t_2\\ dt_3\\ dt_4\\ dt_5\end{matrix}\right]
=\cV
\left[
\begin{matrix} d t_1 \\  d t_2\\ dt_3\\ dt_4\\ dt_5\end{matrix}\right]
},
\label{eq:Elas3.10}
\end{equation}
where
{\small{
\begin{equation}
\left[
\begin{matrix} dt_1 \\ \vdots \\ d t_5\end{matrix}\right]
\!:=\!\cI\cB\!
\left[\begin{matrix} du_1 \\ \vdots \\ d u_5\end{matrix}\right]\!, \ 
\left[\begin{matrix} du_1 \\ \vdots \\ d u_5\end{matrix}\right]
\!=\![\cI\cB]^{-1}\!
\left[\begin{matrix} dt_1 \\ \vdots \\ d t_5\end{matrix}\right]\!, \ 
[\cI\cB]^{-1}\!:=\!
\left[\begin{matrix}
1 & 3\ii & 3 & \ii/2 & 1/2\\
0 & 3\ii/2& 3/2 & \ii/2 & 1/2\\
0 & 0  & 1    &\ii/2 & 1/2 \\
0 & 0  & 0    &\ii/2 & 1/2 \\
0 & 0  &0   &0 & 1 \\
\end{matrix}\right]\!.
\label{eq:dus_dts}
\end{equation}
}}
\end{lemma}

Since (\ref{eq:dus_dts}) also shows
$$
(\cV_1,  \cV_2, \ldots, \cV_5)=
\begin{pmatrix}
\partial_{t_1} \varphi_1&  \partial_{t_2} \varphi_1& 
\cdots & \partial_{t_5} \varphi_1\\
\partial_{t_1} \varphi_2& \partial_{t_2} \varphi_2& 
\cdots & \partial_{t_5} \varphi_2\\
   \vdots & \vdots & \ddots & \vdots\\
\partial_{t_1} \varphi_5& -  \partial_{t_2} \varphi_5& 
\cdots &\partial_{t_5} \varphi_5
\end{pmatrix}, 
$$
$$
\begin{pmatrix} d\varphi_1 \\ d \varphi_2 \\ \vdots \\ d \varphi_5\end{pmatrix}
=
(\cV_1,  \cV_2, \ldots, \cV_5)
\begin{pmatrix} dt_1 \\ d t_2 \\  dt_3 \\ d t_4 \\ d t_5\end{pmatrix},
$$
the relations between differential operators are given by
\begin{equation}
\begin{pmatrix} 
\partial_{u_1} \\ \partial_{u_2} \\ 
\partial_{u_3} \\
\partial_{u_4} \\
 \partial_{u_5}\end{pmatrix}
=\trp [\cI\cB]
\begin{pmatrix} 
\partial_{t_1} \\ 
\partial_{t_2} \\ 
\partial_{t_3} \\
\partial_{t_4} \\
\partial_{t_5}\end{pmatrix}
,\quad
\trp [\cI\cB]=
\begin{pmatrix} 
1 & 0 & 0 & 0 & 0 \\ 
-2 & -2\ii/3 & 0 & 0 & 0\\
0 & \ii & 1 & 0 & 0\\
1 & -\ii/3 & -1 & -2\ii & 0\\
0 & 0 & 0 & \ii & 1 \end{pmatrix}.
\label{eq:pus_pts}
\end{equation}
Noting $\trp [\cI\cB]=[[\frac{\partial t_j}{\partial u_i}]_{ij}]$, we have the relations,
\begin{equation} 
\partial_{u_5} = \partial_{t_5}+\ii \partial_{t_4}, \quad
\partial_{u_4} = -\partial_{t_3}-2\ii \partial_{t_4}
+\partial_{t_1}-\frac{\ii}{3} \partial_{t_2}.
\label{eq:22}
\end{equation}

\begin{definition}
$$
(\be_1, \ii\be_2, \be_3, \ii\be_4, \be_5):=[\cI\cB]^{-1}.
$$
\end{definition}

\subsection{The case of general $g$:}\label{ssc:g=g}

As we displayed the case of $g=5$ in the previous subsection, we generalize the arguments of $g=5$ to general $g$: 
the following naturally contains the genus three case in \cite{Ma24d} and five in the previous subsection.
Further, the precise proofs are in Appendix.

As we introduced $\varepsilon$ in Definition \ref{def:varepsilon} which has both real and imaginary parts, we decompose it into real parts and imaginary parts following (\ref{eq:varepsilon_eo5}).

\begin{definition}
By assuming that $\fs_i$ and $\fc_i$ are real, we let
$$
\varepsilon_{i,j,\even}:=\re(\varepsilon_{i,j}), \quad
\varepsilon_{i,j,\odd}:=\im(\varepsilon_{i,j}).
$$
For the odd $g$ case, let
{\small{
$$
\tcV:= \left[\begin{matrix}
 \varepsilon_{1,0,\even} &  \varepsilon_{1,2,\odd} & \cdots &  
     \varepsilon_{1,2\ell,\odd}  & \varepsilon_{1,2\ell,\even}  
  &  \cdots   \\
 \varepsilon_{2,0,\even} &  \varepsilon_{2,2,\odd} & \cdots &  
     \varepsilon_{2,2\ell,\odd}  & \varepsilon_{2,2\ell,\even}  
  &  \cdots   \\
\vdots   &\vdots& \ddots &  \vdots& \vdots& \ddots   \\
 \varepsilon_{g-1,0,\even} &  \varepsilon_{g-1,2,\odd} & \cdots &  
     \varepsilon_{g-1,2\ell,\odd}  & \varepsilon_{g-1,2\ell,\even}  
  &  \cdots   \\
 \varepsilon_{g,0,\even} &  \varepsilon_{g,2,\odd} & \cdots &  
     \varepsilon_{g,2\ell,\odd}  & \varepsilon_{g-1,2\ell,\even}  
  &  \cdots   
\end{matrix}\right.
\hskip 0.1\hsize
$$
$$
\hskip 0.6\hsize
\left.\begin{matrix}
  \varepsilon_{1,g-1,\odd}  &  \varepsilon_{1,g-1,\even}\\
  \varepsilon_{2,g-1,\odd}  &  \varepsilon_{2,g-1,\even}\\
 \vdots &\vdots \\
  \varepsilon_{g-1,g-1,\odd}  &  \varepsilon_{g-1,g-1,\even}\\
  \varepsilon_{g,g-1,\odd}  &  \varepsilon_{g,g-1,\even}
\end{matrix}\right].
$$
}}
For the even $g$ case, let
{\small{
$$
\tcV:= \left[\begin{matrix}
 \varepsilon_{1,1,\odd} &  \varepsilon_{1,1,\even} & \cdots &   \varepsilon_{1,2\ell-1,\odd}  &  \varepsilon_{1,2\ell-1,\even}  &  \cdots    \\
 \varepsilon_{2,1,\odd} &  \varepsilon_{2,1,\even} & \cdots &   \varepsilon_{2,2\ell-1,\odd}  &  \varepsilon_{2,2\ell-1,\even}  &  \cdots    \\
\vdots   &\vdots& \ddots &  \vdots& \vdots& \ddots   \\
 \varepsilon_{g-1,1,\odd} &  \varepsilon_{g-1,1,\even} & \cdots &  
 \varepsilon_{g-1,2\ell-1,\odd}  &  \varepsilon_{g-1,2\ell-1,\even}  &  \cdots   \\
 \varepsilon_{g,1,\odd} &  \varepsilon_{g-1,1,\even} & \cdots &  
 \varepsilon_{g,2\ell-1,\odd}  &  \varepsilon_{g,2\ell-1,\even}  &  \cdots   \\
\end{matrix}\right.
\hskip 0.1\hsize
$$
$$
\hskip 0.6\hsize
\left.\begin{matrix}
 \varepsilon_{1,g-1,\odd}  &  \varepsilon_{1,g-1,\even}  \\
 \varepsilon_{2,g-1,\odd}  &  \varepsilon_{2,g-1,\even}  \\
 \vdots &\vdots \\
 \varepsilon_{g-1,g-1,\odd}  &  \varepsilon_{g-1,g-1,\even}  \\
 \varepsilon_{g,g-1,\odd}  &  \varepsilon_{g,g-1,\even}  \\
\end{matrix}\right].
$$
}}
\end{definition}

The following is a key lemma.
\begin{lemma}\label{lm4.1}
We define a $g\times g$ matrix:
Then $\cE$ is expressed as
$$
\cE= \tcV \cI\cB,  \qquad
\cI\cB= \begin{pmatrix} \cB^{[g-3,g-3]} & \cB^{[g-3],1}& \cB^{[g-3],2} &0\\
                    0& 1  & -1 & 0 \\
                     0 &0 & -2\ii & \ii \\
                      0 &0        & 0 & 1 \end{pmatrix}\in \GL(\QQ[\ii], g),
$$
where
$$
\cI = \mathrm{diag} \cI_{ii}, \quad
\cI_{ii} = \left\{
\begin{matrix} 1 & \mbox{if }g - i \mbox{ is even,}\\
\ii & \mbox{if }g - i \mbox{ is odd,}
\end{matrix} \right.
$$
$\cB^{[g-3,g-3]} \in \Mat_\CC( (g -3) \times (g-3))$, $\cB^{[g-3],1}, \cB^{[g-3],2}\in \Mat_\CC( (g -3) \times 1)$, and $0$ denotes a zero $\ell \times k$ matrix, although $(0)=0$.
Further, $\cB^{[g-3],2}=\trp(\fb_1, \ii \fb_2, \ldots, \fb_{g-4}, \ii \fb_{g-3})$ for an odd $g$ and $\cB^{[g-3],2}=\trp(\ii\fb_1, \fb_2, \ldots, \fb_{g-4}, \ii \fb_{g-3})$ for an even $g$.
\end{lemma}

In other words, we will decompose the image of the Abelian integral or the Abel-Jacobi map from a real analytic viewpoint or consider the $\RR$-linear transformation in $T^*\CC^g$:

\begin{definition} We define the following:
$$
\cV:=\cK \cD^{-1} \tcV, \quad
\begin{array}{ll}
(\be_1, \ii\be_2, \be_3, \ldots, \ii\be_{g-1}, \be_g):=[\cI\cB]^{-1}
& \text{for odd }g,\\
(\ii\be_1, \be_2, \ii\be_3, \ldots, \ii\be_{g-1}, \be_g):=[\cI\cB]^{-1}
& \text{for even }g.
\end{array}
$$
\end{definition}

We have a key lemma.
In terms of $\cV$, Lemma \ref{4lm:dudphig1} reads as follows.
\begin{lemma} \label{4lm:dudphig}
For $(\varphi_a, K_a)_{a=1, \ldots, g} \in S^g H$, $\varphi_a \neq \varphi_b$ $(a\neq b)$, we have
\begin{equation}
\displaystyle{
\begin{pmatrix} d \varphi_1 \\ d \varphi_2 \\ \vdots \\ d \varphi_g\end{pmatrix}
=\cK\cD^{-1} \cE
\begin{pmatrix} d u_1 \\ d u_2\\ \vdots \\ du_g\end{pmatrix}
=\cV
\begin{pmatrix} d t_1 \\ d t_2\\ \vdots \\ dt_g\end{pmatrix}
},
\qquad
\cL^{-1}=\cK \cE.
\label{eq:Elas3.14}
\end{equation}
\end{lemma}

\begin{proof}
The straightforward computations show it for $(\varphi_a, K_a)$.
Even at the branch point, this expression works since $\nu_i$ is a holomorphic one-form. \qed
\end{proof}

\section{Real hyperelliptic solutions of the FGMKdV equation over $\RR$ of $g$}

As we show its proof and background in Appendix, this angle expression provides  the following lemmas:
Lemma \ref{lm:4.7} is one of the key lemmas in this paper.

We remark that (\ref{eq:Elas3.14}) in Lemma \ref{4lm:dudphig} means that even if $\varphi_j$, $(j=1, 2, \ldots, g)$ is real, $d\varphi_j$ is a complex valued one-form.
We let it decomposed to $d\varphi_j = d\varphi_{j,\rr}+ \ii d \varphi_{j, \ri}$, $(j=1, 2, \ldots, g)$.
Further, we introduce $\varphi := \varphi_1 + \cdots +\varphi_g \in \RR$ and $d\varphi = d\varphi_{\rr}+ \ii d \varphi_{\ri}$; 
$\psi_\rr = 2 \varphi$, $d\psi_\rr = 2 d\varphi_\rr$ and $d\psi_\ri =2 d\varphi_\ri$ for $\psi$ in (\ref{4eq:loopMKdV2}) and (\ref{4eq:gaugedMKdV2}).
We sometimes write $\varphi_{a,\rr}:=\varphi_a$.

Using them, we have the following lemma:
\begin{lemma}\label{lm:4.ga}
$$
\partial_{t_a}\psi = (1, 1,\ldots,1, 1)\cV_a=
\cV_{a,1}+\cV_{a,2}+\cdots+\cV_{a,g}, \quad (a=1, 2,\ldots, g),
$$
$$
\partial_{u_{g-1}}\psi = (1, 1,\ldots,1, 1)(\cV_{g-2}
+2\ii \cV_{g-1}+\sum_{j=1}^{g-3}\sqrt{(-1)^j} \fb_{g-2-j}\cV_{g-2-j}),
$$
$$
\partial_{u_g}\psi = (1, 1,\ldots,1, 1)(\cV_g+\ii \cV_{g-1}).
$$
\end{lemma}

Here we recall the Cauchy-Riemann relations of these parameters.
We have $u_a = u_{a\rr}+\ii u_{a\ri}$ and let $t_a := t_{a\rr}+\ii t_{a\ri}$ for $a=1, 2, \ldots, g$.
For a complex analytic function $\psi = \psi_\rr + \ii \psi_\ri$,
(\ref{eq:Elas3.14}) shows
$$
\begin{pmatrix} 
\partial_{u_1} \\ \partial_{u_2} \\ 
\vdots \\
 \partial_{u_g}\end{pmatrix}
=\trp [\cI\cB]
\begin{pmatrix} 
\partial_{t_1} \\ 
\partial_{t_2} \\ 
\vdots \\
\partial_{t_g}\end{pmatrix},
$$
\begin{align}
&\partial_{u_a\rr}\psi_\rr=\partial_{u_a\ri}\psi_\ri , \qquad
\partial_{u_a\rr}\psi_\ri=-\partial_{u_a\ri}\psi_\rr,
\qquad (a= 1, 2, \ldots,  g),\nonumber \\
&\partial_{t_a\rr}\psi_\rr=\partial_{t_a\ri}\psi_\ri , \qquad
\partial_{t_a\rr}\psi_\ri=-\partial_{t_a\ri}\psi_\rr,
\qquad (a= 1, 2, \ldots, g).
\label{eq:4.0}
\end{align}

Since we are concerned with the case that $\varphi_a\in \RR$, i.e., $\varphi_{a,\ri}=0$, $(a=1, 2, \ldots, g)$ as in $H_\RR^\varphi$, we may assume that $t_a \in \RR$ belongs to $\CC^g$.
Lemma \ref{lm:4.ga} and (\ref{eq:4.0}) show the following lemma:

\begin{lemma}\label{lm:4.7}
For $\varphi_a\in \RR$, i.e., $\varphi_{a,\ri}=0$, $(a=1, 2, \ldots, g)$,
the following relations hold:
\begin{eqnarray*}
&\partial_{u_{g-1},\rr} \psi_\rr=(-\partial_{t_{g-2},\rr} 
+\sum_{j=1}^{\lfloor(g-3)/2\rfloor} \fb_{g-2-2j}\partial_{t_{g-2-2j}})\psi_\rr\\
&=(1,\ldots, 1)(-\cV_{g-2}
+\sum_{j=1}^{\lfloor(g-3)/2\rfloor}\fb_{g-2-2j}\cV_{g-2-2j}),\\
&\partial_{u_{g-1},\rr} \psi_\ri=(2\partial_{t_{g-1},\rr},
+\sum_{j=1}^{\lfloor(g-3)/2\rfloor} \fb_{g-1-2j}\partial_{t_{g-1-2j}})\psi_\ri\\
&=(1, \ldots, 1)(\cV_{g-1}+\sum_{j=1}^{\lfloor(g-3)/2\rfloor} \fb_{g-1-2j}\cV_{g-1-2j}),
\end{eqnarray*}
\begin{equation}
\partial_{u_g,\rr} \psi_\rr=\partial_{t_g,\rr} \psi_\rr=(1,\ldots,1)\cV_g, \quad
\partial_{u_g,\ri} \psi_\rr=-\partial_{t_{g-1},\rr} \psi_\rr=-(1,\ldots,1)\cV_{g-1}.
\end{equation}
\end{lemma}
\begin{proof}
We obviously obtain them.
\end{proof}

\begin{remark}
{\rm{
Lemma \ref{lm:4.7} does not claim that there is a different complex structure in $J_X$ and the image of $\tv$.
These parameterizations are consistent only for the local regions $H_X$ related to the arcs of $S^1$ in $\hkappa_X \hX$, or $\varphi_a \in \RR$ and $\varphi_{a,\ri}=0$ ($a=1, 2, \ldots, g$).
This means that we simply embed the real vector space $\RR^g$ in $\CC^g$ via the matrix $\cB$ and $\cB^{-1}$.
}}
\end{remark}

\bigskip

The FGMKdV equations (\ref{4eq:gaugedMKdV2}) with Lemma \ref{lm:4.7} can be expressed in terms of the parameterizations of $t$'s.
Since these $d t_a$ are linearly independent, we can set $dt_i =0$ ($i<g-2$).
We give the first theorem in this paper.
\begin{theorem}\label{th:4.2}
Assume $\varphi_i\in \RR$,  $(i=1, 2, \ldots, g)$, i.e., $\psi\in \RR$ or $\psi_{\ri}=0$.
Let $t:=t_{g-2\rr}$, $\ft:=t_{g-1\ri}$,  and $s:=t_{g\rr}$ belonging to $\RR$, and let us consider the differential equation,
\begin{equation}
\begin{pmatrix} d\varphi_1 \\ d \varphi_2 \\ \vdots \\ d \varphi_g\end{pmatrix}
=(\cV_{g-2}, \ii\cV_{g-1},\cV_g)\begin{pmatrix} dt \\ d \ft \\ d s\end{pmatrix}.\label{eq:4.2}
\end{equation}
Then (\ref{4eq:loopMKdV2}) is reduced to the coupled FGMKdV equations,
\begin{eqnarray}
(\partial_{u_{g-1},\rr}-\frac{1}{2}
(\lambda_{2g}+3b_0+\frac{3}{4}(\partial_\ft\psi_\rr)^2)
          \partial_{s})\psi_\rr
           +\frac{1}{8}
\left(\partial_{s} \psi_\rr\right)^3
 +\frac{1}{4}\partial_{s}^3 \psi_\rr&=&0,
\label{eq:FGMKdV1r}
\\
(\partial_{u_{g-1},\ri}-\frac{1}{2}
(\lambda_{2g}+3b_0-\frac{3}{4}(\partial_s\psi_\rr)^2)
          \partial_{\ft})\psi_\rr
           +\frac{1}{8}
\left(\partial_{\ft} \psi_\rr\right)^3
 +\frac{1}{4}\partial_{\ft}^3 \psi_\rr&=&0.
\label{eq:FGMKdV1i}
\end{eqnarray}
If $\partial_{u_g,\rr} \psi_{\ri}=\partial_s \psi_\ri=\partial_\ft \psi_\rr=$ constant for the region $S^g H_\RR$, (\ref{eq:FGMKdV1r}) is further reduced to the FMKdV equation over $\RR$.
\end{theorem}

We consider that each $\kappa_\hX\gamma_i$ of $(\gamma_1, \ldots, \gamma_g)\in [S^g \tH_\RR]^0$ forms loops illustrated in Figure \ref{fg:Fig02}.

Theorem \ref{th:4.2} leads to the nice property that the conditions CI and CII in Condition \ref{cnd} are satisfied.
We explicitly describe the global property following \cite{Ma24b}.

\begin{theorem}\label{4th:reality_gga}
Let $(P_{a,0}=( \ee^{\ii\varphi_{a,0}}$, $K_{a,0}))_{a=1,\ldots,g}$ be a point in $\kappa_\hX[S^g \tH_\RR]^0$ where $\varphi_{a,0}\in [\varphi_\fb^{[a-]}, \varphi_\fb^{[a+]}]$, and set $\gamma_0 \in [S^g \tH_\RR]^0$ such that $\kappa_{\hX}\gamma_0=(P_{a,0})_{a=1,\ldots,g}$.
For $(t,s) \in \RR^2$,
\begin{equation}
  \begin{pmatrix} \varphi_1(t,s) \\ \varphi_2(t,s) \\ \vdots\\ \varphi_g(t,s) 
  \end{pmatrix}
 :=  \left(\int^t_0\cV_{g-2} d t + \int^s_0 \cV_g d s\right)+
\begin{pmatrix} \varphi_{1,0} \\ \varphi_{2,0} \\ \vdots\\ \varphi_{g,0} 
  \end{pmatrix}
\label{eq:Xipnt}
\end{equation}
forms $\gamma(t,s)\in S^g \tH_\RR$ such that $\gamma(0,0)=\gamma_0$. 
Then the image of $\gamma$,
\begin{equation}
\tSS_{\gamma_0}:=\{ \gamma(t,s) \ | \ (t,s) \in \RR^2\}\subset S^g\tH_\RR,
\label{eq:tS_varphi}
\end{equation}
provides a global solution of the FGMKdV equation in Theorem \ref{th:4.2}.
\end{theorem}

\begin{proof}
We follow the proof in Proposition 5.4 in \cite{Ma24b}, and Proposition 4.11 in \cite{Ma24d}.
Since $\gamma(t,s)=(\gamma_1, \ldots, \gamma_g)$ satisfies $\gamma_i \cap \gamma_j=\emptyset$ for $i \neq j$, we can simply integrate the orbits and find $\gamma \in S^g \tH_\RR$ for every $(t,s) \in \RR^2$.
\end{proof}


By letting $\LL_{\tv(\gamma_0)}:=\{ v_{t,s}:=\be_1 s +\be_3 t +v_0\ |\ (t,s)\in \RR^2\}\subset \CC^g$, we have the graph of $\tv^{-1}$ in $\CC^g\times S^g \tH_\RR$,
\begin{equation}
\cG(\tv^{-1}|\LL_{\tv(\gamma_0)})=\left\{ \left(v_{t,s},\gamma(t,s) \right)\ | \ (t,s) \in \RR^2\right\}\subset \LL_{\tv(\gamma_0)} \times \tSS_{\gamma_0} 
\label{eq:LL_tS}.
\end{equation}
At every point in the concerned subspace in $S^g\tH_\RR$, the meromorphic function $\psi$ satisfies the FGMKdV equations (\ref{eq:FGMKdV1r}) and (\ref{eq:FGMKdV1i}) as differential identities.
For a certain $\gamma =(\gamma_1, \gamma_2, \ldots, \gamma_g)\in \tSS_{\gamma_0}$, we can express it by using (\ref{eq:trans_Phi}), i.e., $\varphi_a:=\Phi(\gamma_a)$ for $a=1,2,\ldots, g$, and obtain a real valued one $\psi_{\rr}=2(\varphi_{1}+ \varphi_{2}+\cdots+\varphi_{g})$.
$\cG(\tv^{-1}|\LL_{\tv(\gamma_0)})$ shows a solution of the FGMKdV equation including the time-development in $t$.

As in the previous paper \cite{Ma24d}, we will restrict $\tSS_{\gamma_0}$, $\LL_{v_0}$, and $\cG(\tv^{-1}|\LL_{v_0})$ to the following subspaces with $t=0$ for simplicity, in order to obtain loops beyond Euler's figure eight as a snapshot of $t=0$.
Let $v_0:=\tv(\gamma_0)$.
\begin{equation}
\begin{split}
&\tS_{\gamma_{0}}:=\{ \gamma(s):=\gamma(0,s) \ | \ s \in \RR\},\\
&L_{v_0}:=\left\{v_s:=\be_1 s +v_0 \Bigr|\ s\in \RR\right\}, \\
&\cG(\tv^{-1}|L_{v_0}):=\left\{ (v_s,\gamma(s))\ | \ s \in \RR\right\}\subset \tS_{\gamma_{0}}\times L_{v_0}\subset S^g \tH_\RR \times \CC^g.
\label{eq:tS_varphi,t}
\end{split}
\end{equation}
Then using the graph space $(\gamma, \psi(\kappa_X(\gamma))) \in S^g \tH_\RR \times \CC$, we have a graph structure from (\ref{eq:Xipnt}),
\begin{equation}
\Psi_{t=0}:=\left\{ \left(s, \psi_\rr(\kappa_X(\gamma(t=0,s))\right)\ | \ s \in \RR\right\}\subset \RR\times \RR.
\label{eq:Psit}
\end{equation}

More precisely speaking, we have the third theorem of this paper.

\begin{theorem}\label{pr:solgMKdV}
Let $t=0$. For $\displaystyle{\begin{pmatrix}\varphi_{1}(0, s)\\ \varphi_{2}(0, s)\\ \vdots \\ \varphi_{g}(0, s)\end{pmatrix}}$ satisfying $\varphi_a \in [\varphi_\fb^{[a-]}, \varphi_\fb^{[a+]}]$ $(a = 1, \ldots, g)$ of a solution $\gamma(0,s)$ of the differential equation (\ref{eq:4.2}) in $S^g \hX$, we let $\psi_{\rr}(t=0, s)=2(\varphi_{1}(0,s)+ \varphi_{2}(0,s)$ $+\cdots+\varphi_{g}(0,s))$.
Then $\psi_\rr(t, s)$ at $t=0$ satisfies the FGMKdV equation (\ref{eq:FGMKdV1r}), where the gauge field is $A(t,s)=(\lambda_{2g}-3+\frac{3}{4}(\partial_{s}\psi_\ri(t,s))^2)/2$ given by 
\begin{equation}
\partial_{s}\psi_\ri(t=0, s)=(1,1,\ldots, 1)\cV_{g-1}.
\label{eq:g3CIIIi}
\end{equation}
\end{theorem}

\begin{proof}
The following is basically the same as the proof of Theorem 6.1 in \cite{Ma24b} and Theorem 4.12 in \cite{Ma24d}.
However, since this proof provides the algorithm for obtaining the global solutions for the hyperelliptic Riemann surface $\hX$, which is essential in Section 5.
We will show it briefly.

From Theorems \ref{th:4.2} and \ref{pr:solgMKdV}, for $du_{g,\rr} = ds$, we have $\psi_\rr(t,s)$ and $\partial_{s}\psi_\ri$ in (\ref{eq:g3CIIIi}).

We show that a global solution of (\ref{eq:4.2}) exists as an orbit in $S^g H_\RR^\varphi$. 
Let us prove this statement for the case of Figure \ref{fg:Fig01} (a) as follows.

We consider a primitive quadrature (\ref{eq:4.2}) for a real infinitesimal value $\delta s$ instead of $ds$ by letting $dt =0$ and $d\ft =0$.
We find $\delta \varphi_{a,\rr} =\cV_{g,a} \delta s$ for the real part of the entries of $\cK \cM$.
The orbit of $\varphi_a$ moves back and forth between the branch points $[\varphi_\fb^{[a-]}, \varphi_\fb^{[a+]}]\subset \hvarpi_x \kappa_H[S^g \tH_\RR]^0$ in (\ref{4eq:HEcurve_phi}) as in Figure \ref{fg:Fig01}.
Thus $\ee^{\ii \varphi_i}$ exists on the arc of the unit circle as in Figure \ref{fg:Fig01} (a) so that $(\ee^{\ii \varphi_a},  K_a)\in H_\RR$ draws a loop $\tgamma_i$ as in Figure \ref{fg:Fig01}.

Then we find an orbit in $S^g H_\RR^\varphi$ as a solution of (\ref{eq:4.2}) since $\cV_a$ is regular for disjoint $\varphi_a$, and $[\varphi_\fb^{[a-]}, \varphi_\fb^{[a+]}]$ and $[\varphi_\fb^{[b-]}, \varphi_\fb^{[b+]}]$ are disjoint if $a\neq b$ in $[S^g \tH_\RR]^0$.
Hence we prove it.
\end{proof}

\begin{remark}\label{rmk:4.10}
{\rm{
In this paper we restrict ourselves to $[S^g \tH_\RR]^0$.
However, there are other possibilities; we can directly extend our arguments for the cases $S^g \tH_{\ii \RR}$ which correspond to $u_{a,\ri}$ as in (\ref{4eq:gaugedMKdV2}).
Here $S^g \tH_{\ii \RR}$ is defined by $H_{\ii \RR}:=\{ (\ee^{2\ii \varphi}, K)\in H \ |\ K \in \ii\RR\}$ as in Definition \ref{Asmp2}.
Furthermore, we avoid the intersection of the paths $\gamma_i$ and $\gamma_j$, $(i \neq j)$ in this paper, but we could consider the intersection as argued in the previous paper \cite{Ma24b}.
Moreover, there are further possibilities as the real hyperelliptic solutions of the FGMKdV equations; for example, some of $e_j$ and $e_{j+1}=1/e_j$ can be real.
}}
\end{remark}

\begin{remark}\label{rmk:4.11}
{\rm{
Here we give some comments on condition III, $\partial_{u_g, \rr}\psi=$ constant.
The condition is now given by $(1,1,\ldots,1) \cV_{g-1}=$ constant.
This is realized as the vanishing of the meromorphic functions on $S^g \hX$ and thus on $\hJ_X$.
It might be obtained by the ratio of the sigma functions explicitly.
Thus, it should be written down more concretely and studied from an algebraic geometric point of view in the future.
}}
\end{remark}

\begin{remark}\label{rmk:4.12}
{\rm{
In our investigation, we have considered the differential relations given by the symmetric functions on $S^g \hX$.
There we deal with the derivatives of the quotient ring $\CC[\fc_1, \fs_1, K_1, \ldots,$ $ \fc_g, \fs_g, K_g]$ divided by $\fc_i^2+\fs_i^2 =1$, and (\ref{4eq:HEcurve_phi}).
Recently,  Buchstaber and Mikhailov investigated such systems as the Lie algebras of vector fields on universal bundles of symmetric product of curves, and as an integrable Hamiltonian system there \cite{BuMi1, BuMi2}.
These visions give a sophisticated interpretation from a modern mathematical point of view of the methods of Weierstrass \cite{Wei54} and Baker \cite{Baker03} on which we are based.
The rewrite may reveal the mathematical essence of our approach and provide a foundation between biophysics and modern mathematics, since this system is closely related to the shapes of supercoiled DNA via the excited states of Euler's elastica.
}}
\end{remark}

\begin{remark}\label{rmk:4.12}
{\rm{
Corresponding to (\ref{4eq:gaugedMKdV2}), we have (\ref{eq:FGMKdV1r}) and (\ref{eq:FGMKdV1i}) in Theorem \ref{th:4.2}.
Thus, solving the differential equation (\ref{eq:4.2}) with respect to $\ft$ also yields the solution to the second equation in (\ref{4eq:gaugedMKdV2}) or (\ref{eq:FGMKdV1i}), although this is not mentioned in the paper.
}}
\end{remark}

\section{Numerical computations for closed curves of $g=5$}

Following the algorithm written in the proof of Theorem \ref{pr:solgMKdV}, we will numerically evaluate the hyperelliptic solutions of the FGMKdV equation in this section and demonstrate a generalization of Euler's elastica as mentioned in Introduction.
We present the numerical results of genus five using the same notations and methods as in the previous paper \cite{Ma24d}.

To estimate the magnitude of the gauge field $\partial_{s} \psi_{\ri}$, we introduce 
$\psi_\ri^\circ:=$
$\displaystyle{\int^s \partial_{s} \psi_{\ri} ds}$ noting $ \partial_{s} \psi_{\ri}= \partial_{u_5\rr} \psi_{\ri}$.
Further, we consider a plane curve $Z(s)=$ $\displaystyle{\int_{s_0}^s \ee^{\ii \psi_\rr(t=0, s)} ds}$ by fixing $t$ and $s_0$.

We numerically integrate the differential equation (\ref{eq:4.2}) following the algorithm explained in the proof of Theorem \ref{pr:solgMKdV}.
We regard $Z$ as a generalization of Euler's elastica $C_Z:=\{(Z_\rr(s), Z_\ri(s)) \ |\ s \in [s_0, s_1]\}\subset \RR^2 =\CC$ because our evaluation of $g=1$ case recovers Euler's elastica \cite{MP22, Ma24d}.

By tuning the parameters $k_a$ of the hyperelliptic curves and the initial conditions of $\varphi_{a,0}$, we can deal with closed generalized elastica given by the FGMKdV equation.
As we show in the following, we found four closed curves $C_Z$ as examples.

As in \cite{Ma24c, Ma24d}, we have used the Euler's numerical quadrature method \cite{RLeV} for $\{(s, \gamma(0,s))
\ |\ s \ge 0\}\subset \RR \times S^5 \tH_\RR$ and 
$\{(s, \psi_\rr((\gamma(0,s))) \ |\ $ $ s \ge 0\}\subset \RR \times \RR$.
We will draw some graphs $C_Z$, $\{(s, \psi_\rr(t,s))\}_s$, 
$\{(s, \psi^\circ_\ri(t,s))\}_s$ and others based on the graph structure $\Psi_t$ in (\ref{eq:Psit}).

\medskip
\begin{table}[htb]
 \caption{Computational parameters of models A, B, C, D, E, F, and G }\label{tb:CompCond}
  \begin{tabular}{|c|cccc|ccc|}
\hline
& A& B& C& D &E &F &G\\
\hline
$k_1$ & 1.990000 & 2.659000 & 1.86805 & 1.90010 & 1.82000 & 2.016 & 1.8664 \\
$k_2$ & 1.989998 & 2.658000 & 1.86800 & 1.90000 & 1.80000 & 1.900 & 1.8663 \\
$k_3$ & 1.700780 & 1.880010 & 1.82000 & 1.82000 & 1.70000 & 1.820 & 1.8209 \\
$k_4$ & 1.700000 & 1.880000 & 1.81800 & 1.81000 & 1.68000 & 1.810 & 1.8189 \\
$k_5$ & 1.000100 & 1.000001 & 1.00650 & 1.01085 & 1.00001 & 1.026 & 1.0093\\
\hline
$\varphi_{1,0}$ &-0.518  &0.308  &-0.531  &0.460  &-0.503  &0.478  &-0.045 \\
$\varphi_{2,0}$  &0.527  &0.386  &0.565  &0.554  &0.589  &0.580  &0.565 \\
$\varphi_{3,0}$ &-0.527  &-0.386  &-0.565  &-0.554  &-0.589  &-0.556  &-0.565 \\
$\varphi_{4,0}$ &0.630  &0.588  &0.587  &1.400  &1.560  &1.310  &0.588 \\
$\varphi_{5,0}$  &-1.560  &-0.658  &-0.846  &-0.590  &-0.662  &-0.616  &-1.300 \\
\hline
  \end{tabular}
\end{table}
We demonstrate several models of A -- G including the cases that are not closed  and not differentiable closed cases; some of them are discussed in Section 6.
The moduli parameters $(k_a)$ of the hyperelliptic curve given and the initial conditions are listed in Table \ref{tb:CompCond}.

\bigskip
\subsection{Models A and B}

The first and the second results are displayed in Figures \ref{fg:shapeA} and \ref{fg:shapeB} as Models A and B.
Both are periodic solutions of FGMKdV equations from the numerical viewpoints.

Figures \ref{fg:shapeA} and \ref{fg:shapeB} (a) show $C_Z$, (b) display $\psi_\rr$ and $\psi_\ri^\circ$ while (c) show $\partial_s\psi_\rr$ and $\partial_s\psi_\ri$.
Since the value of $\partial_s\psi_\ri$ is not constant, $\psi_\rr(t=0,s)$ in both cases are solutions of the FGMKdV equation rather than the FMKdV equation.
Figures \ref{fg:shapeA} and \ref{fg:shapeB} (d) show $\varphi =\psi_\rr/2 +\varphi_0$ and $(\varphi_a)_{a=1,2,\ldots,5}$, which display the contributions of $\varphi_a$ to $\varphi$.
Since the SMKdV equation (\ref{4eq:SMKdV_k}) is a special case of the FGMKdV equation (\ref{eq:FGMKdV1r}), we can regard that Figure~\ref{fg:shapeA} (a) is a generalization of Euler's elastica in Figure \ref{fg:shape1744_24} (a).

Figures \ref{fg:shapeA} and \ref{fg:shapeB} numerically show that $Z(s+\ell) = Z(s)$, $\partial_s Z(s+\ell) = \partial_s Z(s)$, $\psi_\rr(s+\ell) = \psi_\rr(s)$ and $\partial_s\psi_\rr(s+\ell) = \partial_s\psi_\rr(s)$ in both cases.
Since the tangential vector of the real plane curve $C_Z$ is given by $\partial_s Z=\ee^{\ii \psi}$, they mean that the curve $C_Z$ in Figure \ref{fg:shapeA} and Figure \ref{fg:shapeB} are numerically continuous by the second order differentiable with respect to $s$; $\partial_s^a Z(s+\ell) = \partial_s^a Z(s)$, $(a=0, 1, 2)$.
Figures \ref{fg:shapeA} and \ref{fg:shapeB} (d) show that the sudden jumps in $\partial_s\psi_\rr(s)$ in (c) are caused by a turning of some $\varphi_a$ around the branch point.
Although $K$ in Section 3 vanishes at the branch point, we only handle holomorphic one-forms. Thus, the matrices appearing in Section 3 are regular.
This implies that the jumps in $\partial_s\psi_\rr(s)$ in (c), are not non-differentiable behavior, though we must analyze it more carefully in the future.

Further,  each $\varphi_a$ in (d) are also periodic from a numerical viewpoint.
Hence it is expected that both models A and B are periodic as regular functions in $s$.

Furthermore, defining the $\mathrm{index}[Z]:=$
$\displaystyle{\frac{1}{2\pi} \int_0^\ell \partial_s \varphi\, d s}$, we have
$\mathrm{index}[Z] =0$ for Models A and B.

\begin{figure}
\begin{center}

\includegraphics[angle=90,height=0.3\hsize]{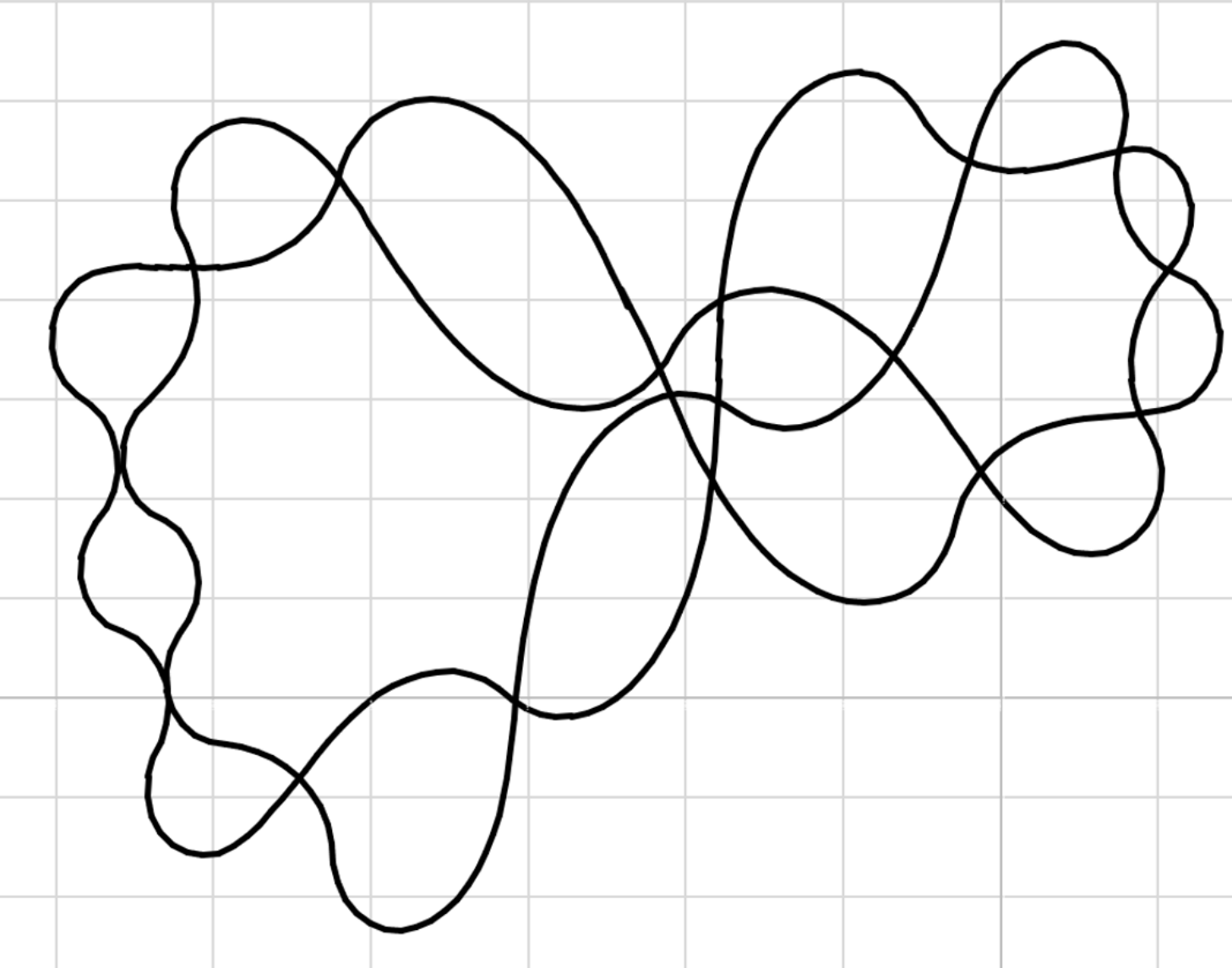}
\hskip 0.05\hsize
\includegraphics[height=0.2\hsize,]{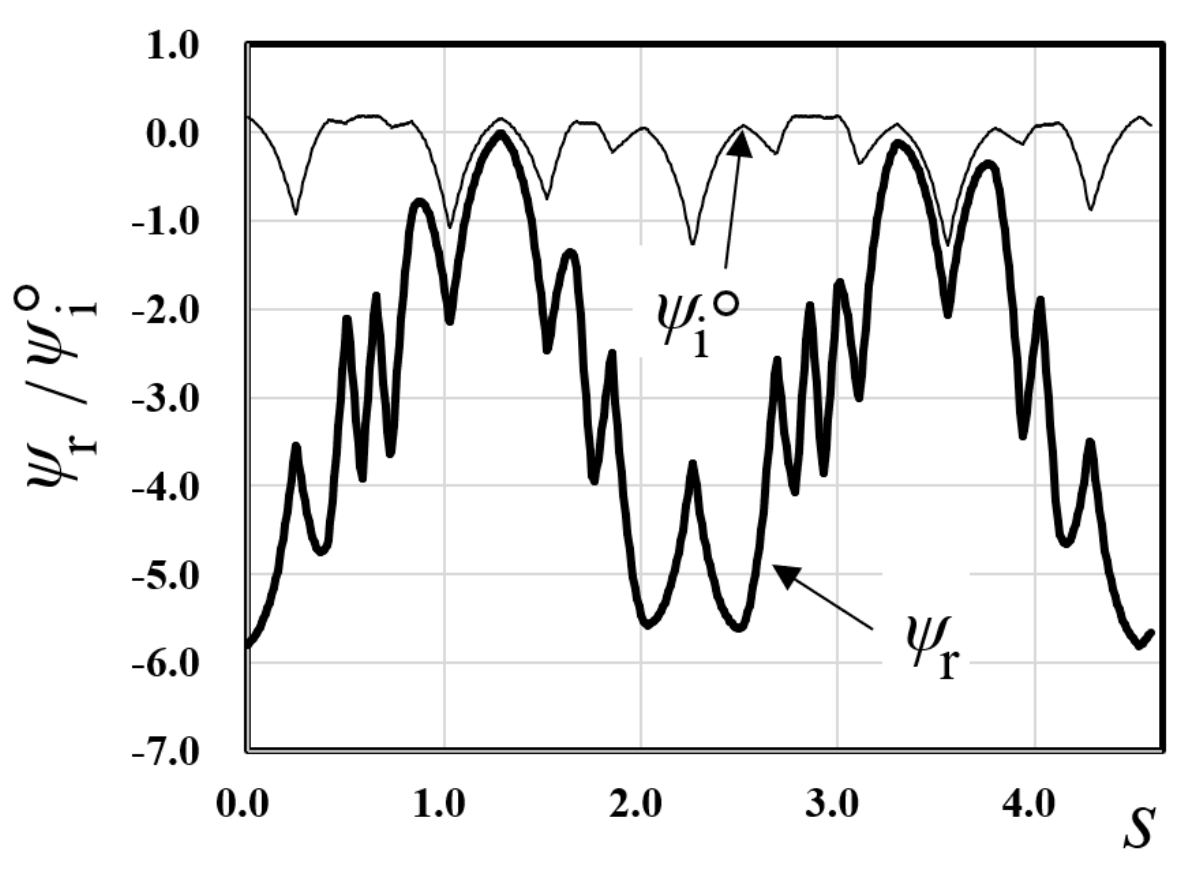}

(a)\hskip 0.4\hsize
(b)

\smallskip

\includegraphics[height=0.2\hsize,]{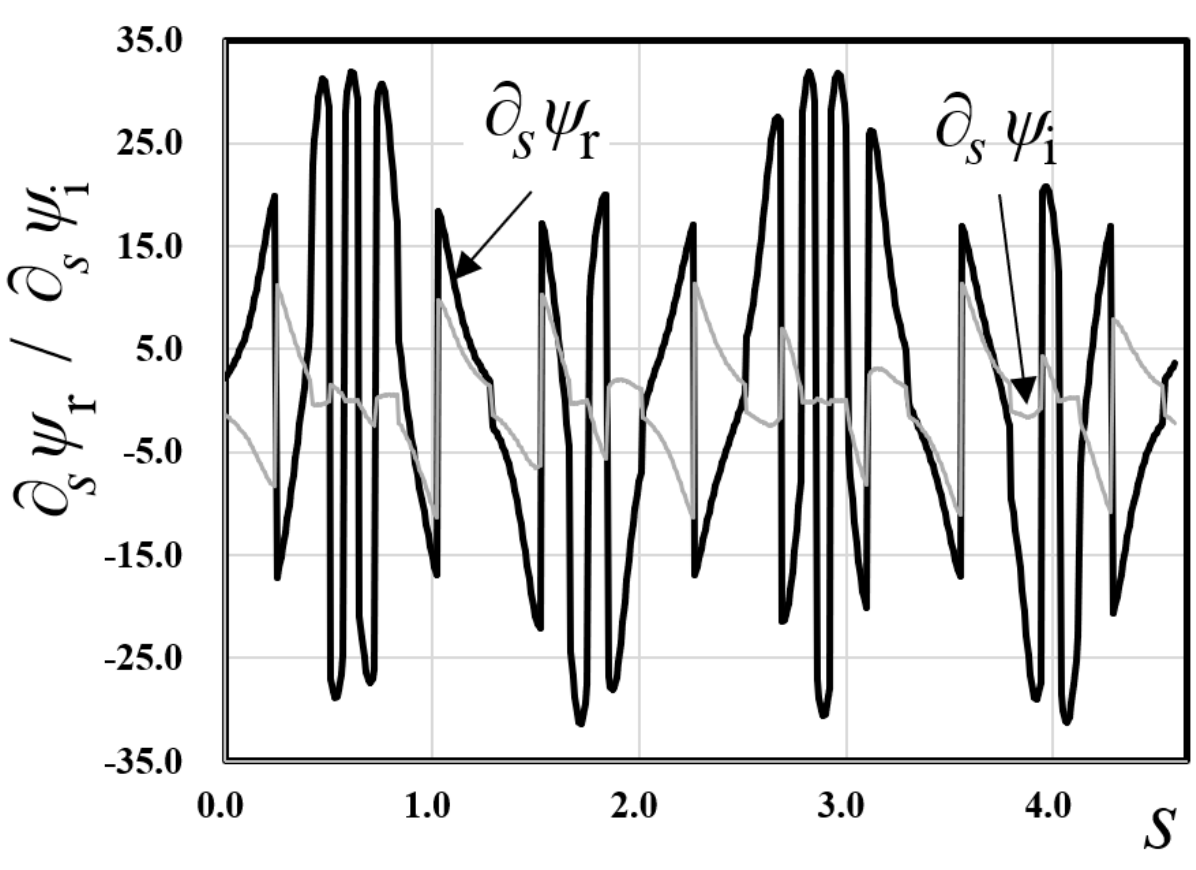}
\hskip 0.05\hsize
\includegraphics[height=0.2\hsize,]{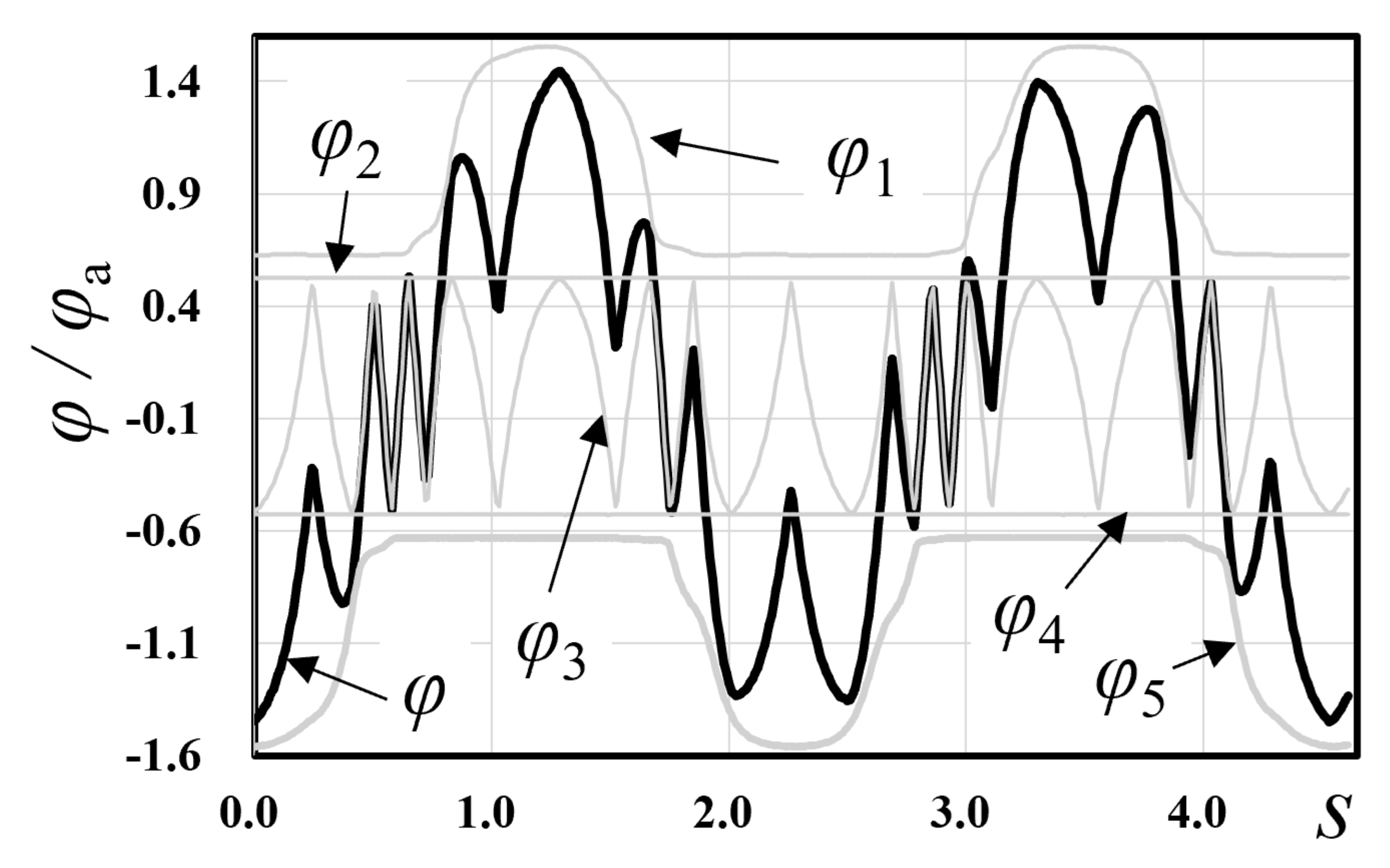}

(c) \hskip 0.4\hsize
(d)
\end{center}

\caption{
The solutions of (\ref{eq:Xipnt}) for $t=0$ of model A:\ 
(a): $C_Z$, (b): $\psi_\rr$ and $\psi_\ri^\circ$
and (c): $\partial_s\psi_\rr$, $\partial_s\psi_\ri$, and $\varphi=\psi_\rr/2-\varphi_0$ with $(\varphi_a)_{a=1,\ldots,5}$.
}\label{fg:shapeA}
\end{figure}

\bigskip

\begin{figure}
\begin{center}

\includegraphics[height=0.3\hsize]{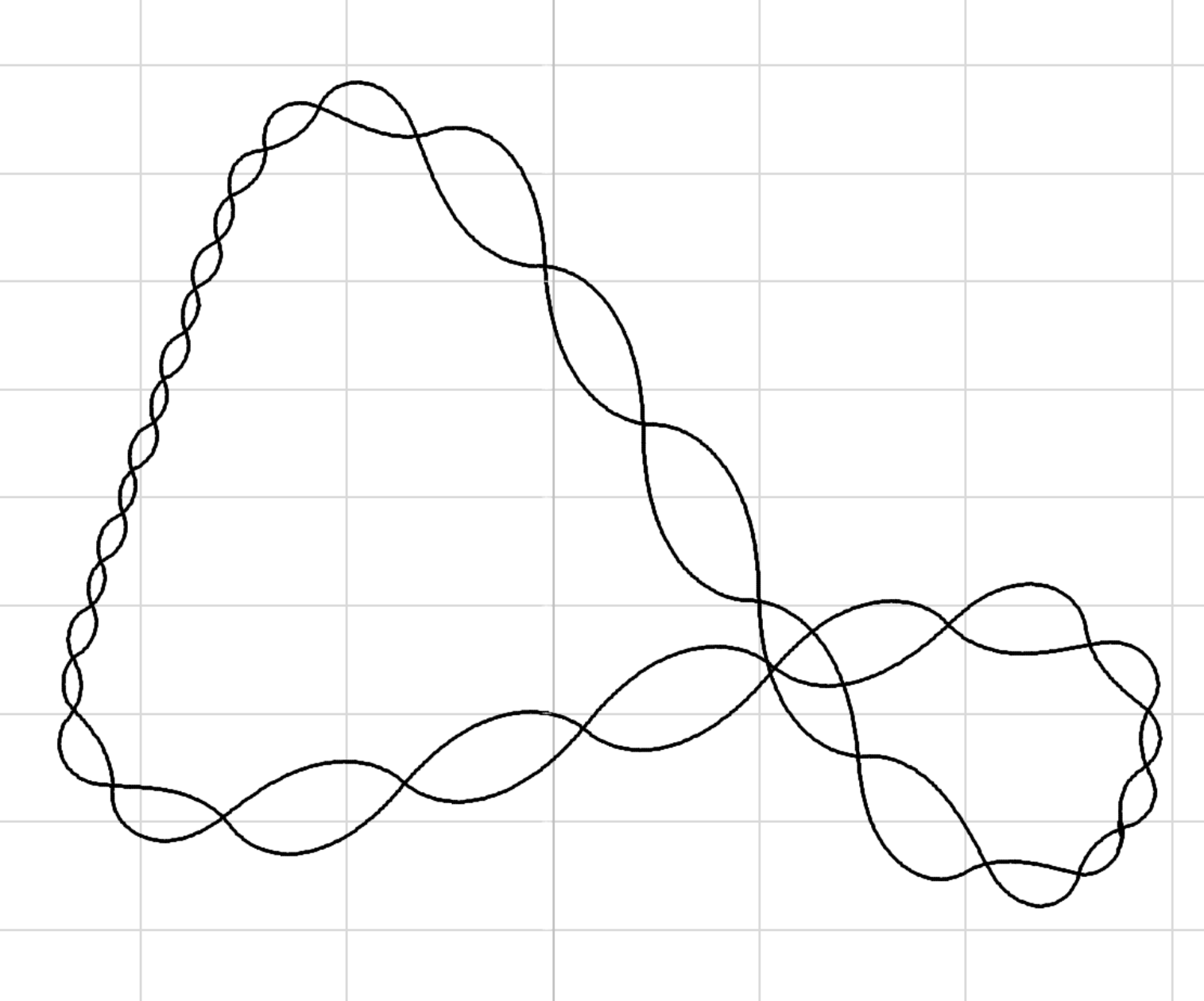}
\hskip 0.05\hsize
\includegraphics[height=0.17\hsize,]{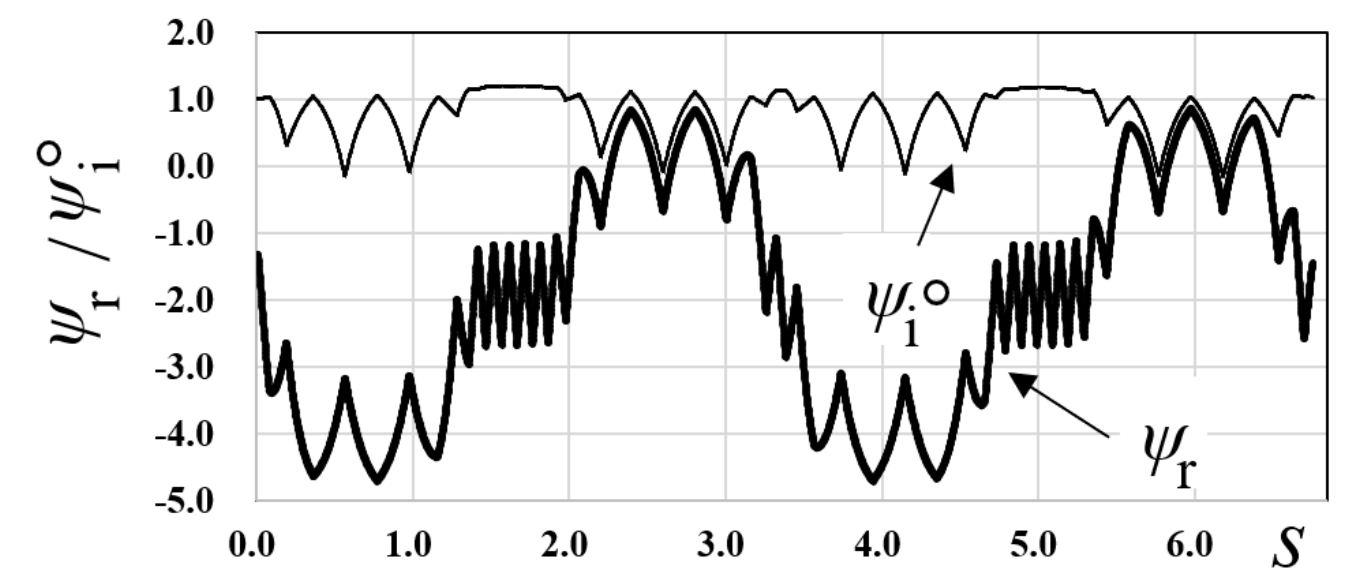}

(a)\hskip 0.4\hsize
(b) 

\smallskip

\includegraphics[height=0.17\hsize,]{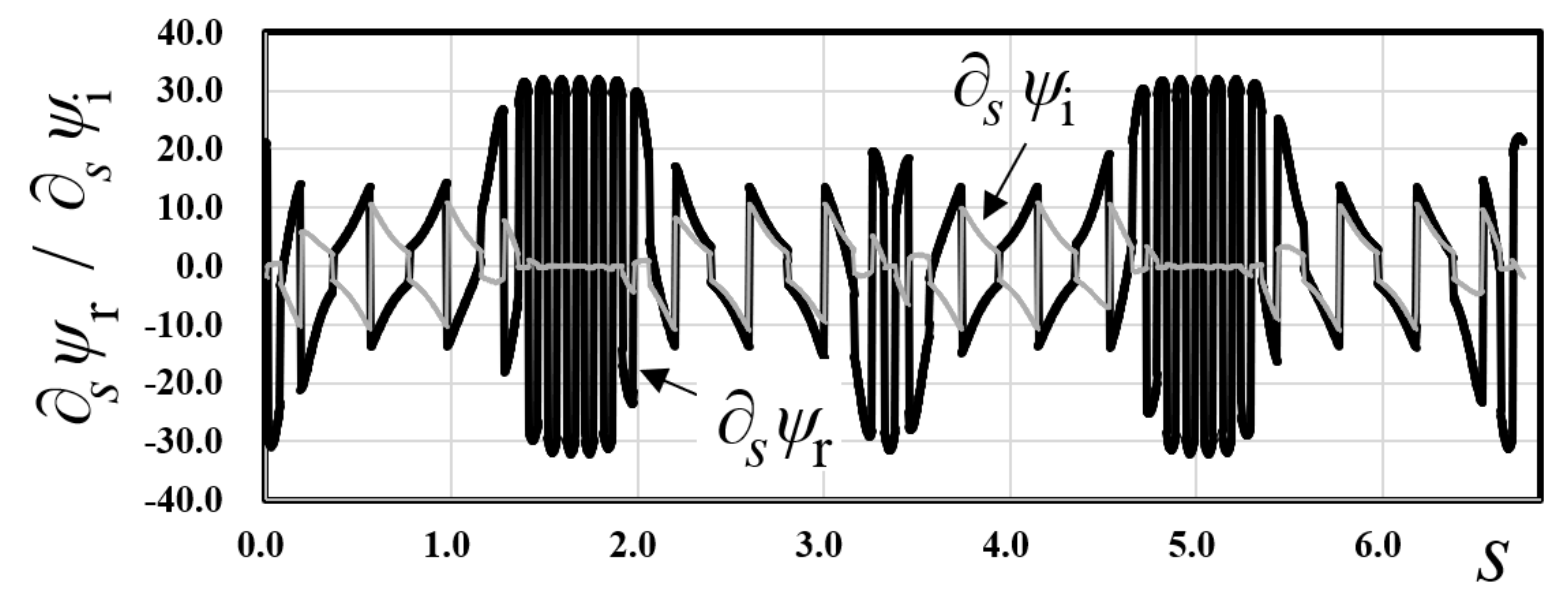}
\hskip 0.05\hsize
\includegraphics[height=0.17\hsize,]{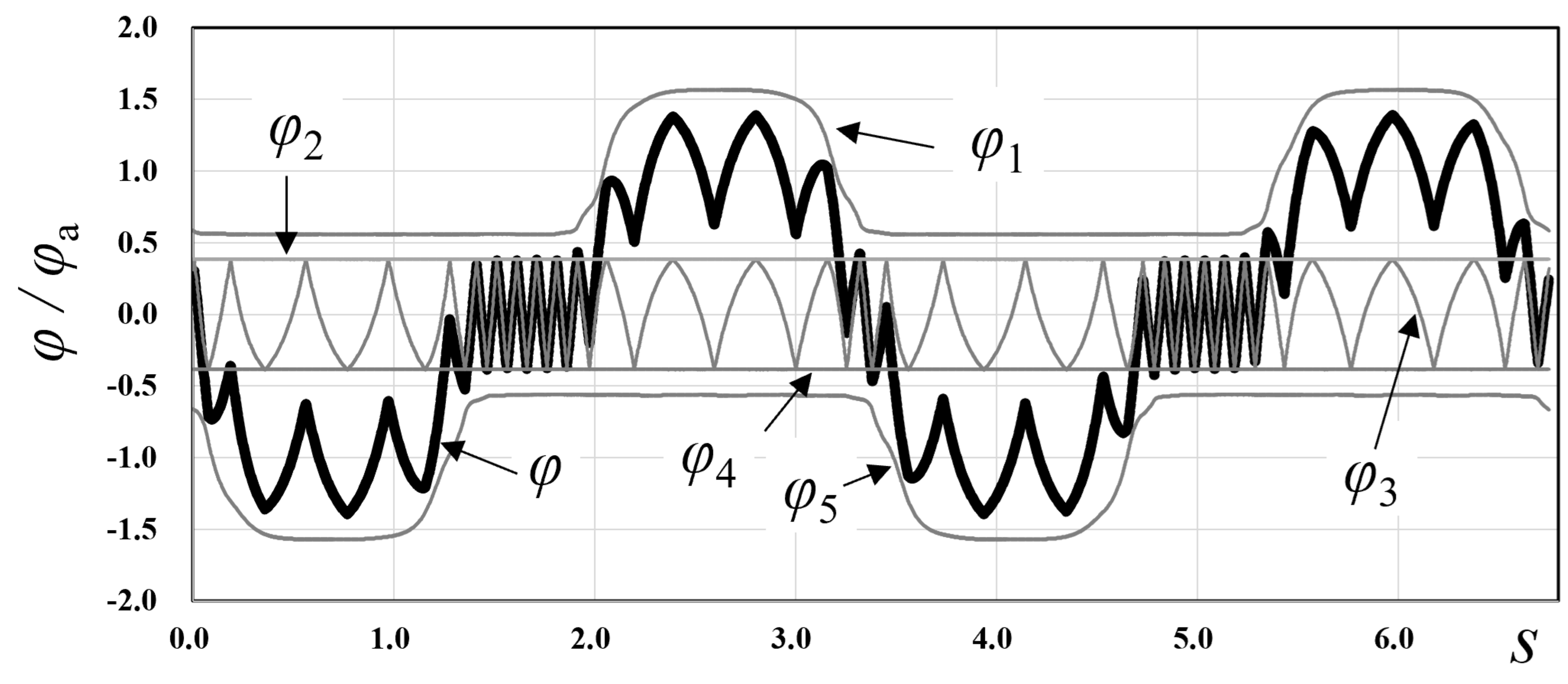}

(c) \hskip 0.4\hsize
(d)

\end{center}

\caption{
The solutions of (\ref{eq:Xipnt}) for $t=0$ of model B:\ 
(a): $C_Z$, (b): $\psi_\rr$ and $\psi_\ri^\circ$
and (c): $\partial_s\psi_\rr$, $\partial_s\psi_\ri$, and $\varphi=\psi_\rr/2-\varphi_0$ with $(\varphi_a)_{a=1,\ldots,5}$.
}\label{fg:shapeB}
\end{figure}

\subsection{Models C and D}

The third and fourth results are illustrated in Figures \ref{fg:shapeC} and \ref{fg:shapeD} as Models C and D.
When we tuned the parameters, we consider only the condition that $Z(s+\ell) = Z(s)$, and $\partial_s Z(s+\ell) = \partial_s Z(s)$ since $\partial_s Z(s) =\ee^{\ii \psi_\rr(s)}$.
Figures \ref{fg:shapeC} and \ref{fg:shapeD} (a) show these $Z$'s.
As in Figure \ref{fg:shapeC} and \ref{fg:shapeD} (b) and (c) show $\psi_\rr$, $\psi^\circ_\ri$. $\partial_s \psi_\rr$ and $\partial_s\psi_\ri$.
Since the value of $\partial_s\psi_\ri$ is neither constant, $\psi_\rr(t=0,s)$ is also a hyperelliptic solutions of the FGMKdV equation rather than the FMKdV equation.

\begin{figure}
\begin{center}

\includegraphics[height=0.19\hsize]{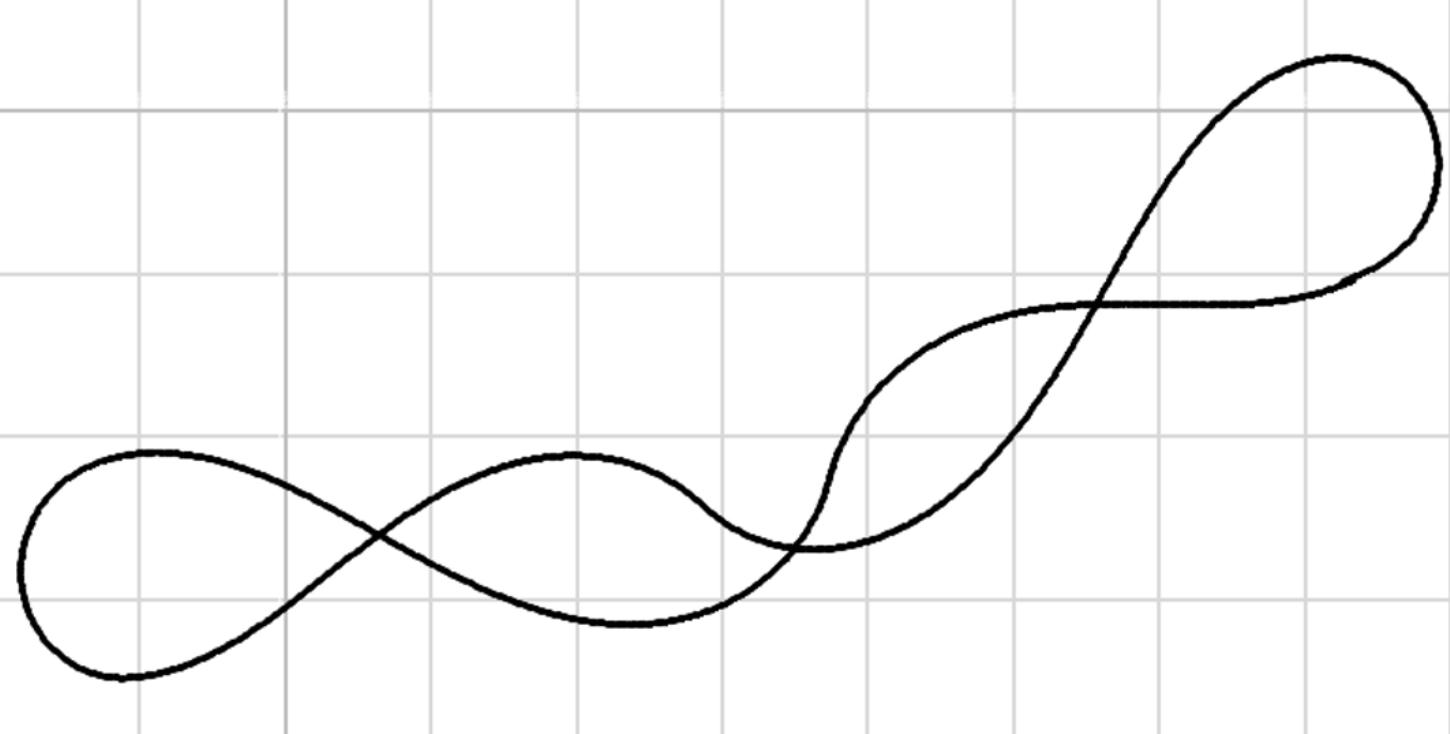}
\hskip 0.05\hsize
\includegraphics[height=0.18\hsize,]{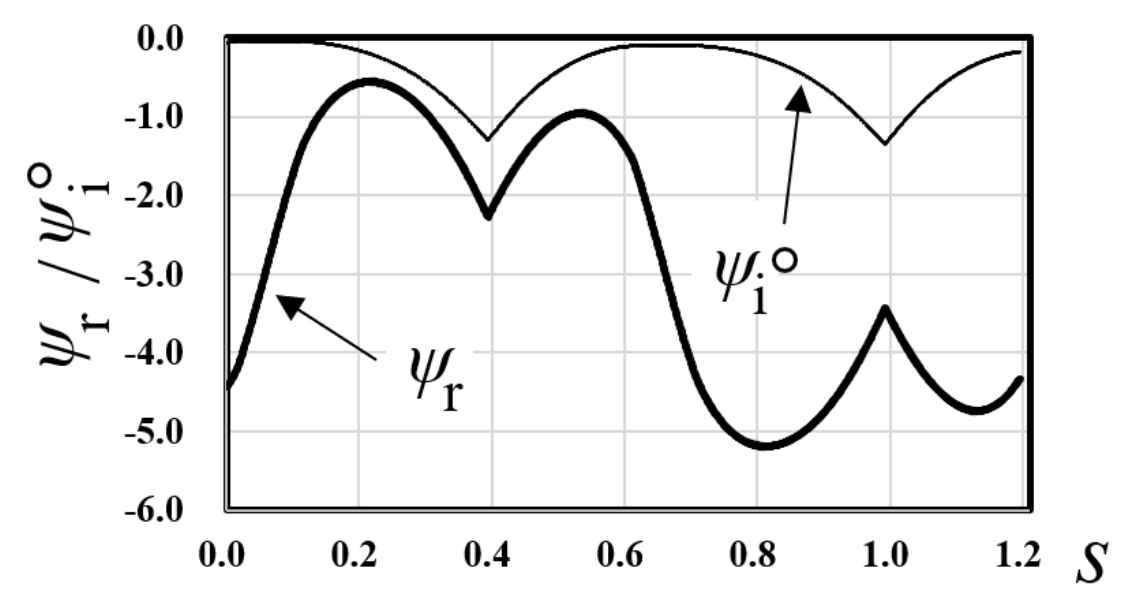}

(a)\hskip 0.4\hsize
(b)

\smallskip

\smallskip

\includegraphics[height=0.18\hsize,]{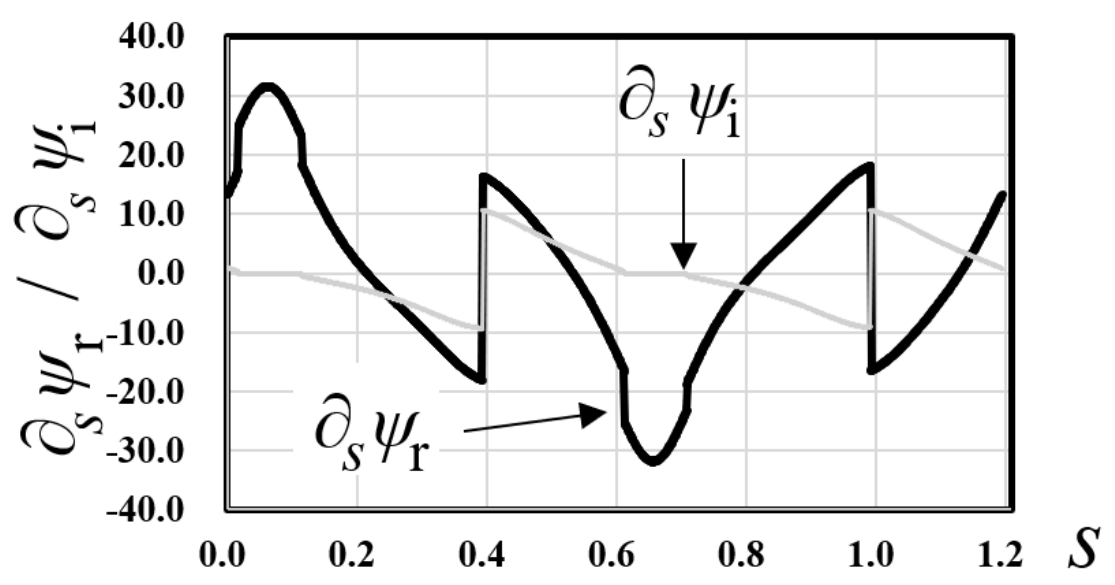}
\hskip 0.05\hsize
\includegraphics[height=0.18\hsize,]{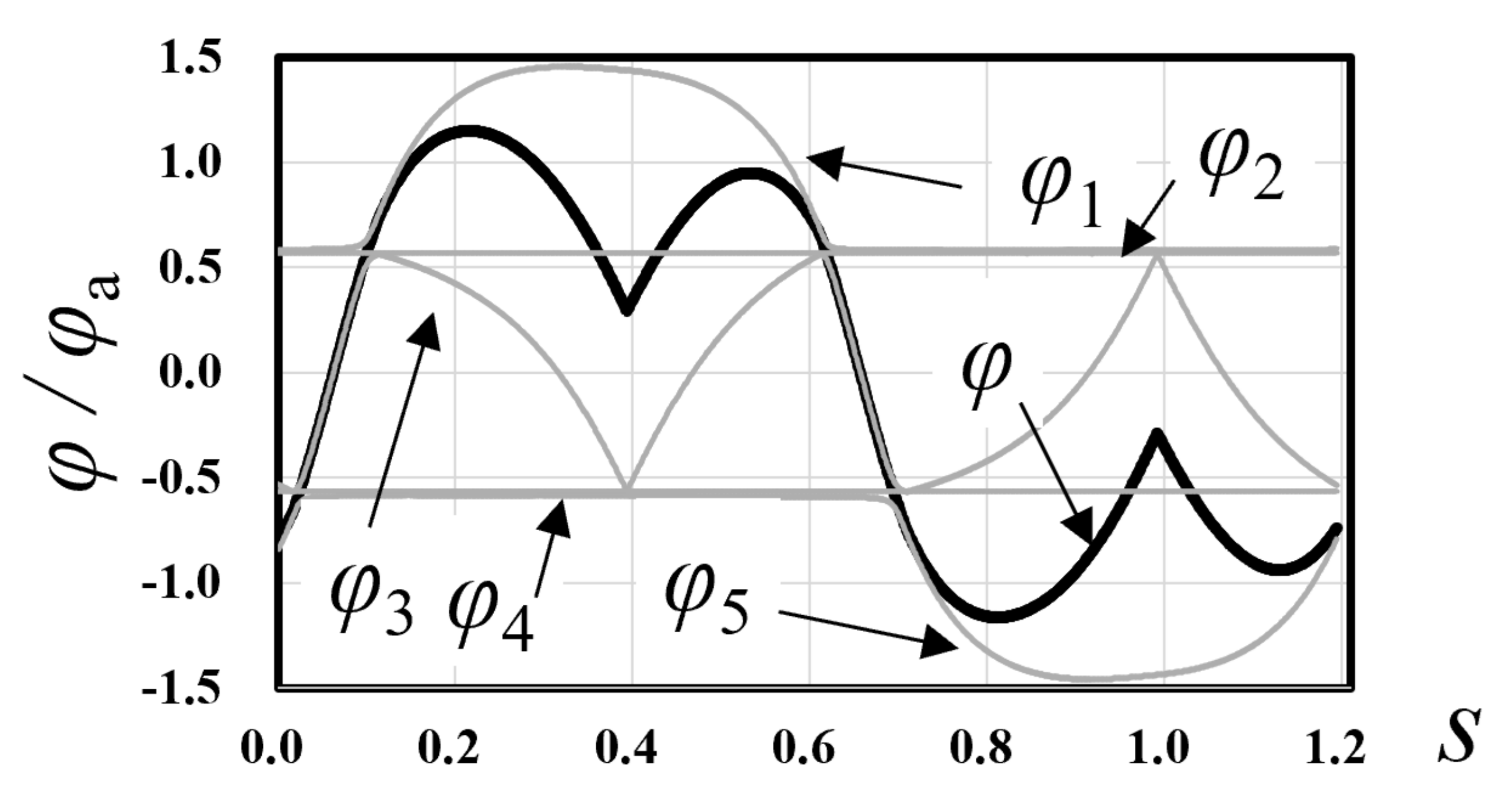}

(c) \hskip 0.4\hsize
(d)
\end{center}

\caption{
The solutions of (\ref{eq:Xipnt}) for $t=0$ of model C:\ 
(a): $C_Z$, (b): $\psi_\rr$ and $\psi_\ri^\circ$
and (c): $\partial_s\psi_\rr$, $\partial_s\psi_\ri$, and $\varphi=\psi_\rr/2-\varphi_0$ with $(\varphi_a)_{a=1,\ldots,5}$.
}\label{fg:shapeC}
\end{figure}

\begin{figure}
\begin{center}

\includegraphics[height=0.25\hsize]{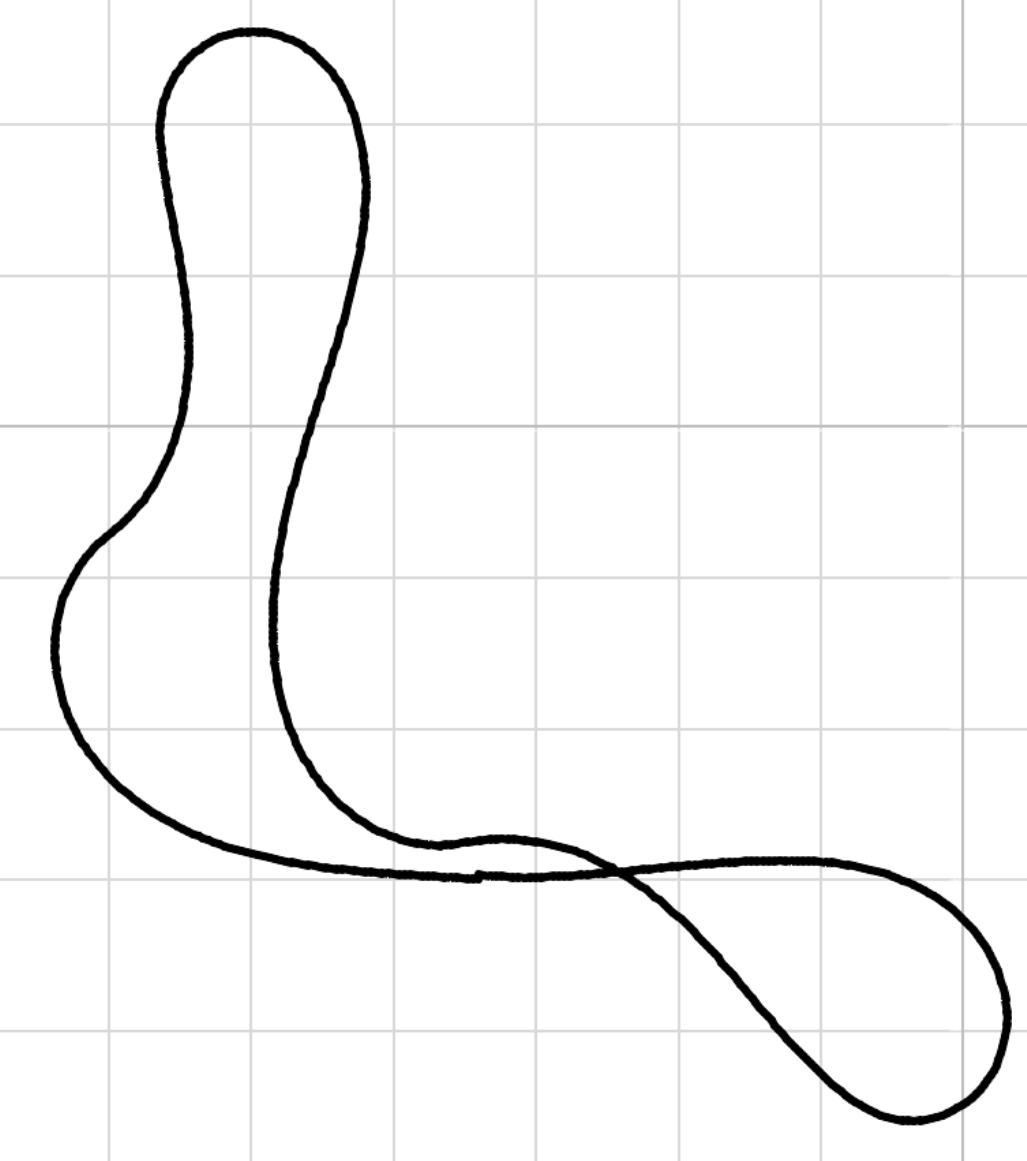}
\hskip 0.05\hsize
\includegraphics[height=0.18\hsize,]{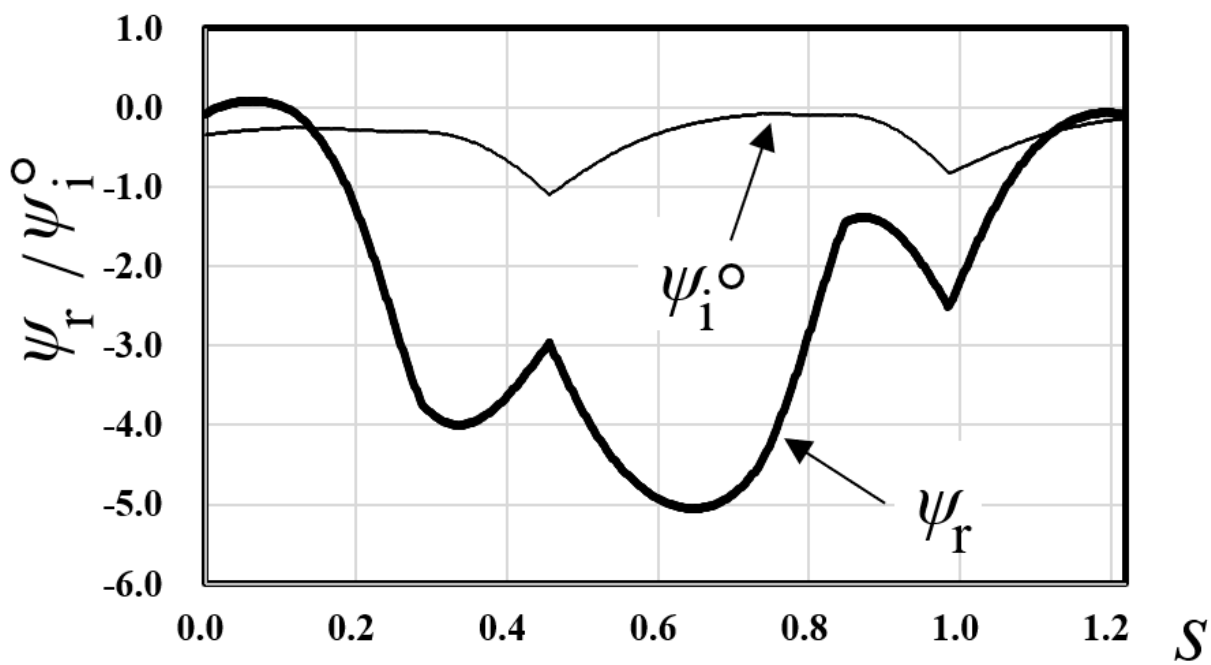}

(a)\hskip 0.4\hsize
(b) 
\smallskip

\includegraphics[height=0.18\hsize,]{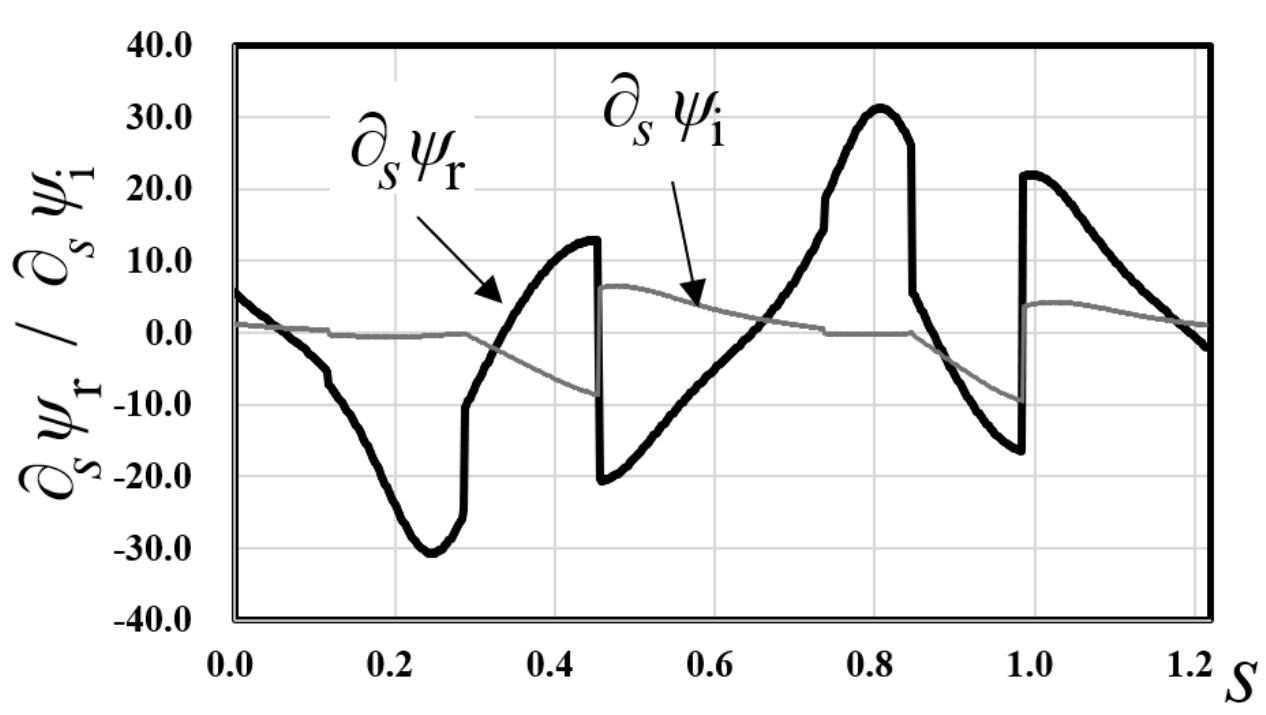}
\hskip 0.05\hsize
\includegraphics[height=0.18\hsize,]{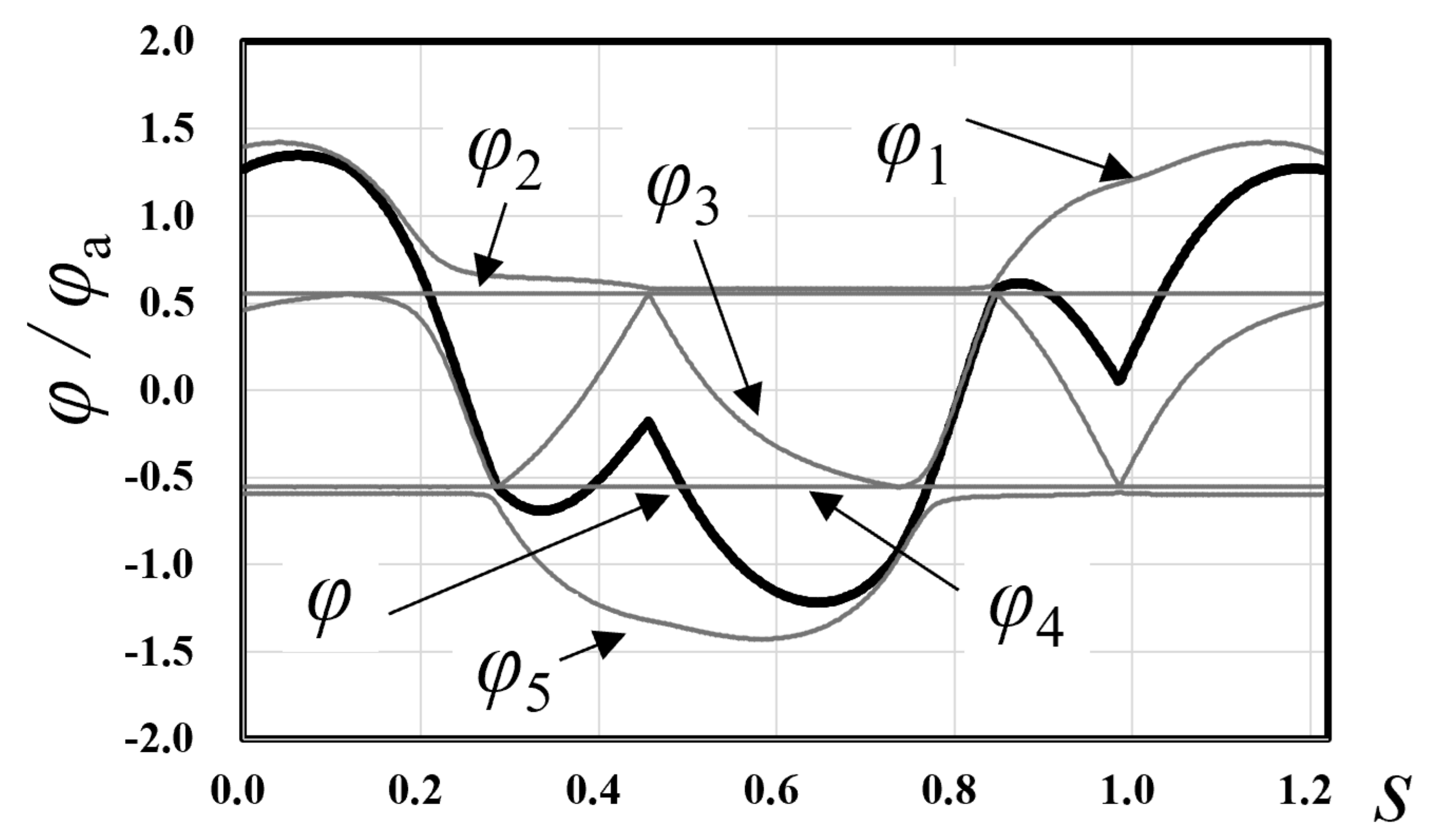}

(c) \hskip 0.4\hsize
(d)
\end{center}

\caption{
The solutions of (\ref{eq:Xipnt}) for $t=0$ of model D:\ 
(a): $C_Z$, (b): $\psi_\rr$ and $\psi_\ri^\circ$
and (c): $\partial_s\psi_\rr$, $\partial_s\psi_\ri$, and $\varphi=\psi_\rr/2-\varphi_0$ with $(\varphi_a)_{a=1,\ldots,5}$.
}\label{fg:shapeD}
\end{figure}

Figures \ref{fg:shapeC} and \ref{fg:shapeD} (d) also show $\varphi =\psi_\rr/2 +\varphi_0$ and $(\varphi_a)_{a=1,2,\ldots,5}$, which display the contributions of $\varphi_a$ to $\varphi$.
We can also regard that Figure~\ref{fg:shapeC} and \ref{fg:shapeD} (a) are also a generalization of Euler's elastica in Figure \ref{fg:shape1744_24} (a) and its generalization of genus three case Figure \ref{fg:shape1744_24} (b).

Figure \ref{fg:shapeC} (a), (b) and (c) show that the curve $C_Z$ in Figure \ref{fg:shapeC} is numerically continuous by the second order differential with respect to $s$; $\partial_s^a Z(s+\ell) = \partial_s^a Z(s)$, $(a=0, 1, 2)$ or 
$Z(s+\ell) = Z(s)$ and $\partial_s^a\psi_\rr(s+\ell) = \partial_s^a\psi_\rr(s)$  $(a=0, 1)$.
This contrasts with the genus three case in Figure \ref{fg:shape1744_24} (c) and (d), which is not continuous by the second order derivative with respect to $s$ as in \cite{Ma24d}; $\partial_s \psi_\rr(s+\ell) \neq \partial_s \psi_\rr(s)$.
Thus although Figures \ref{fg:shape1744_24} (b) and Figure \ref{fg:shapeC} (a) are similar, they are different.
The augmented number of parameters inherent to the results of the genus five in enhanced representations of shapes and an elevated degree of continuity when compared with the genus three.

However, Figure \ref{fg:shapeC} (d) shows that each $\varphi_a$ is not periodic  slightly and thus it does not corresponds to the periodic solution of the FGMKdV equation, though they are periodic $\partial_s^a \psi_\rr(s+\ell) = \partial_s^a \psi_\rr(s)$, $(a=0, 1)$.

Figure \ref{fg:shapeD} numerically shows that $Z(s+\ell) = Z(s)$, $\partial_s Z(s+\ell) = \partial_s Z(s)$, and $\psi_\rr(s+\ell) = \psi_\rr(s)$ but 
$\partial_s \psi_\rr(s+\ell) \neq \partial_s \psi_\rr(s)$;
Thus $C_Z$ in Figure \ref{fg:shapeC} is numerically continuous by the first order differential with respect to $s$;  $\partial_s^a Z(s+\ell) = \partial_s^a Z(s)$, $(a=0, 1)$.

However, as we have the parameters ($(k_a)_{a=1,2,\ldots,g}$ and the initial condition $(\varphi_a|_{s=0})_{a=1,2,\ldots.g}$),  we have several more complicated closed loops of the genus five as the generalized elastica than the previous paper of the genus three \cite{Ma24d}.
Higher genera are expected to produce more interesting curves.

Furthermore, we note that $\mathrm{index}(Z)=0$ for both Models C and D.

\begin{remark}
{\rm{
The elastica problem can be interpreted as geodesics in the sub-Riemannian geometry on the group of rigid motions of the plane, which is known as {\lq\lq}bicycling mathematics{\rq\rq}\cite{ABLMS}.
Even in the generalized elastica problem, the computation of $(\varphi_a)_{a=1, \ldots, g}$ may also be related to this new interpretation.
}}
\end{remark}

\section{Discussion and Conclusion}
By extending the constructions in \cite{Ma24b,Ma24d}, in this paper, we showed a novel real algebro-geometric method to obtain the hyperelliptic solutions of general genera $g$ of the FGMKdV equation (\ref{4eq:gaugedMKdV2}) as in Theorems \ref{th:4.2} and \ref{4th:reality_gga}.
As we introduced the real parameters $t$ in the Jacobian $J_X$ in (\ref{eq:dus_dts}), we showed that these new parameters $t_g$ and $t_{g-2}$ provide the correspondence between the real vector space $\LL_{\tv(\gamma_0)}$ in $J_X$ and real $\varphi$'s in $S^g \tH_\RR \subset S^g \hX$.
Since the FGMKdV equation (\ref{4eq:gaugedMKdV2}) is a differential identity on $S^g \hX$, the correspondence means the construction of the real hyperelliptic solutions of the FGMKdV equation.
In the correspondence, the shifted elementary symmetric polynomials $\varepsilon$ in Definition \ref{def:varepsilon} play the essential role, since even on the angle expression, the correspondence basically comes from the properties of the Vandermonde matrices as shown in Lemma \ref{4lm:KdV1}.
(We describe the properties of the shifted elementary symmetric polynomials in Appendix.)

In the construction we have used the data of the hyperelliptic curves $X$ directly instead of the Jacobian $J_X$. 
We note that our algebraic study of the algebraic curves on two decades \cite{BEL97b, MP22, M25} based on studies by  Weierstrass \cite{Wei54} and Baker \cite{Baker03} allows the such treatment.

\bigskip

Section 5 demonstrated the closed plane curves as the hyperelliptic solutions of the FGMKdV equation of genus five as an extension of \cite{Ma24d}.
They are regarded as a generalization of Euler's eight-figure and Figure \ref{fg:shape1744_24} (a) and its generalization of genus three case Figure \ref{fg:shape1744_24} (b).
There we have more complicated and continuous in the higher degree than the genus three case of \cite{Ma24d}.
It is expected that due to the dimension of parameter space, the higher the genus, the greater the ability of the solution to express complicated shapes.
As we obtain the hyperelliptic solutions of the FGMKdV equation in Theorems \ref{th:4.2} and \ref{4th:reality_gga}, it is not difficult to provide the numerical evaluations of the generalization of elastica of even higher genus $g> 5$ following Theorem \ref{pr:solgMKdV}.
Section 5 shows that the extension of \cite{Ma24d} and Euler's figure-eight can be obtained by the algorithm in this paper.

As the work with Previato \cite{MP16} showed that the FMKdV flows provide the isometric deformation with isoenergy of the Euler-Bernoulli energy $E[Z]$ of plane curves $C_Z$, and such flows are so important in differential geometry as in \cite{CKPP}, the ultimate purpose of this study is to find the real solution of the FMKdV equation of higher genus explicitly.
It is also closely related to the fascinating relation between shapes of supercoiled DNA and the integrable system as mentioned in Introduction and in \cite{M24a};
The FMKdV flows must reproduce the class of excited states of elastica whose energy is the same value, which was also numerically demonstrated as the discrete segment model of the supercoiled DNA in the Monte-Carlo simulation by Vologodskii and Cozzarelli \cite{VC}.

We only show that the closed plane curves coming from the FGMKdV equation rather than the FMKdV equation.
Since the condition CIII that the gauge field $\partial_s \psi_{\ri}$ is constant seems to require more higher genus, our results may step to reveal the properties of the conditions as mentioned in Remark \ref{rmk:4.11}.

\begin{figure}
\begin{center}

\includegraphics[width=0.6\hsize]{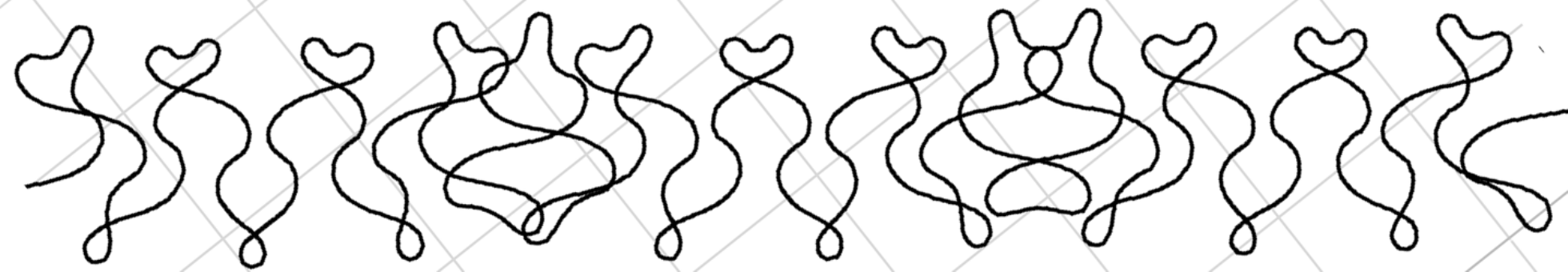}

(a)

\includegraphics[width=0.6\hsize]{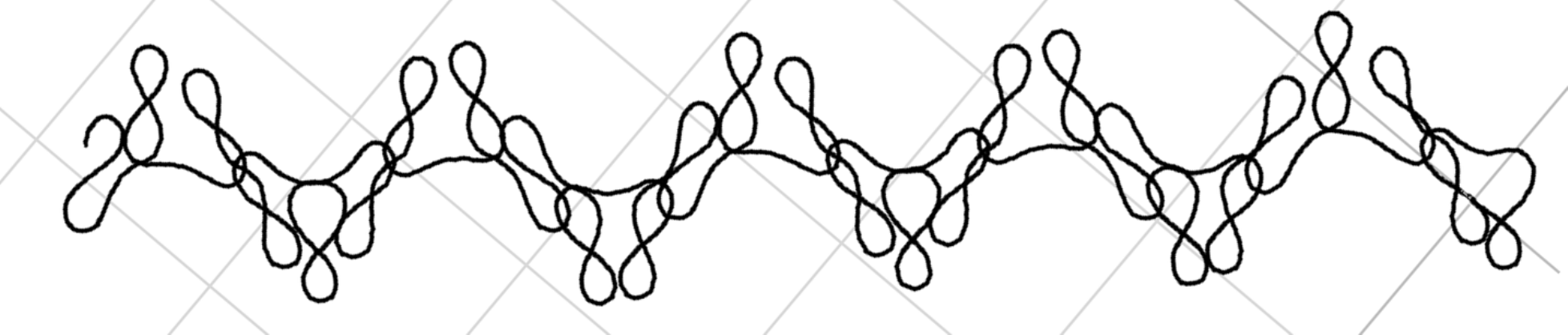}

(b)

\smallskip

\end{center}

\caption{
The generalized elasticae of Models E and F.
}\label{fg:shape_etc}
\end{figure}

\bigskip

\begin{figure}
\begin{center}

\includegraphics[height=0.2\hsize]{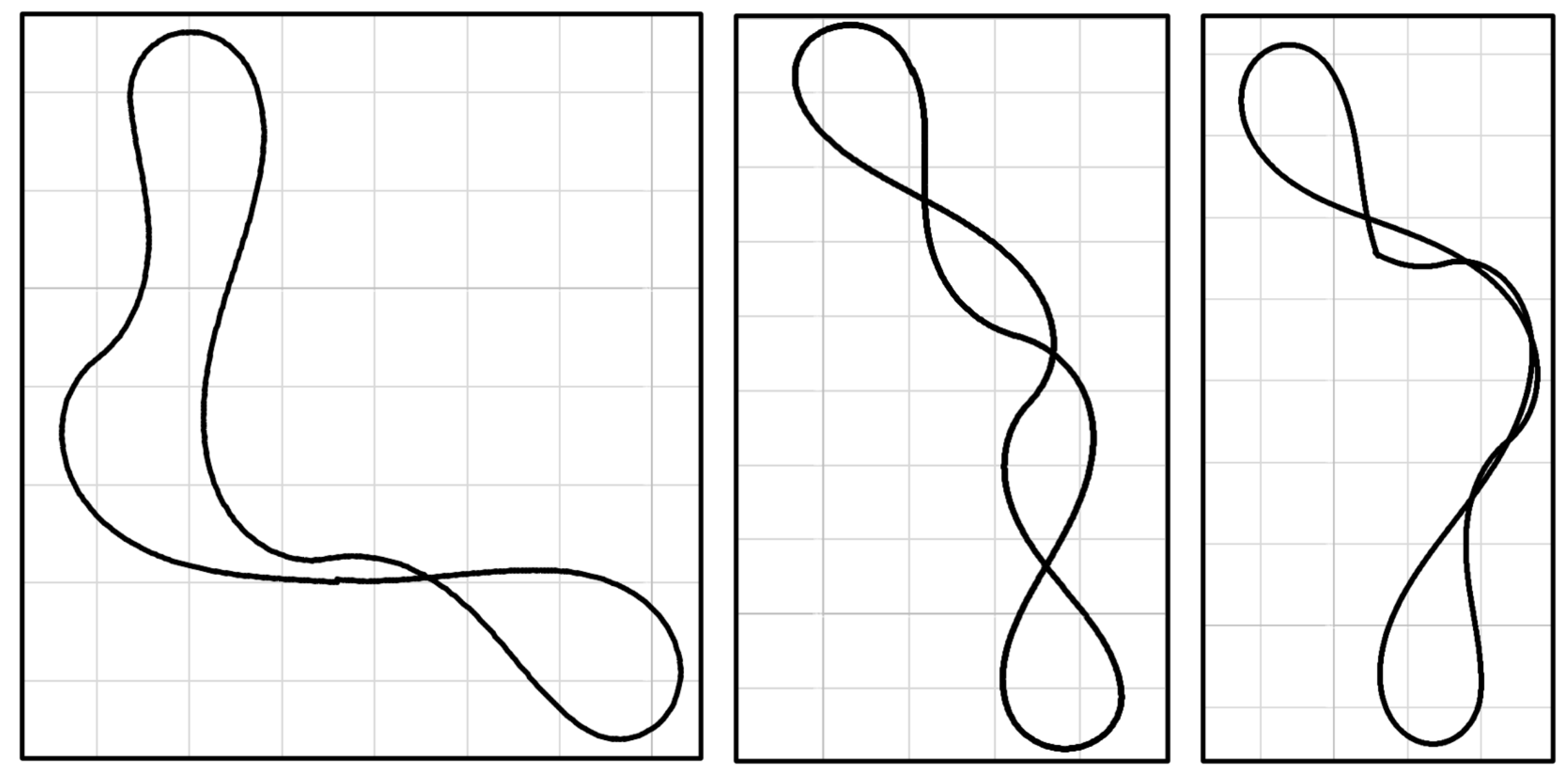}
\hskip 0.05\hsize
\includegraphics[height=0.19\hsize]{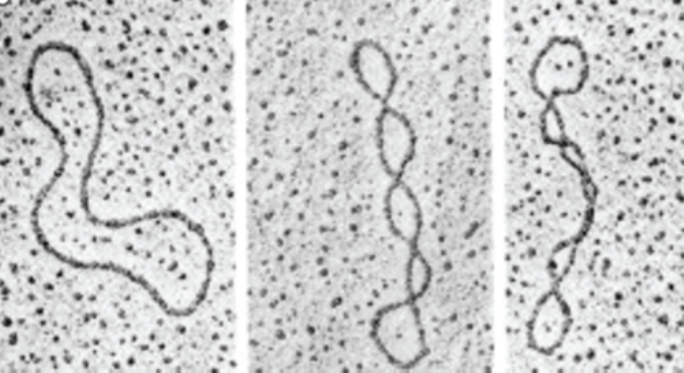}

(a)
\hskip 0.4\hsize
(b)
\hskip 0.4\hsize

\smallskip

\end{center}

\caption{
The shapes of Models C, D and G, and electron micrographs of DNA (mini ColE1 plasmid dimer, 5kb) by Laurien Polder  \cite[Fig.~1-24 p.~36]{KornbergBaker}36], Reproduced with permission from Roger Kornberg.
}\label{fg:shape_vsDNA}
\end{figure}

\bigskip

However, although our computations in Section 5 do not satisfy condition CIII, our results resemble the shapes of supercoiled DNAs.
Further, if we give up the closed condition, we obtain interesting shapes, as shown in Figure \ref{fg:shape_etc} for Models E and F.
Moreover, if we allow the curve with a not differentiable point, we have the right-hand side in Figure \ref{fg:shape_vsDNA} (a) as a closed plane curve of the FGMKdV equation, as shown in Model G.

Even though there are many differences between them, Figures \ref{fg:shape_vsDNA} (a) and (b) ((b) is an electron micrograph of DNA in  \cite[Fig.~1-24 p.~36]{KornbergBaker}) show that they have several geometrical properties in common.
Although mathematical investigations of supercoiled DNA shapes have only been conducted from a topological viewpoint, weak elastic forces stemming from the energy $E[Z]$ play an important role, as demonstrated by Monte Carlo computations in the works of Vologodskii and Cozzarelli \cite{VC}.

The shape of the elastica is not solely determined by its topological characteristics, such as knot theory, even if it is in an excited state.
Euler's list of the ground states of the elastica contains fruitful geometry, not only from a topological standpoint, but also from a differential geometric standpoint \cite{PG}.
The degree of bending, as a local property, determines the shapes.
After integrating these local properties, we obtain the topological indices of elasticae.
In the case of Euler's elastica, its topological properties also play a role that should be regarded as supplementary.
For example the difference between the figures on the left in Figure \ref{fg:shape_vsDNA} (a) and (b) comes from the topological indices since $\mathrm{index}(Z) = 0$ in (a) and $\mathrm{index}(Z) = 1$ in (b).

Due to the isometric condition, the geometry of its excited states is even more fruitful, even from a differential geometric perspective.
The excited states of elastica are related to the MKdV flows and are expected to have fascinating shapes. 
However, in stark contrast to the extensive research on their ground states, it is a serious problem that only topological discussions on the shapes of DNAs except numerical computations \cite{VC} have been formulated thus far.

We believe that we are approaching our ultimate goal of providing a mathematical foundation for the excited states of elastica within the framework of affine differential geometry and algebraic function theory.

As mentioned above, comparing cases genera three and five is expected to improve the curve's ability to express complicated shapes by using much higher genus solutions.
To show the mathematical fundamentals of supercoiled DNA, it is crucial to proceed with the study of plane curves governed by the FMKdV and the FGMKdV equations, as Euler developed elliptic function theory, curve geometry, and the variational method as pure mathematics by demonstrating the mathematical foundation of elastic rods \cite{Mat10, M25E}.

\bigskip

Although we studied the case of $[S^g \tH_\RR]^0$ of $S^g \tH_\RR$ and the branch points $(b_i, 0)$ in $S^1$, there are many other configurations of $\gamma_i$ and branch points that generate real $K$ as in \cite{Ma24b,Ma24d} as mentioned in Remark \ref{rmk:4.10}.
In the future, we should classify the moduli of the hyperelliptic solutions of the FGMKdV equation of genus $g$.

Furthermore, when we extend the statistical mechanics of plane curves to that of space curves, we require the similar construction of the hyperelliptic solutions of the nonlinear Schr\"odinger equation of genus $g$, although we partially obtained them in \cite{Ma24NSE}.
The results in this paper will have strong implications for such a generalization.

The results of this paper indicate that further advancement of this study is expected to yield positive results in theory of the integrable systems, the algebraic curves, the affine differential geometry, and the algebraic functions, and even biological science itself.

\setcounter{section}{0}
\renewcommand{\thesection}{\Alph{section}}

\section{Appendix}
Lemmas \ref{4lm:dudphig}, \ref{lm4.1} and \ref{lm:4.7} are key lemmas in this paper but their proofs are slightly complicated.
In this appendix, we show their background and proofs.
Here, we consider the shifted elementary symmetry polynomials.

Let us consider the polynomial ring $R:=\CC[\hfs_1, \ldots, \hfs_g, \hfc_1, \ldots, \hfc_g]$ and its permutation $\tau\in \fS_g$ on their indices, $\CC[\hfs_{\tau(1)},$ $ \ldots, \hfs_{\tau(n)}, \hfc_{\tau(1)}, \ldots, \hfc_{\tau(n)}]$.
$\fS_g$ is the symmetric group of degree $g$.
We have a symmetric ring $SR:=R/\fS_g$, and its homogeneous part $SR_{\ell}$ of degree $\ell$.
Further, for its subring $R^{(j)}:=\CC[\hfs_1, \ldots, $ $ \check{\hfs}_j,\ldots, \hfs_g, \hfc_1, \ldots, \check{\hfc}_j,\ldots, \hfc_g]$ and its permutation $\tau^{(j)}\in \fS_{g-1}$ on their indices $\CC[\hfs_{\tau^{(j)}(1)},$ $\ldots, \check{\hfs}_j,\ldots, \hfs_{\tau^{(j)}(n)}, \hfc_{\tau^{(j)}(1)}, \ldots, \check{\hfc}_j,\ldots, \hfc_{\tau^{(j)}(n)}]$, we consider the symmetric ring $SR^{(j)}:=R^{(j)}/\fS_{g-1}$ whose elements are invariant by $\tau^{(j)}\in \fS_{g-1}$, and its homogeneous part $SR_{\ell}^{(j)}$ of degree $\ell$.
Here check on top of a letter signifies deletion.

We define the shifted elementary symmetric functions $\varepsilon_{ij}$ by the generating polynomial $\displaystyle{\prod_{i=1,\neq j}^{g} ((\hfc_i - \hfs_i) x - 2 \hfs_i)}$
$=\hvarepsilon_{j,g-1} x^{g-1} + \hvarepsilon_{j,g-2} x^{g-2} + \cdots + \hvarepsilon_{j,1} x + \hvarepsilon_{j,0}$.
Then we obviously have the relations, $\hvarepsilon_{j,k}\in SR^{(j)}_{g-1}$, $\hvarepsilon_{j,0}=\displaystyle{\prod_{i=1,\neq j}^{g} 2\hfs_i}$, and $\hvarepsilon_{j,g-1}=\displaystyle{\prod_{i=1,\neq j}^{g} (\hfc_i - \hfs_i)}$. 

The following proposition is essentially the same as the inverse matrix of the Vandermonde matrix in Definition \ref{4df:KdV_def2} and Lemma \ref{4lm:KdV1} (2).
We show the following first proposition.
Here let $\Mat_{R}(\ell)$ denote the set of $\ell \times \ell$ matrices over $R$.
\begin{proposition} For a matrix $\hcM \in \Mat_{R}(n)$ given by
{\small{
$$
\left[\begin{matrix}
(\hfc_1+\hfs_1)^{g-1} &    \cdots &   (\hfc_\ell+\hfs_\ell)^{g-1}  &  \cdots   & (\hfc_g - \hfs_g)^{g-1} \\
2\hfs_1(\hfc_1+\hfs_1)^{g-2} &   \cdots &   2\hfs_\ell(\hfc_\ell+\hfs_\ell)^{g-2} & \cdots & 
2 \hfs_{g}(\hfc_\ell+\hfs_\ell)^{g-2} \\
\vdots   & \ddots &  \vdots & \ddots  &\vdots \\
(2\hfs_1)^\ell(\hfc_1+\hfs_1)^{g-\ell-1}   & \cdots &   (2\hfs_\ell)^\ell(\hfc_\ell+\hfs_\ell)^{g-\ell-1} & \cdots & 
 (2\hfs_{g})^\ell(\hfc_\ell+\hfs_\ell)^{g-\ell-1} \\
\vdots   & \ddots &  \vdots  & \ddots   &\vdots \\
(2\hfs_1)^{g-2}(\hfc_1+\hfs_1)  & \cdots &  
(2\hfs_\ell)^{g-2}(\hfc_\ell+\hfs_\ell) &  \cdots   
& (2\hfs_{g})^{g-2}(\hfc_{g}+\hfs_{g}) \\
(2\hfs_1)^{g-1} &    \cdots &  
(2\hfs_\ell)^{g-1} &  \cdots   & (2\hfs_{g})^{g-1}
\end{matrix}\right],
$$
}}
and for 
$$
\hcE:=\left[\begin{matrix}
\hvarepsilon_{1,0} & \hvarepsilon_{1,1} & \cdots &  \hvarepsilon_{1,\ell}  &  \cdots   & \hvarepsilon_{1,g-2}  & \hvarepsilon_{1,g-1}  \\
\hvarepsilon_{2,0} & \hvarepsilon_{2,1} & \cdots &  \hvarepsilon_{2,\ell}  &  \cdots   & \hvarepsilon_{2,g-2}  & \hvarepsilon_{2,g-1}  \\
\vdots   &\vdots& \ddots & \vdots & \ddots  & \vdots  &\vdots \\
\hvarepsilon_{\ell,0} & \hvarepsilon_{\ell,1} & \cdots &  \hvarepsilon_{\ell,\ell}  &  \cdots   & \hvarepsilon_{\ell,g-2}  & \hvarepsilon_{\ell,g-1}  \\
\vdots   &\vdots& \ddots &  \vdots  & \ddots    &  \vdots &\vdots \\
\hvarepsilon_{g-1,0} & \hvarepsilon_{g-1,1} & \cdots &  \hvarepsilon_{g-1,\ell}  &  \cdots   & \hvarepsilon_{g-1,g-2}  & \hvarepsilon_{g-1,g-1}  \\
\hvarepsilon_{g,0} & \hvarepsilon_{g,1} & \cdots &  \hvarepsilon_{g,\ell}  &  \cdots   & \hvarepsilon_{g,g-2}  & \hvarepsilon_{g,g-1}  \\
\end{matrix}\right]\in \Mat_{R}(n),
$$
we have the determinant $\displaystyle{|\hcM|=2^{g(g-1)/2}\prod_{i<j}(\hfs_i \hfc_j +\hfs_j \hfc_i)}$, and the fact that $\hcE \hcM$ generates a diagonal matrix,
{\small{
$$
\hcE \hcM\! =\!
2^{g(g-1)/2}
\left[\begin{matrix}
\prod_{i\neq 1} (\hfs_i \hfc_1 - \hfs_1 \hfc_i) &  &   &    \\
 & \prod_{i\neq 2} (\hfs_i \hfc_2 - \hfs_2 \hfc_i) &   &    \\
 &                              &\ddots    &    \\
 &     &    & \prod_{i\neq n} (\hfs_i \hfc_g - \hfs_g \hfc_i)
\end{matrix}\right].
$$
}}
\end{proposition}

\begin{proof}
By letting
$$
\cW_0:=\left[\begin{matrix}
1 &   1  &\cdots   & 1\\
2\hfs_1/(\hfc_1+\hfs_1) &   2\hfs_2/(\hfc_2+\hfs_2)  & \cdots &  2\hfs_g/(\hfc_g+\hfs_g) \\
\vdots   &\vdots& \ddots & \vdots \\
(2\hfs_1/(\hfc_1+\hfs_1))^{g-2} &   (2\hfs_2/(\hfc_2+\hfs_2))^{g-2}  & \cdots &  
(2\hfs_g/(\hfc_g+\hfs_g))^{g-2} \\
(2\hfs_1/(\hfc_1+\hfs_1))^{g-1} &   (2\hfs_2/(\hfc_2+\hfs_2))^{g-1}  & \cdots &  
(2\hfs_g/(\hfc_g+\hfs_g))^{g-1} \\
\end{matrix}\right],
$$
$$
\cW_1:=
\left[\begin{matrix}
(\hfc_1+\hfs_1)^{g-1} &     &   &  & \\
 & (\hfc_2+\hfs_2)^{g-1} &      &  & \\
 &           & \ddots & & \\
&           & & (\hfc_{g-1}+\hfs_{g-1})^{g-1}&\\
&           & & &  (\hfc_g+\hfs_g)^{g-1}\\
\end{matrix}\right],
$$
we have $\hcM = \cW_0 \cW_1$.
Here $\cW_0$ is just the Vandermonde matrix and thus, we can compute its determine and inverse matrix as in Lemma \ref{4lm:KdV1}.
Since $\displaystyle{\frac{2\hfs_1}{\hfc_1+\hfs_1}-\frac{2\hfs_2}{\hfc_2+\hfs_2}}$
$=\displaystyle{\frac{2(\hfs_1\hfc_2+\hfs_2\hfc_1)}{(\hfc_1+\hfs_1)(\hfc_2+\hfs_2)}}$,  the determinant of $\hcM$ is evaluated as above.

Let $\displaystyle{\prod_{i=1, i \neq j}^{g-1} ( x - 2 \hfs_i/(\hfc_i+\hfs_i))}$
$= \tvarepsilon_{j, g-1} x^{g-1} + \tvarepsilon_{j, g-2} x^{g-2} + \cdots + \tvarepsilon_{j,1} x + \tvarepsilon_{j,0}$, ($\tvarepsilon_{j, g-1}=1$).
Recalling $\hvarepsilon_{j,g-1}=\displaystyle{\prod_{i=1,\neq j}^{g} (\hfc_i+\hfs_i)}$, we have $\hvarepsilon_{j,g-1}\tvarepsilon_{j,\ell}=\hvarepsilon_{j,\ell}$.
By letting
$$
\hcE_0:=
\left[\begin{matrix}
\hvarepsilon_{1,g-1} &     &   &  & \\
 & \hvarepsilon_{2,g-1} &      &  & \\
 &           & \ddots & & \\
&           & & \hvarepsilon_{g-1,g-1}&\\
&           & & &  \hvarepsilon_{g,g-1}\\
\end{matrix}\right],
$$
$$
\hcE_1:=\left[\begin{matrix}
\tvarepsilon_{1,0} & \tvarepsilon_{1,1} & \cdots  & \tvarepsilon_{1,g-2} & \tvarepsilon_{1,g-1} \\
\tvarepsilon_{2,0} & \tvarepsilon_{2,1} & \cdots  & \tvarepsilon_{2,g-2} & \tvarepsilon_{2,g-1} \\
\vdots   &\vdots& \ddots & \vdots & \vdots \\
\tvarepsilon_{g-1,0} & \tvarepsilon_{g-1,1} & \cdots  & \tvarepsilon_{g-1,g-2} & \tvarepsilon_{g-1,g-1} \\
\tvarepsilon_{g,0} & \tvarepsilon_{g,1} & \cdots  & \tvarepsilon_{g,g-2} & \tvarepsilon_{g,g-1} \\
\end{matrix}\right],
$$
we have $\hcE = \hcE_0 \hcE_1$.
Since $\hcE_1 \cW_0$ provides the diagonal matrix whose $j$-th diagonal part is given by $\displaystyle{\prod_{i\neq j}
\left(\frac{2(\hfs_j\hfc_i+\hfs_i\hfc_j)}{(\hfc_j+\hfs_j)(\hfc_i+\hfs_i)}\right)}$, we prove the equality in the proposition.
\end{proof}

As we introduce the column vectors $\bvarepsilon_i$ in (\ref{Aeq:clmvector}) as $\hcE=(\bvarepsilon_0, \bvarepsilon_1, \ldots, \bvarepsilon_{g-1})$, we implicitly consider the vector space $\langle \bvarepsilon_0, \ldots, \bvarepsilon_{g-1}\rangle_\CC$ and its dimension since we deal with the matrix $\hcM$ and $\hcE$.
In the linear space as an $\RR$ vector space, we can consider appropriate basis for the vector space.
Thus, we propose a natural basis for the vector space without arguments of the linear transformation.

We assume that $\hfc_i \in \RR$ and $\hfs_i \in \ii \RR$.
We can decompose the basis of $\CC^n$ whose elements are real part and pure-imaginary part.
We denote the even and odd degree parts with respect to $\hfs$' of $\hvarepsilon_{j,\ell}$ by $\hvarepsilon_{j,\ell,\even}$ and $\hvarepsilon_{j,\ell,\odd}$;
 $\hvarepsilon_{j,\ell,\even}$ belongs to $\RR$, and $\hvarepsilon_{j,\ell,\odd}$ belongs to $\ii \RR$, and thus we  can alternatively refer to them as $\re \hvarepsilon_{i,j}=\hvarepsilon_{j,\ell,\even}$ and $\ii \im \hvarepsilon_{i,j}=\hvarepsilon_{j,\ell,\odd}$.
The following is a model of Lemma \ref{lm4.1}.
We consider odd $g$ and even $g$ cases respectively.

\begin{proposition}\label{pr:A.2}
For the odd $g$ case, let
{\small{
$$
\hcV:= \left[\begin{matrix}
\hvarepsilon_{1,0,\even} & \hvarepsilon_{1,2,\odd} & \cdots &  \hvarepsilon_{1,2\ell, \odd}  & \hvarepsilon_{1,2\ell, \even}  &  \cdots   \\
\hvarepsilon_{2,0,\even} & \hvarepsilon_{2,2,\odd} & \cdots &  \hvarepsilon_{2,2\ell, \odd}  & \hvarepsilon_{2,2\ell, \even}&\cdots  \\
\vdots   &\vdots& \ddots &  \vdots& \vdots& \ddots   \\
\hvarepsilon_{g-1,0,\even} & \hvarepsilon_{g-1,2,\odd} & \cdots &  \hvarepsilon_{g-1,2\ell, \odd}  & \hvarepsilon_{g-1,2\ell, \even}&\cdots \\
\hvarepsilon_{g,0,\even} & \hvarepsilon_{g,2,\odd} & \cdots &  \hvarepsilon_{g,2\ell, \odd}  & \hvarepsilon_{g,2\ell, \even}&\cdots  \\
\end{matrix}\right.\hskip 0.1\hsize
$$
$$
\hskip 0.5\hsize
 \left.\begin{matrix}
 \hvarepsilon_{1,g-1,\odd}  & \hvarepsilon_{1,g-1,\even}  \\
 \hvarepsilon_{2,g-1,\odd}  & \hvarepsilon_{2,g-1,\even}  \\
 \vdots &\vdots  \\
 \hvarepsilon_{g-1,g-1,\odd}  & \hvarepsilon_{g-1,g-1,\even}  \\
 \hvarepsilon_{g,g-1,\odd}  & \hvarepsilon_{g,g-1,\even}    \\
\end{matrix}\right].
$$
}}
For the even $g$ case, let
{\small{
$$
\hcV:= \left[\begin{matrix}
 \hvarepsilon_{1,1,\odd} & \hvarepsilon_{1,1,\even} & \cdots &  \hvarepsilon_{1,2\ell-1, \odd}  & \hvarepsilon_{1,2\ell-1, \even}  &  \cdots    \\
 \hvarepsilon_{2,1,\odd} & \hvarepsilon_{2,1,\even} & \cdots &  \hvarepsilon_{2,2\ell-1, \odd}  & \hvarepsilon_{2,2\ell-1, \even}&\cdots \\
\vdots   &\vdots& \ddots &  \vdots& \vdots& \ddots   \\
 \hvarepsilon_{g-1,1,\odd} & \hvarepsilon_{g-1,1,\even}  & \cdots &  \hvarepsilon_{g-1,2\ell-1, \odd}  & \hvarepsilon_{g-1,2\ell-1, \even}&\cdots \\
 \hvarepsilon_{g,1,\odd} & \hvarepsilon_{g,1,\even}  & \cdots &  \hvarepsilon_{g,2\ell-1, \odd}  & \hvarepsilon_{g,2\ell-1, \even}&\cdots   \\
\end{matrix}\right.\hskip 0.1\hsize
$$
$$
\hskip 0.5\hsize
\left.\begin{matrix}
    \hvarepsilon_{1,g-1,\odd}  & \hvarepsilon_{1,g-1,\even}  \\
 \hvarepsilon_{2,g-1,\odd}  & \hvarepsilon_{2,g-1,\even}  \\
 \vdots &\vdots  \\
\hvarepsilon_{g-1,g-1,\odd}  & \hvarepsilon_{g-1,g-1,\even}  \\
 \hvarepsilon_{g,g-1,\odd}  & \hvarepsilon_{g,g-1,\even}    \\
\end{matrix}\right].
$$
}}
Then there exists an element $\hcB \in \GL(n, \QQ)$ such that 
\begin{equation}
\hcE = \hcV \hcB, \qquad
\hcB= \begin{pmatrix} \hcB^{[g-3,g-3]} & \hcB^{[g-3],1}& \hcB^{[g-3],2} &0\\
                    0& 1  & -1 & 0 \\
                     0 &0 & -2 & 1 \\
                      0 &0        & 0 & 1 \end{pmatrix},
\label{eq:A.1}
\end{equation}
where $\hcB^{[g-3,g-3]} \in \Mat_\CC( (n -3) \times (g-3))$,and $\hcB^{[g-3],1}, \hcB^{[g-3],2}\in \Mat_\CC( (n -3) \times 1)$.
Particularly,  $\hcB_{g, i}=0$ for $i<n$,  $\hcB_{i, n}=0$ for $i<g-1$, $\hcB_{g-2, i}=0$ for $i<g-2$, $\hcB_{g-1, i}=0$ for $i<g-1$,
$\hcB_{g-2,g-2}=\hcB_{g-2,g-1}$ $=\hcB_{g-1,n}=\hcB_{g,n}=1$, and $\hcB_{g-1,g-1}=-2$.
\end{proposition}

\begin{proof}
Let $c_{g,m}:=\begin{pmatrix} g-1 \\ m \end{pmatrix}$ and $\hc_{g,m}:=(1/2)^m c_{g,m}$ for $n\ge m$ $:=0$ otherwise.
We introduce the symmetric polynomials $\hfe_{j,i}$ as follows:
\begin{eqnarray}
\hfe_{j, 0} &:=& (-)^{g-1} \hvarepsilon_{j,0},  \nonumber \\
\hfe_{j, 1} &:=& (-)^{g-2}\hvarepsilon_{j,1}- \hc_{g-1,1} \hfe_{j,0}, \nonumber \\
\hfe_{j, 2} &:=& (-)^{g-3}\hvarepsilon_{j,2} - \hc_{g-1,2}  \hfe_{j, 0} - \hc_{g-2,1}  \hfe_{j,1},  \nonumber  \\
&\vdots&\nonumber \\
\hfe_{j, \ell} &:=& (-)^{g-\ell-1}\hvarepsilon_{j, \ell} - \hc_{g-1,g-\ell-1} \hfe_{j,0} - \hc_{g-2,g-\ell-2} \hfe_{j,1} -\cdots - \hc_{g-\ell,1} \hfe_{\ell-1},
 \nonumber \\
&\vdots&\nonumber \\
\hfe_{j, g-3} &:=& \hvarepsilon_{j, g-3} - \hc_{g-1,g-3} \hfe_{j,0} -\hc_{g-2,g-4} \hfe_{j, 1} -\cdots -  \hc_{3, 1}\hfe_{j, g-4},  \nonumber \\
\hfe_{j, g-2} &:=& -\hvarepsilon_{j, g-2} -\hc_{g-1,g-2} \hfe_{j,0} -\hc_{g-2,g-3} \hfe_{j, 1} -\cdots -  \hc_{2, 1}\hfe_{j, g-3},  \nonumber \\
\hfe_{j, g-1} &:=& \hvarepsilon_{j, g-1} -
\left(\frac{1}{2}\right)^{g-1}\hfe_{j, 0} -\left(\frac{1}{2}\right)^{g-2}\hfe_{j, 1} -\cdots - \frac{1}{2 }\hfe_{g-2} .
\end{eqnarray}
$\{\hfe_{j,\ell}\}$ is a modified elementary symmetric polynomial of degree $g-1$ such that its degree of $\hfc$ is $\ell$ as in Lemma \ref{lm:A6.3}.
Due to the following lemmas, we have the complete proof of this proposition, i.e., Lemma \ref{lm:A.9} shows (\ref{eq:A.1}).
\end{proof}

We will refer to $\hfe_{j,i}$ as a shifted elementary symmetric polynomial.

\begin{lemma}\label{lm:A6.3}
$\{\hfe_{j,0}, \ldots, \hfe_{0,g-1}\}$ is the basis of the $(g-1)$-th degree homogeneous part of $SR_{g-1}^{(j)}$ as a $\CC$-vector space and the degree of $\hfe_{j,\ell}$ with respect to $\hfc$ is $\ell$.
\end{lemma}

\begin{proof}
As in (\ref{4eq:KdV_def2.1}), $\chi_{i,j}$ is the elementary symmetric polynomial.
Further, it is equal to $\tvarepsilon_{i,j}$, so bipolynomial expansion yields $\hfe$.
\end{proof}

For example, $g=6$ case,
\begin{eqnarray*}
\hfe_{6,0}/2^5&=&\hfs_1 \hfs_2 \hfs_3 \hfs_4 \hfs_5,\\
\hfe_{6,1}/2^4&=& \hfs_2 \hfs_3 \hfs_4 \hfs_5 \hfc_1 + \hfs_1 \hfs_3 \hfs_4 \hfs_5 \hfc_2 + \hfs_1 \hfs_2 \hfs_4 \hfs_5 \hfc_3 
+ \hfs_1 \hfs_2 \hfs_3 \hfs_5 \hfc_4 + \hfs_1 \hfs_2 \hfs_3 \hfs_4 \hfc_5,\\
\hfe_{6,2}/2^3&=&  \hfs_3 \hfs_4 \hfs_5 \hfc_1 \hfc_2 + \hfs_2 \hfs_4 \hfs_5 \hfc_1 \hfc_3 + \hfs_2 \hfs_3 \hfs_5 \hfc_1 \hfc_4
 + \hfs_2 \hfs_3 \hfs_4 \hfc_1 \hfc_5
 + \hfs_1 \hfs_4 \hfs_5 \hfc_2 \hfc_3\\
&&+ \hfs_1 \hfs_3 \hfs_5 \hfc_2 \hfc_4
+ \hfs_1 \hfs_3 \hfs_4 \hfc_2 \hfc_5
+ \hfs_1 \hfs_2 \hfs_5 \hfc_3 \hfc_4
+ \hfs_1 \hfs_2 \hfs_4 \hfc_3 \hfc_5
+ \hfs_1 \hfs_2 \hfs_3 \hfc_4 \hfc_5, \\
\hfe_{6,3}/2^2&=&  \hfs_4 \hfs_5 \hfc_1 \hfc_2 \hfc_3 +  \hfs_3 \hfs_5 \hfc_1 \hfc_2 \hfc_4 +  \hfs_3 \hfs_4 \hfc_1 \hfc_2 \hfc_5+ 
  + \hfs_2 \hfs_5 \hfc_1 \hfc_3 \hfc_4 + \hfs_2 \hfs_4 \hfc_1 \hfc_3 \hfc_5
\\
&&+ \hfs_2 \hfs_3 \hfs_4 \hfc_1 \hfc_5
+ \hfs_1 \hfs_5 \hfc_2 \hfc_3 \hfc_4
+ \hfs_1 \hfs_4 \hfc_2 \hfc_3 \hfc_5
+ \hfs_1 \hfs_3 \hfc_2 \hfc_4 \hfc_4 + \hfs_1 \hfs_2 \hfc_3 \hfc_4 \hfc_5\\
\hfe_{6,4}/2&=&  \hfs_5 \hfc_1 \hfc_2 \hfc_3 \hfc_4 + 
\hfs_4 \hfc_1 \hfc_2 \hfc_3 \hfc_4+
  \hfs_3 \hfc_1 \hfc_2 \hfc_4 \hfc_5+
  \hfs_2 \hfc_1 \hfc_3 \hfc_4 \hfc_5+
  \hfs_1 \hfc_2 \hfc_3 \hfc_4 \hfc_5,\\
\hfe_{6,5}&=&\hfc_1 \hfc_2 \hfc_3 \hfc_4 \hfc_5.\\
\end{eqnarray*}

\begin{lemma}\label{lm:A4}
In terms of the shifted elementary symmetric polynomials $\hfe_{j,i}$, we have
$$
\left[\begin{matrix} \hvarepsilon_{j,0} \\ \hvarepsilon_{j,1}\\ \vdots \\
\hvarepsilon_{j,\ell}\\ \vdots \\ \hvarepsilon_{j,g-2}\\  \hvarepsilon_{j,g-1}
\end{matrix}\right]
=\trp\cA
\left[\begin{matrix} \hfe_{j,0} \\ \hfe_{j,1}\\ \vdots \\
\hfe_{j,\ell}\\ \vdots \\ \hfe_{j,g-2}\\  \hfe_{j,g-1}
\end{matrix}\right], 
$$
\begin{equation} 
\trp\cA:=
\left[\begin{matrix}
(-)^{g-1}        &      &  &          &    &       & & \\
(-)^{g-2}\hc_{g-1,1} & (-)^{g-2}   &  &          &     &       & & \\
\vdots   &\vdots& \ddots &          &     &       & & \\
(-)^{g-\ell-1}\hc_{g-1,\ell} & (-)^{g-\ell-1} \hc_{g-2,\ell-1}    & \cdots & 
(-)^{g-\ell-1}   &     &    &    & \\
\vdots   &\vdots& \ddots &  \vdots  & \ddots    &    &    & \\
\hc_{g-1,g-3} & \hc_{g-2,g-4}    & \cdots & \hc_{\ell,\ell-2}   & \cdots &   1 &      & \\
-\hc_{g-1,g-2} & -\hc_{g-2,g-3}    & \cdots & -\hc_{\ell,\ell-1}   & \cdots &   -\frac{2}{2} & -1      & \\
(\frac{1}{2})^{g-1}  & (\frac{1}{2})^{g-2}  & \cdots & (\frac{1}{2})^{g-\ell} & \cdots  & \frac{1}{4}  & \frac{1}{2} & 1 \\
\end{matrix}\right].
\end{equation} 
\end{lemma}

\begin{proof}
From the definition, it is obvious.
\end{proof}

Since the matrix $\cA$ is a upper triangular matrix with a unit determinant $|\cE|=\pm 1$ up to sign, its inverse is simply obtained.
\begin{lemma}\label{lm:A5}
$\trp(\cA^{-1})=\trp(\cA)$.
\end{lemma}

Lemma \ref{lm:A4} obviously allows us to decompose these $\hvarepsilon$'s into even and odd degree parts with respect to $\hfs$.

\begin{lemma}\label{lm:A6}
The even and odd degree parts with respect to $\hfs$' of $\hvarepsilon_\ell$, $\hvarepsilon_{j,\ell,\even}$ and $\hvarepsilon_{j,\ell,\odd}$ are represented by $\hfe$'s.
\end{lemma}
\begin{proof}
They are obvious.
\end{proof}
As we are concerned with $\hvarepsilon_{g-1 - 2\ell, \even}$ and $\hvarepsilon_{g-1 - 2\ell, \odd}$, they are given as follows.

\begin{enumerate}
\item The case of odd $g$:
{\small{
\begin{eqnarray*}
\hvarepsilon_{j,0,\even}&=& \hfe_{j,0},\\ 
\hvarepsilon_{j,2,\even}&=&  \hc_{g-1,2}\hfe_{j,0} +  \hfe_2,\\ 
&\vdots&\\
\hvarepsilon_{j,2\ell,\even} &=&  \hc_{g-1, 2\ell}\hfe_{j,0} 
+\hc_{g-3, 2\ell-2} \hfe_{j,2} +\cdots + 
 \hc_{g-2\ell+1,2}\hfe_{j,2\ell-2}+ \hfe_{j,2\ell},\\
\hvarepsilon_{j,g-3,\even} &=&  \hc_{g-1, g-3} \hfe_{j,0} + \hc_{g-3, g-5} \hfe_{j,2} +\cdots +  \hc_{4,2}\hfe_{j,g-5}+\hfe_{j,g-3},\\
\hvarepsilon_{j,g-1, \even} &=& \hfe_{j,0}/2^{g-1} + \hfe_{j,2}/2^{g-3}  +\cdots  +\hfe_{j,g-5}/16+\hfe_{j,g-3}/4+ \hfe_{j,g-1}, \\
(\hvarepsilon_{j,0,\odd}&=&  0),\\
\hvarepsilon_{j,2,\odd}&=&  \hc_{g-2,1}\hfe_{j,1},\\
\hvarepsilon_{j,2\ell, \odd} &=&  \hc_{g-2, 2\ell-1}
 \hfe_{j,1} + \hc_{g -4, 2\ell-3} \hfe_{j,3} +\cdots + \hc_{g-2\ell,1}\hfe_{j,2\ell-1},\\
&\vdots&  \nonumber \\
\hvarepsilon_{j,g-3, \odd} &=&  \hc_{g-2,g-4} \hfe_{j,1} + \hc_{g-4, g-6} \hfe_{j,3} +\cdots-\hc_{5,3} \hfe_{j,g-6} +\hc_{3,1}\hfe_{j,g-4},\\
\hvarepsilon_{j,g-1, \odd} &=& \hfe_{j,1}/2^{g-2}+  \hfe_{j,3}/2^{g-4}+\cdots +\hfe_{j,g-6}/32  +\hfe_{j,g-4}/8  +\hfe_{j,g-2}/2.
\end{eqnarray*}
}}

\item The case of even $g$:
{\small{
\begin{eqnarray*}
\hvarepsilon_{j,1,\even}&=& \hfe_{j,1},\\ 
\hvarepsilon_{j,3,\even}&=& \hfe_{j,3} + \hc_{g-2,2}\hfe_{j,1},\\ 
&\vdots&  \nonumber \\
\hvarepsilon_{j,2\ell+1, \even} &=&  \hc_{g-2, 2\ell+1} \hfe_{j,1} 
+ \hc_{g -3, 2\ell-4} \hfe_{j,2\ell-2} +\cdots 
+ \hc_{g-2\ell+1,2}\hfe_{j,2\ell-2}+\hfe_{j,2\ell},\\
&\vdots&  \nonumber \\
\hvarepsilon_{j,g-3, \even} &=& \hc_{g-2,g-4}\hfe_{j,1} +
\hc_{g-4,g-6}\hfe_{j,3} +\cdots  + \hc_{4,2} \hfe_{j,g-5} +\hfe_{j,g-3}, \nonumber \\
\hvarepsilon_{j,g-1, \even} &=& \hfe_{j,1}/2^{g-2} + \hfe_{j,3}/2^{g-4} +\cdots
+ \hfe_{j,g-5}/16 + \hfe_{j,g-3}/4+ \hfe_{j,g-1}, \\
\hvarepsilon_{j,1,\odd}&=& \hc_{g-1,1}\hfe_{j,0},\\
\hvarepsilon_{j,3,\odd}&=&  \hc_{g-1,3}\hfe_{j,0}+ \hc_{g-3,1}\hfe_{j,2},\\
&\vdots&\\
\hvarepsilon_{j,2\ell+1,\odd} &=& 
\hc_{g-1,2\ell+1}\hfe_{j,0} + \hc_{g-3,2\ell-1} \hfe_{j,2} +\cdots 
+ \hc_{g-2\ell+2,3}\hfe_{j,2\ell-2}+\hc_{g-2\ell,1}\hfe_{j,2\ell},\\
&\vdots&\\
\hvarepsilon_{j,g-3,\odd} &=& 
\hc_{g-1, g-3} \hfe_{j,0} + \hc_{g-3, g-5} \hfe_{j,2} +\cdots
+ \hc_{5,3}\hfe_{j,g-6} +\hc_{3,1}\hfe_{j,g-4},\\
\hvarepsilon_{j,g-1, \odd} &=&  \hfe_{j,0}/2^{g-1} +\hfe_{j,2}/2^{g-3} +\cdots +  \hfe_{j,g-6} /32 +\hfe_{j,g-4} /8-\hfe_{j,g-2} /2.
\end{eqnarray*}
}}
\end{enumerate}

Using them, we have the following lemma, which can be regarded as the linear transformation of the vector space spanned by $\fe_{i,j}$ as in (\ref{Aeq:clmvector}).
We introduce the matrix $\hcA$ in the lemma.

\begin{lemma}
For the odd $g$ case
\begin{eqnarray}
\left[\begin{matrix} 
\hvarepsilon_{j,0,\even} \\ 
\hvarepsilon_{j,2,\odd}\\ 
\hvarepsilon_{j,2,\even}\\ 
\vdots \\
\hvarepsilon_{j,2m, \odd}\\
\hvarepsilon_{j,2m, \even}\\
 \vdots \\ 
\hvarepsilon_{j,g-3,\even}\\ 
\hvarepsilon_{j,g-1,\odd}\\  
\hvarepsilon_{j,g-1,\even}
\end{matrix}\right]
=
\trp \hcA
\left[\begin{matrix} \hfe_{j,0} \\ \hfe_{j,1}\\ \vdots \\
\hfe_{j,2\ell}\\
\hfe_{j,2\ell+1}\\ \vdots \\ \hfe_{j,g-2}\\  \hfe_{j,g-1}
\end{matrix}\right],
\end{eqnarray}
where $\trp\hcA$ is given by 
{\small{
\begin{equation}
\left[\begin{matrix}
1        &      &  &          &  &   &       &  & \\
 & \hc_{g-2,1}    &  &        &   &     &       & & \\
\vdots   &\vdots& & \ddots   &   &   &       &  & \\
 &\hc_{g-2,g-2m-2}    & \cdots & 
                      & \hc_{2\ell-1,2(\ell-m)-1}   &    &    &  & \\
 \hc_{g-1,g-2m-1} &    & \cdots & 
\hc_{2\ell,2(\ell-m)}&    &    &      & \\
\vdots   &\vdots& \ddots &  \vdots &  \vdots & \ddots  &      &   & \\
 \hc_{g-1,g-3}& &\cdots & & \hc_{g-2\ell-1,g-2\ell-3}  &\cdots & 1  & &  \\
   & 1/2^{g-2} &\cdots & & 1/2^{g-2\ell-1}  &\cdots & & 1/2  &   \\
1/2^{g-1}  & &\cdots &    1/2^{g-2\ell} &   &\cdots &1/4 &   & 1 \\
\end{matrix}\right].
\end{equation}
}}
For even $g$ case
\begin{eqnarray}
\left[\begin{matrix} 
\hvarepsilon_{j,1,\even}\\
\hvarepsilon_{j,1,\odd}\\
 \vdots \\
\hvarepsilon_{j,2m+1, \even}\\
\hvarepsilon_{j,2m+1, \odd}\\
 \vdots \\ 
\hvarepsilon_{j,g-3,\odd}\\  
\hvarepsilon_{j,g-1,\even}\\  
\hvarepsilon_{j,g-1,\odd}
\end{matrix}\right]
=
\trp \hcA
\left[\begin{matrix} \hfe_{j,0} \\ \hfe_{j,1}\\ \vdots \\
\hfe_{j,2\ell}\\
\hfe_{j,2\ell+1}\\ \vdots \\ \hfe_{j,g-2}\\  \hfe_{j,g-1}
\end{matrix}\right],
\end{eqnarray}
where $\trp\hcA$ is given by 
{\small{
\begin{equation}
\left[\begin{matrix}
\hc_{g-1,1}        &      &  &          &  &   &       &  & \\
 & 1    &           &        &   &     &       &   & \\
\vdots   &\vdots& &    &   &   &       &  & \\
 \hc_{g-1,g-2m-1} &    & \cdots & 
& \hc_{2\ell-1,2(\ell-m)-1}   &   &      &  & \\
 &\hc_{g-2,g-2m-2}    & \cdots & 
            \hc_{2\ell,2(\ell-m)} &    &    &     &  & \\
\vdots   &\vdots& \ddots &  \vdots &  \vdots & \ddots  &        & \\
 \hc_{g-1,g-3}& &\cdots & & \hc_{g-2\ell-1,g-2\ell-3}  &\cdots & 1  & &  \\
   & 1/2^{g-2} &\cdots & & 1 /2^{g-2\ell+1} &\cdots   & & 1/2  &  \\
1/2^{g-1}  & &\cdots &    1/2^{g-2\ell} &   &\cdots  &1/4 &   & 1 \\
\end{matrix}\right].
\end{equation}
}}
Then $|\hcA|\neq 0$ and the coefficients in $\hcA$ are rational numbers.
\end{lemma}

We introduce the following column vectors:
{\small{
\begin{equation}
\bvarepsilon_{\ell}\!:=\!
\left[
\begin{matrix}
\hvarepsilon_{1,\ell}\\
\hvarepsilon_{2,\ell}\\
\vdots\\
\hvarepsilon_{g-1,\ell}\\
\hvarepsilon_{g,\ell}
\end{matrix}\right]\!, \ 
\bvarepsilon_{\ell, \even}\!:=\!
\left[
\begin{matrix}
\hvarepsilon_{1,\ell,\even}\\
\hvarepsilon_{2,\ell,\even}\\
\vdots\\
\hvarepsilon_{g-1,\ell,\even}\\
\hvarepsilon_{g,\ell,\even}
\end{matrix}\right]\!, \ 
\bvarepsilon_{\ell, \odd}\!:=\!
\left[
\begin{matrix}
\hvarepsilon_{1,\ell,\odd}\\
\hvarepsilon_{2,\ell,\odd}\\
\vdots\\
\hvarepsilon_{g-1,\ell,\odd}\\
\hvarepsilon_{g,\ell,\odd}
\end{matrix}\right]\!, \ 
\hfe_{\bullet, \ell}\!:=\!
\left[
\begin{matrix}
\hfe_{1,\ell}\\
\hfe_{2,\ell}\\
\vdots\\
\hfe_{g-1,\ell}\\
\hfe_{g,\ell}
\end{matrix}\right]\!, 
\label{Aeq:clmvector}
\end{equation}
}}
and a matrix $\hfe:=(\hfe_{\bullet, 0}, \hfe_{\bullet, 1}, \ldots, \hfe_{\bullet, g-1})$.

We introduce $\hcV$ for the odd $g$ case, 
$$
\hcV=(\bvarepsilon_{0, \even}, \bvarepsilon_{2, \odd},\bvarepsilon_{2, \even}
,\ldots, 
\bvarepsilon_{g-3, \odd},
\bvarepsilon_{g-3, \even},
\bvarepsilon_{g-1, \odd},
\bvarepsilon_{g-1, \even}),
$$
and $\hcV$ for the even $g$ 
$$
\hcV=(\bvarepsilon_{1, \even}, \bvarepsilon_{1, \odd},\bvarepsilon_{3, \even},
\bvarepsilon_{3, \odd},
\ldots, 
\bvarepsilon_{g-3, \odd},
\bvarepsilon_{g-3, \even},
\bvarepsilon_{g-1, \odd},
\bvarepsilon_{g-1, \even}).
$$

\begin{lemma}\label{lm:A.9}
By recalling $\hcE=(\bvarepsilon_{0}, \bvarepsilon_{1}, \ldots, \bvarepsilon_{g-1})$, we have
$$
\hcE = \hfe \cA, \quad \hcV = \hfe \hcA, \quad
\hcE = \hcV \hcA^{-1}\cA.
$$
$\hcB=\hcA^{-1}\cA$ has the form (\ref{eq:A.1}) in Proposition \ref{pr:A.2}.
\end{lemma}

\begin{proof}
Direct computations provide the following lemma:
$\hcA^{-1}$ has the form,
\begin{equation}
\trp(\hcA)^{-1}:=\left[\begin{matrix}
*      &      &  &          &  &   &       &  & \\
* & *    &           &        &   &     &       &   & \\
\vdots   &\vdots&\ddots &    &   &   &       &  & \\
* &  *   & \cdots & * &    &   &      &  & \\
*  &*   & \cdots &  * & *   &    &     &  & \\
\vdots   &\vdots& \ddots &  \vdots &  \vdots & \ddots  &        & \\
 *& * &\cdots & *& *  &\cdots & 1  & &  \\
 * & * &\cdots & *& * &\cdots    & &2  &  \\
* &*  &\cdots &    * & *   &\cdots  &-1/4 &   & 1 \\
\end{matrix}\right],
\end{equation}
where $*$ means a rational number or zero.
Further from the definition, $\hcB_{i,g}=0$ for $i<g-1$, and $\hcB_{g-1,g}=\hcB_{g,g}=1$.
Thus, we have the form $\hcB$.
\end{proof}

The correspondence between quantities in this appendix and the main text are following:
$$
\hvarepsilon_{k,j}=\ii^j\varepsilon_{k,j}, \quad
\hfe_{k,j}=\ii^j\fe_{k,j}.
$$
Using these, we can derive the results in Section 3 of the main text by reinterpreting the above results.

Then the following is obvious:
\begin{corollary}\label{cor:A.10}
By letting $\tcI:=
\displaystyle{
\left[\begin{matrix}
\ddots &  \\
       & \ii &  &  &  \\
       &     & 1&  & \\
       &     &  &\ii& \\
       &     &  &   & 1
\end{matrix}\right]}$, we have $\tcV = \hcV \tcI^{-1}$, and 
 $\cB = \tcI^{-1} \hcB$ in Lemma \ref{lm4.1}.
\end{corollary}

%
%
%


\section*{Acknowledgments}
The author would like to thank the organizers, especially Professor Yoshinori Machida, and the participants of the conference "Numazu aratame Shizuoka kenkyuukai" in 2025 for allowing him to give a talk and discuss this topic.
This work was partly supported by MEXT Promotion of Distinctive Joint Research Center Program JPMXP0723833165 and Osaka Metropolitan University Strategic Research Promotion Project (Development of International Research Hubs), ``Geometry and Algebra in the Biological Sciences'', September 11--12, 2025"  25-02.
The author acknowledges support from the Grant-in-Aid for Scientific Research (C) of Japan Society for the Promotion of Science, Grant No.21K03289.
He would also like to thank the anonymous referee for the helpful comments and Professor Arkady L. Kholodenko for his long-time encouragement of this project.

\end{document}

%% file: define.tex
\def\be{\mathbf e}

\def\CC{{\mathbb C}}

\def\LL{{\mathbb L}}
\def\RR{{\mathbb R}}
\def\ZZ{{\mathbb Z}}
\def\QQ{{\mathbb Q}}

\def\PP{{\mathbb P}}

\def\be{{\mathbf e}}

\def\bvarepsilon{{\pmb{ \varepsilon}}}

\def\ee{{\mathrm e}}

\def\ii{{\mathfrak i}}

\def\cA{{\mathcal A}}
\def\cB{{\mathcal B}}

\def\cD{{\mathcal D}}
\def\cE{{\mathcal E}}
\def\cF{{\mathcal F}}
\def\cG{{\mathcal G}}

\def\cI{{\mathcal I}}

\def\cK{{\mathcal K}}
\def\cL{{\mathcal L}}
\def\cM{{\mathcal M}}

\def\cU{{\mathcal U}}
\def\cV{{\mathcal V}}
\def\cW{{\mathcal W}}
\def\cX{{\mathcal X}}
\def\cY{{\mathcal Y}}

\def\fS{{\mathfrak S}}

\def\fb{{\mathfrak b}}
\def\fc{{\mathfrak c}}

\def\fe{{\mathfrak e}}

\def\fs{{\mathfrak s}}
\def\ft{{\mathfrak t}}

\def\ri{{\mathrm i}}
\def\rr{{\mathrm r}}

\def\tH{{\widetilde H}}
\def\tK{{\widetilde K}}

\def\tX{{\widetilde X}}

\def\tS{{\widetilde S}}

\def\tX{{\widetilde X}}

\def\tv{{\widetilde v}}

\def\tgamma{{\widetilde \gamma}}

\def\tvarepsilon{{\widetilde \varepsilon}}

\def\tcI{{\widetilde{\mathcal I}}}

\def\tcV{{\widetilde{\mathcal V}}}
\def\tSS{{{\widetilde{\mathbb S}}}}

\def\hc{{\widehat c}}

\def\hB{{\widehat B}}

\def\hJ{{\widehat J}}

\def\hX{{\widehat X}}

\def\hfc{{\widehat{\mathfrak c}}}
\def\hfe{{\widehat{\mathfrak e}}}

\def\hfs{{\widehat{\mathfrak s}}}

\def\hcA{{\widehat {\mathcal A}}}
\def\hcB{{\widehat {\mathcal B}}}

\def\hcE{{\widehat {\mathcal E}}}

\def\hcM{{\widehat {\mathcal M}}}
\def\hcV{{\widehat {\mathcal V}}}

\def\hkappa{{\widehat \kappa}}
\def\hvarepsilon{{\widehat \varepsilon}}

\def\hnu{{\widehat \nu}}

\def\hvarpi{{\widehat \varpi}}

\def\hGamma{{\widehat \Gamma}}

\def\Mat{\mathrm{Mat}}

\def\GL{\mathrm{GL}}

\def\im{\mathfrak{Im}}
\def\re{\mathfrak{Re}}

\def\even{\mathrm{even}}
\def\odd{\mathrm{odd}}

\def\cn{{\mathrm {cn}}}
\def\sn{{\mathrm {sn}}}
\def\dn{{\mathrm {dn}}}

\def\al{{\mathrm {al}}}

\def\qed{\hbox{\vrule height6pt width3pt depth0pt}}
\def\nuI#1{{\nu_{#1}}}

\def\trp{{\, {}^t\negthinspace}}

\def\book#1{\rm{#1}, }
\def\paper#1{\textit{#1}, }
\def\jour#1{\rm{#1}, }
\def\yr#1{({\rm{#1}) }}
\def\vol#1{\textbf{#1}}
\def\pages#1{\rm{#1}}

\def\publaddr#1{\rm{#1}, }
\def\publ#1{\rm{#1}, }
\def\by#1{{\rm{#1}, }}